\documentclass[sigplan,screen]{acmart}
\usepackage[utf8]{inputenc}
\usepackage{xcolor}
\usepackage{graphicx}
\usepackage{grffile}
\usepackage{subcaption}
\usepackage{paralist}
\graphicspath{ {images/} }

\usepackage{float}
\makeatletter

\newfloat{floatbox}{thp}{lob}[section]
\floatname{floatbox}{Text Box}

\usepackage{algorithm}
\usepackage{algpseudocode}
\usepackage{amsthm}

\usepackage{amssymb}
\usepackage{amsmath}
\usepackage[page]{appendix}
\usepackage{wrapfig}
\usepackage{enumitem}
\usepackage{etoolbox}
\makeatletter
\patchcmd{\@verbatim}
  {\verbatim@font}
  {\verbatim@font\small}
  {}{}
\makeatother

\newtheorem{theorem}{Theorem}[section]
\newtheorem{defn}[theorem]{Definition}
\newtheorem{obs}[theorem]{Observation}

\algtext*{EndWhile}%
\algtext*{EndIf}%
\newcommand\CONDITION[2]%
  {\begin{tabular}[t]{@{}l@{}l@{}}
     #1&#2
   \end{tabular}%
  }
\usepackage{xpatch}
\xpatchcmd{\algorithmic}{\setcounter}{\algorithmicfont\setcounter}{}{}
\providecommand{\algorithmicfont}{}
\providecommand{\setalgorithmicfont}[1]{\renewcommand{\algorithmicfont}{#1}}

\setalgorithmicfont{\footnotesize}

\algnewcommand{\IIf}[1]{\State\algorithmicif\ #1\ \algorithmicthen}
\algnewcommand{\EndIIf}{\unskip\ \ \algorithmicend\ \algorithmicif}

\algnewcommand{\ElsIIf}[1]{\unskip\State\algorithmicelse\ \algorithmicif\ #1\ \algorithmicthen}
\algnewcommand{\ElseI}[1]{\unskip\State\algorithmicelse\ }

\algdef{SE}[DOWHILE]{Do}{doWhile}{\algorithmicdo}[1]{\algorithmicwhile\ #1}%

\acmConference[PPoPP'22]{Principles and Practice of Parallel Programming 2022}{April 2--6, 2022}{Seoul, South Korea}
\acmYear{2022}
\acmISBN{978-1-4503-9204-4/22/04} %
\acmDOI{} %
\startPage{1}

\acmYear{2022}\copyrightyear{2022}
\setcopyright{rightsretained}
\acmConference[PPoPP '22]{27th ACM SIGPLAN Symposium on Principles and Practice of Parallel Programming}{April 2--6, 2022}{Seoul, Republic of Korea}
\acmBooktitle{27th ACM SIGPLAN Symposium on Principles and Practice of Parallel Programming (PPoPP '22), April 2--6, 2022, 2022, Seoul, Republic of Korea}
\acmPrice{}
\acmDOI{10.1145/3503221.3508410}
\acmISBN{978-1-4503-9204-4/22/04}

\usepackage{booktabs}   %
\usepackage{subcaption} %
                        
\renewcommand{\paragraph}[1]{\vspace{0.3em}\noindent \textbf{#1}}

\definecolor{pdfbgcolor}{RGB}{180,180,180}

\algnewcommand{\LineComment}[1]{\State \(\triangleright\) #1}

\begin{document}

\title{PathCAS: An Efficient Middle Ground for Concurrent Search Data Structures}  

\author{Trevor Brown}
\affiliation{
  \institution{University of Waterloo}            %
  \country{Canada}                    %
}
\email{me@tbrown.pro}          %

\author{William Sigouin}
\affiliation{
  \institution{University of Waterloo}            %
  \country{Canada}                    %
}
\email{wpsigoui@uwaterloo.ca}          %

\author{Dan Alistarh}
\affiliation{
  \institution{Institute of Science and Technology}           %
  \country{Austria}                   %
}
\email{dan.alistarh@ist.ac.at}         %

\begin{abstract}
    To maximize the performance of concurrent data structures, researchers have often turned to highly complex fine-grained techniques, resulting in efficient and elegant algorithms, which can however be often difficult to understand and prove correct.
    While simpler techniques exist, such as transactional memory, they can have limited performance or portability relative to their fine-grained counterparts. 
    Approaches at both ends of this complexity-performance spectrum have been extensively explored, but relatively less is known about the middle ground: approaches that are willing to sacrifice some performance for simplicity, while remaining competitive with state-of-the-art handcrafted designs. 
    In this paper, we explore this middle ground, and present PathCAS, a primitive that combines  ideas from multi-word CAS (KCAS) and transactional memory approaches, while carefully avoiding overhead.
    We show how PathCAS can be used to implement efficient search data structures relatively simply, using an internal binary search tree as an example, then extending this to an AVL tree.
    Our best implementations outperform many handcrafted search trees: in search-heavy workloads, it rivals the BCCO tree \cite{bronson-tree}, the fastest known concurrent binary tree in terms of search performance~\cite{BSTRoot}.
    Our results suggest that PathCAS can yield concurrent data structures that are relatively easy to build and prove correct, while offering surprisingly high performance.
\end{abstract}



\maketitle

\begin{figure}[t]
 \hspace{-8mm}
    \centering
    \includegraphics[width=0.45\textwidth]{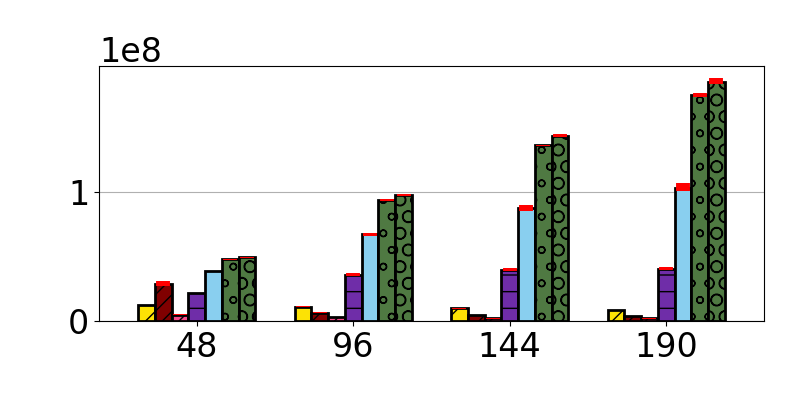}
    
    \vspace{-4mm}
    \includegraphics[width=0.45\textwidth]{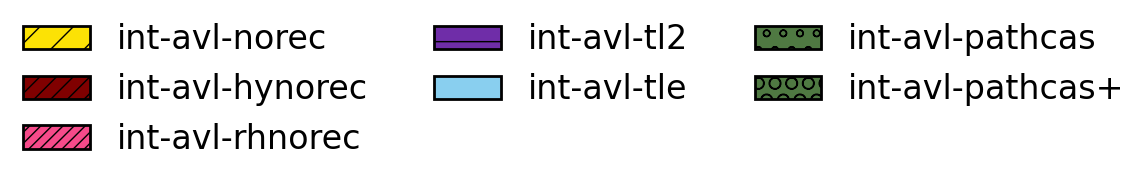}

\centering
\vspace{-5mm}
\caption{AVL trees using PathCAS vs state-of-the-art transactional memory. 10\% updates, 1M key trees. x-axis = number of threads. y-axis = millions of operations per second.}
\label{fig:example_perf}
\vspace{-3mm}
\end{figure}

\section{Introduction}
Significant work has been invested in building scalable concurrent variants of fundamental data structures, and fast implementations are now known for many instances, from search trees to hash tables, or to containers such as queues and stacks~\cite{HSBook}.
On the one hand, designs based on  \emph{fine-grained locking} or \emph{fine-grained lock-free algorithms}, have arguably emerged as the best-performing solution~\cite{BSTRoot}. Yet, such designs tend to have high complexity, and are  notoriously difficult to analyze and prove correct.

On the other hand, significant attention has been given to \emph{general} techniques for obtaining  fast and simple concurrent data structures. 
The classic example is transactional memory (TM)~\cite{HerlihyMoss}, which is now available in software, hardware, and hybrid variants, and allows one to derive concurrent implementations from sequential ones with lower programming effort relative to fine-grained designs. When TM is available, such designs can  provide excellent performance.  

However, TM-based designs still have drawbacks. 
Software TM (STM) provides a hardware-independent alternative to HTM, but can incur higher overheads.
Moreover, although hardware transactional memory (HTM) is technically available on many platforms, via Intel's TSX/TSX-NI, IBM's POWER8/9 TM and ARM's TME \cite{armtle}, it is notably missing from AMD chips (despite the proposal of ASF \cite{ASF} more than a decade ago), and has been disabled in Intel and recent POWER processors~\cite{tsx:disabled,tsx:disabled:by:default, power10}, due to various concerns, chief among which is security.  %

In this context, it is natural to seek a middle ground between the high efficiency, but high complexity, of fine-grained designs, and the relative ease-of-use, but potential pitfalls, of general designs such as the ones based on TM. 
A number of techniques exploring this trade-off have been investigated over the years, often based on either restricted STMs or extended Multi-Compare Multi-Swap (MCMS) \cite{mcms} implementations. However, as we illustrate later in the paper, known instances of these techniques often fail to scale in the context of high-performance data structures.

We revisit this question, and present a primitive for building correct and efficient concurrent search data structures from scratch, called \emph{PathCAS}. 
\textit{PathCAS} combines key ideas from efficient multi-word compare-and-swap (KCAS), and transactional memory, to allow for concurrent data structures which are both efficient, and easier to reason about than hand-crafted data structures using fine-grained primitives.

This mechanism is reminiscent of STM techniques, but it has two \textbf{semantic restrictions}: %
(1) PathCAS \emph{does not guarantee opacity}, and (2) PathCAS has a \textit{bounded read-set}.
These restrictions allow for significant performance benefits, as exemplified in Figure \ref{fig:example_perf}, which we believe are worth the increase in programming complexity compared to TM.\footnote{
Bounding the read-set size is not strictly necessary, but doing so helps us avoid the overheads associated with dynamically sized data structures.
One might imagine, for example, first filling a fixed array with nodes, then appending further nodes to a linked list.
Checking whether a read should be added to the array or the linked list would require an extra branch for each read (possibly on a hot code path).
Even with branch prediction, we have found this overhead to be significant in unbounded transactional memory implememtations.
}
Despite PathCAS being less expressive than TM, it is still sufficient to implement useful data structures.

{We begin by describing the semantics and rationale behind PathCAS in Section 2, and then illustrate how it can be used to implement a simple concurrent unbalanced internal BST, a data structure whose concurrent implementations  warranted publication on their own, e.g.~\cite{howley, dana-tree, Natarajan14, ram-tree, Ellen10, david-tree}, in Section 3. 
To illustrate the difficulty of implementing \emph{correct} variants of these data structures, we note that, during our investigation, we identified a correctness bug in the \emph{lock-based internal BST} of Drachsler et al.~\cite{dana-tree}, and that the publicly-available implementations of the \emph{lock-free internal BSTs} of~\cite{howley} and~\cite{ram-tree} fail experimental validation tests, and still lack complete correctness proofs. (We defer a detailed discussion of these issues to the full version of this paper.) %
In this context, PathCAS provides an implementation that is both efficient and is easy to prove correct.} 

{To further highlight the expressive power of PathCAS, we also show a how a lock-free \emph{balanced} version of this tree can be derived, creating an implementation of a \emph{fully-internal} relaxed AVL tree (Section 4), which performs favorably when compared to state-of-the-art solutions (please see Figure~\ref{fig:example_perf}).}

On the practical side, we present two efficient implementations of PathCAS: an HTM-enabled one which targets Intel systems, and a software-only variant applicable to AMD systems, and use them to empirically validate the above data structure designs. 
We perform an in-depth comparison against previous methods: from HTM- and STM-based  designs, to fine-grained lock-free variants, across both Intel and AMD systems with up to 256 threads. 
We find that PathCAS data structures are highly competitive, across the range from read-heavy to update-heavy workloads. In particular, our unbalanced BST implementation manages to outperform the state-of-the-art in unbalanced BSTs, and our balanced implementation matches the performance of the fastest known balanced BST in read-mostly workloads. %

In sum, PathCAS introduces an additional point in the  trade-off between  expressiveness and programming effort, on the one hand, and efficiency of the resulting data structures, on the other. 
Although PathCAS builds on known techniques, the resulting mechanisms are novel, and our search tree implementations can achieve state-of-the-art results on a variety of workloads.

\section{Related Work} \label{sec:related-work}
The question of identifying synchronization primitives with the ``right'' balance between expressivity and efficiency is as old as the field of concurrency. We describe related work with the goal of situating PathCAS on the spectrum of synchronization techniques that have been developed.

Treiber~\cite{Treiber} gave one of the first illustrations of a non-blocking data structure via CAS, while seminal work by Herlihy~\cite{wait-free-sync} showed that CAS is \emph{universal}. 
Anderson and Moir gave constant-time implementations of \textit{Load Link} (LL) and \textit{Store Conditional} (SC) from CAS \cite{LLSC}, which is more expressive than CAS and helps circumvent ABA problems.
Luchangco et al. ~\cite{kcss} expand on LL/SC by implementing k-compare-single-swap, allowing for the change of a single field to be conditional on \textit{multiple} fields containing their expected values. 
Another extension of LL/SC by Brown et al. \cite{TrevorLLX} introduced \textit{LLX} and \textit{SCX}, which function on \textit{data records}. Data records contain \textit{multiple} related fields which are loaded by LLX, and SCX only succeeds in changing a single memory location if none of the fields loaded by \textit{LLX} have changed since its invocation. LLX/SCX is less expressive than PathCAS, as it atomically: changes \textit{one} field to a \textit{new value}, and \textit{marks} some number of nodes.
Brown et al.~\cite{TrevorLLX} introduced several search tree design based on LLX/SCX, one of which we compare with in the experimental section (ext-chromatic-lf).

Harris et al. \cite{harris-kcas}, introduced a lock-free version of KCAS, which is the basis of our implementation along with optimizations from Arbel-Raviv and Brown \cite{trevor-recycle}.
\textit{KCAS} facilitates atomic multi-word updates, however it does nothing to simplify the arguments around values read but not updated, for example the path followed during a search. 

Timnat et al.~\cite{mcms} introduced a direct evolution of KCAS which they called Multi-Compare Multi-Swap (MCMS)~\cite{mcms}.
Their proposal has similar goals to our work, attempting to achieve a ``middle-ground'' between performance and ease of implementation. 
MCMS attempts to simplify the implementation of concurrent data structures by increasing the expressiveness of KCAS, allowing fields to \textit{compared} without being \textit{swapped}.
When HTM is available, this algorithm attempts to carry out operations in a transaction as a fast-path, similar to our approach. 

One key difference is that, in its slow path (or if HTM is not available), the MCMS algorithm incurs high synchronization costs. 
Specifically, when using MCMS for search data structures, one would need to include the entire search path in the arguments to MCMS.
On the software code path, MCMS would write to every node on the entire search path, including the root, both in updates and in searches.
This aborts all concurrent hardware transactions, likely causing cascading aborts on NUMA systems.
In turn, this induces a global contention bottleneck at the root of the search tree. 
Another key difference is that, on machines without hardware TM, we offer high performance, whereas MCMS essentially becomes the HFP KCAS algorithm.
We present a comparison of MCMS and PathCAS in the experimental section. 

The work of Mahreshanski et al. \cite{lockswapelide} analyzes interplay between HTM and other concurrent designs and how they function together.
This work suggests that HTM does not obviate other designs, but can be used to improve them.
We apply a similar technique in PathCAS, using HTM as a \textit{fast path} for operations, and falling back to our software algorithm when transactions fail a certain number of times.

Guerraoui and Trigonakis present \textit{Optik} \cite{OptikPPoPP}, which is a %
methodology for implementing concurrent data structures using optimistic concurrency and versioned locks.
PathCAS is built using similar low-level techniques, and encapsulates their complexities in an expressive primitive. %

Herlihy and Moss~\cite{HerlihyMoss} proposed HTM to provide flexible hardware support for non-blocking data structures. %
Shavit and Touitou introduced software TM \cite{STM}, as a software-only alternative.
PathCAS has similarities to software transactional memory (STM). STM is easier to use than HTM, as %
STM does not have the same space limitations, but it suffers from various overheads, such as requiring locks per word, dependency on dynamic data structures, and function call overheads on reads and writes. 

Kumar et al.~\cite{HyTM} introduced \textit{Hybrid Transactional Memory} (HyTM), which is a combination of HTM and STM. Our experiments include results from state-of-the-art HyTM algorithms.
To accelerate TM, other work has attempted to break up transactions to avoid overheads. One example is the speculation-friendly tree \cite{synchrobench}, which uses ElasticTM \cite{elastic}. However, %
this tree has relatively poor performance compared to the state of the art, as we show in Section~\ref{eval}. %

\vspace{-1em}
\section{PathCAS}

We provide an overview of the PathCAS primitive, and show how it can be used to implement a concurrent data structure.
Data structures implemented with PathCAS should be \textit{node-based}.
(A data structure can have many different \textit{types} of nodes, and PathCAS can be used to modify any or all of them.)
PathCAS combines ideas from KCAS and version based validation; the rest of this section will provide a description of these components and how they interact. 

\subsection{Background}

In essence, PathCAS is a generalization of KCAS with additional capabilities. 
KCAS is semantically similar to compare-and-swap (CAS), with the key difference that it is able to atomically change \textbf{multiple} addresses (which do \textit{not} have to be contiguous).
KCAS supports a single operation in the form of: \textit{KCAS(addr$_1$, oldValue$_1$, newValue$_1$, ... addr$_k$, oldValue$_k$, newValue$_k$)}.
KCAS does the following atomically: if \textit{addr}$_i$ contains \textit{oldValue}$_i$ for all \textit{i}, the value stored at \textit{addr}$_i$ is changed to \textit{newValue}$_i$ for all \textit{i} and returns true. 
If not, false is returned. 

Our implementation of PathCAS builds on the lock-free KCAS implementation of %
Harris, Fraser and Pratt (HFP)~\cite{harris-kcas}.

\paragraph{Harris, Fraser and Pratt (HFP) algorithm.}
A KCAS %
operation first creates a \textit{KCAS descriptor} $D$ that contains the arguments to the KCAS  as well as a \textit{status} word that indicates whether the KCAS is \textit{InProgress}, \textit{Succeeded} or \textit{Failed}.
It then performs a sequence of atomic double-compare single-swap (DCSS) operations to change all addresses from their respective old values to point to the KCAS descriptor $D$ \textbf{only if} the \textit{status} is still \textit{InProgress}.
(DCSS atomically determines whether \textit{two} (potentially non-contiguous) addresses contain their respective old values, and if so, changes one to a new value and returns \textit{true}. Otherwise it returns \textit{false}.)

If all of the DCSS operations are successful, then the \textit{status} is changed to \textit{Succeeded} and the addresses are all changed from the descriptor pointer to their respective new values using CAS.
Otherwise, the \textit{status} is changed to \textit{Failed} and the addresses are changed back to their old values.
This atomic change to the \textit{status} field  decides the outcome of the operation (and dictates the behaviour of \textit{helper} threads).

Since old values are replaced by descriptor pointers, any time a thread reads an address that could be modified by a KCAS, it must use a special \textit{KCASRead} function that knows how to handle descriptor pointers.
In particular, whenever \textit{KCASRead} encounters a KCAS descriptor, it will \textit{help} the corresponding KCAS operation to complete (by performing the same set of DCSSs and CASs that would be performed by the thread that initially invoked the KCAS).
The use of DCSS avoids ABA problems that could otherwise be introduced by this lock-free helping~\cite{harris-kcas}.
Crucially, DCSS prevents any helper from storing new pointers to the KCAS descriptor once the \textit{status} has become \textit{Succeeded} or \textit{Failed} (preventing helpers from resurrecting completed KCAS operations).

The KCAS descriptor pointer behaves conceptually like a \textbf{lock} that grants exclusive access of a field to a KCAS \textit{operation}, rather than to a particular \textit{thread}.
Once all addresses contain pointers to a KCAS descriptor, they can only be changed in accordance with the corresponding KCAS operation.
If the KCAS descriptor's \textit{status} is \textit{Succeeded}, then all helpers will try to change addresses to their respective new values.
The value contained in a memory address \textit{logically} changes when the \textit{status} of the descriptor changes to \textit{Succeeded}, and a successful KCAS is linearized then.
If the \textit{status} is \textit{Failed}, helpers will try to revert addresses to their old values.
A failed KCAS is linearized when it saw a value that did not match the address' old value.
Changing an address from a descriptor pointer to a value conceptually \textit{unlocks} it.

The authors implemented lock-free DCSS in software, using CAS and \textit{DCSS descriptor} objects to facilitate helping.
There is no need to pierce the atomic DCSS abstraction in our work, except to mention that DCSS descriptors exist.

\subsection{Semantics}

Whereas KCAS takes all of the addresses to be modified (and their respective old and new values) as explicit arguments, the PathCAS interface is closer to transactional memory.
In the following, we say a node $n$ has been \textit{visited} (resp., an address $addr$ has been \textit{added}) if there has been an invocation of \textit{visit(n)} (resp., \textit{add(addr, ...)}) since the last invocation of \textit{start()}.
PathCAS offers operations to:

\begin{itemize}
\item \textit{start()} gathering arguments for a PathCAS,
\item \textit{read(addr)} an address that might be modified via PathCAS,
\item \textit{add(addr, old, new)} an address \textit{addr} to be changed atomically from \textit{old} to \textit{new},
\item \textit{visit(n)} a node \textit{n}, %
and
\item \textit{validate()} to check whether any visited nodes have changed since they were visited. \textit{validate} succeeds and returns \textit{true} \textit{only if} no such change has occurred.
To allow for implementations with diverse progress properties, \textit{validate} can fail and return \textit{false} \textit{spuriously}. %
\item \textit{exec()} performs a \textit{KCAS} according to the arguments passed to invocations of \textit{add} since the last \textit{start}. That is, if all \textit{added} addresses contain their respective old values, then \textit{exec} succeeds, changing all \textit{added} addresses to their respective new values and returning \textit{true}. Otherwise it returns \textit{false}.
\item \textit{vexec()} performs \textit{exec} \textbf{only if} \textit{validate} would succeed.
\end{itemize} 

Behaviour is undefined if an address is \textit{added} multiple times with conflicting old and new values.
If a node is \textit{visited} multiple times (after a particular invocation of \textit{start}), then any changes to it after the \textit{earliest} such visit will cause \textit{exec} and \textit{validate} to return \textit{false}.

Note that \textit{start}, \textit{add} and \textit{exec} can simply be viewed as syntactic sugar for accumulating arguments to a KCAS operation, and \textit{read} is essentially \textit{KCASRead}.
However, \textit{visit} and \textit{vexec} have no direct analogue in KCAS.
To emulate the behaviour of \textit{visit(n)} and \textit{vexec} using KCAS, one could include \textit{every} address in node \textit{n} in the arguments to the KCAS, ``changing'' each address from its \textit{current} value \textit{v} to \textit{v}. %

\subsection{Implementation}

At a high level, the algorithm differs from HFP %
in the following ways: we implement the syntactic sugar described above for incrementally accumulating arguments, and %
we add a new \textbf{validation} phase wherein visited nodes are inspected to determine whether they have changed since they were visited.
Validation affects progress in subtle ways.

\paragraph{Basic operations: start, read, add, visit}
A PathCAS descriptor consists of a \textit{status} field, a sequence of $\langle addr, old, new \rangle$ triples denoted \textit{entries}, and a sequence of $\langle node, ver \rangle$ pairs denoted \textit{path}.
A \textit{start()} operation creates a new descriptor, and we refer to it as the \textit{thread's descriptor} (until \textit{start} is invoked again). %
Similarly to \textit{KCASRead}, a \textit{read(addr)} reads \textit{addr}, and if it sees a pointer to a descriptor, then it \textit{helps} the corresponding \textit{exec} or \textit{vexec} to complete (more about helping below), and repeats these steps.
If it sees a non-descriptor value, that value is returned.
An \textit{add(addr, old, new)} adds a triple to the thread's descriptor's \textit{entries}.

Version numbers are used to track changes to the data structure's nodes.
More specifically, each node is augmented with a version number \textit{ver} that should be incremented every time the node is changed.
The programmer using PathCAS is responsible for ensuring that s/he increments the version numbers of any node $n$ that s/he modifies using PathCAS.
This simply entails reading $n.ver$ and invoking \textit{add(node.ver, v, v+1)} to increment the value $v$ that was read from $n.ver$.
We discuss the motivation behind the decision further in Section~\ref{a:increment}. %

A \textit{visit(n)} operation reads the version $v$ of node \textit{n} using \textit{read}, saves $\langle \&n.ver, v \rangle$ in the thread's descriptor's \textit{path}, and returns $v$.\footnote{
Since the read-set (i.e., path) is bounded, we should mention what happens if the read-set size is exceeded.
In our code, exceeding the read set size triggers an assertion.
In practice, we imagine that the programmer will either over-allocate a large array for visited addresses, or will implement data structures for which a practical height bound is known. %
}
The use of \textit{read} means \textit{visit(n)} will help any \textit{exec} or \textit{vexec} it encounters that is in the process of modifying \textit{n}.

\paragraph{vexec}
An invocation of \textit{vexec} simply passes the thread's descriptor to a subroutine called \textit{help} and returns the result. %

\begin{algorithm}[tb]
\caption{PathCAS::help(desc)}\label{pathcas:help}
\begin{algorithmic}[1]
  \LineComment{Phase 1: ``lock'' addresses for this PathCAS}
  \If{\textit{desc.state $==$ Undecided}}
    \State \textit{newState $=$ Succeeded}
    \For{\textbf{each} \textit{(addr, old, new)} \textbf{in} \textit{desc.entries}}
    \State \textbf{retry\_dcss:}
      \State \textit{valueSeen $=$ DCSS$(\langle$addr, old, desc$\rangle$, $\langle$desc.state, Undecided$\rangle)$}
      \If{\textit{isDescriptor(valueSeen)}}
        \State \textit{help(valueSeen)} \Comment{DCSS failed because of other PathCAS}
        \State \textbf{goto} retry\_dcss \Comment{retry after helping}
      \ElsIf{\textit{valueSeen $\neq$ old}}
        \State \textit{newState $=$ Failed} \Comment{DCSS failed because old $\notin$ addr}
        \State \textbf{break} \Comment{Stop trying to ``lock'' addresses}
      \EndIf
    \EndFor
    \textcolor{red}{\If{\textit{newState $==$ Succeeded} \textbf{and not} \textit{validateDesc(desc)}}
      \State \textit{newState $=$ Failed} \Comment{If validation fails, fail and release ``locks''}
    \EndIf}
    \State \textit{CAS(desc.state, Undecided, newState)}
  \EndIf
  \LineComment{Phase 2: ``unlock'' addresses to new or old values according to state}
  \State \textit{result $=$ (desc.state $==$ Succeeded)}
  \For{\textbf{each} \textit{(addr, old, new)} \textbf{in} \textit{desc.entries}}
    \State \textit{CAS(addr, desc, (result ? new : old))}
  \EndFor
  \State \textbf{return} \textit{result}
\end{algorithmic}
\end{algorithm}

Consider the set $S$ of addresses \textit{added} to the thread's descriptor \textit{desc} (i.e., the addresses that should be \textit{changed} by this PathCAS operation).
An invocation of \textit{help(desc)} first uses DCSS to change all of the addresses in $S$ from their respective old values to point to the PathCAS descriptor.
If any of these DCSSs fail, then all of the addresses that \textit{were} changed to point to the PathCAS descriptor are reverted to their old values using CAS.
Otherwise, now that all addresses are conceptually locked for this PathCAS operation, we can start \textit{validation}.
The two \textbf{red} lines of code in Algorithm~\ref{pathcas:help} are the only changes from the HFP KCAS algorithm.

\begin{algorithm}[tb]
\caption{PathCAS::validateDesc(desc)}\label{pathcas:validate}
\begin{algorithmic}[1]
  \For{\textbf{each} \textit{(node, visitVer)} \textbf{in} \textit{desc.path}}
    \State \textit{currentVer $=$ node.ver}
    \If{\textit{currentVer $=$ desc}}
      \State \textbf{continue} \Comment{``locked'' for \textit{our} PathCAS}
    \EndIf
    \If{\textit{isDescriptor(currentVer)} \textbf{and} \textit{currentVer $\neq$ desc}}
      \State \textbf{return} \textit{false} \Comment{``locked'' for a \textit{different} PathCAS}
    \EndIf
    \If{\textit{currentVer $\neq$ visitVer} \textbf{or} (visitVer \& 1)}
      \State \textbf{return} \textit{false} \Comment{node's version has been changed or marked}
    \EndIf
  \EndFor
  \State \textbf{return} \textit{true}
\end{algorithmic}
\end{algorithm}

\paragraph{Validation}
To perform validation, \textit{help} invokes a subroutine called \textit{validateDesc(desc)}, which rereads the version number of each \textit{visited} node and checks whether it has changed, We discuss the practical considerations of using version numbers, namely wrapping, in the full version of the paper. %

To simplify and optimize the implementation of data structures that \textit{mark} nodes when removing them, we steal the least-significant bit from each node's version number to indicate whether the node has been marked.
(In data structures without marking, this bit is simply not used.)
Validation succeeds only if all visited nodes are unmarked.
In a data structure that marks nodes, success implies that no visited node has been deleted.
Storing the marked bit in the same word as the version number allows a node to be marked as deleted at the same time as its version number is updated with minimal overhead.
(Note that \textit{visit} returns the mark along with the version number.)

If validation succeeds, none of the visited nodes have changed (or been deleted, in a data structure with marking) since they were \textit{visited}. %
In this case, the PathCAS descriptor's \textit{status} field is changed from \textit{InProgress} to \textit{Succeeded} using CAS.
Otherwise, it is changed from \textit{InProgress} to \textit{Failed} using CAS.
Once the \textit{status} field changes to either \textit{Succeeded} or \textit{Failed} it cannot change again.
Finally, if the \textit{status} is \textit{Succeeded}, the addresses in $S$ are changed to their new values using CAS.
Otherwise, their old values are restored via CAS.

\paragraph{Helping}
As in the related KCAS algorithms, since old values are replaced by descriptors, a special \textit{read()} function (analogous to \textit{KCASRead()})\textbf{ must be used to read any fields that can ever be modified by PathCAS}.
A \textit{read()} function that encounters a descriptor pointer will \textit{help} the corresponding PathCAS operation to complete.
Helpers perform the same sequence of steps as the thread that first invoked \textit{vexec} for this PathCAS.
Note that the validation phase will be performed by all helpers, and slow helpers may fail validation even if a fast helper succeeded.
However, a slow helper that fails validation cannot revert addresses to old values, since it will attempt to do so using CAS, and this CAS will fail if the node no longer points to the same PathCAS descriptor (with the same version).
Moreover, as long as a node points to the PathCAS descriptor, it cannot cause validation to fail.

\paragraph{Progress and helping}
At this point, one might wonder why forward progress is guaranteed even though an operation $O$ can invoke \textit{read} and begin helping another operation $O'$ \textit{before} $O$ has finished invoking \textit{add} on all of its fields: Can this cause $O$ and $O'$ to abort each other?
We note that such helping also occurs the lock-free HFP KCAS (in \textit{KCASRead}).
The key observation is: although $O$ can help another operation before $O$ has finished adding its addresses, the operation \textit{being helped} must have already finished adding all of \textit{its} own addresses.
So, such mutual aborts cannot occur.
Progress is discussed in greater detail below.

\paragraph{exec}
The \textit{exec} operation is just a stripped down version of \textit{vexec} that does not perform validation.
It can be implemented simply by removing all pairs for visited nodes from the thread's descriptor before invoking \textit{help}.
The intention of including \textit{exec} in the interface is to allow nodes to be \textit{visited} during a data structure traversal \textit{in case validation will be needed}, and then to decide \textit{not} to validate (reducing overhead) at the end of the traversal.

\paragraph{validate}
The \textit{validate} operation simply passes the thread's descriptor \textit{desc} to \textit{validateDesc} and returns the result.

\subsection{Correctness and Progress}

\paragraph{Correctness}
The \textit{exec} operation is the same as the linearizable lock-free HFP KCAS algorithm, and is linearized in the same way.
In other words, for a successful \textit{exec}, we linearize at the change to the descriptor's \textit{status} field, and for a failed \textit{exec}, we linearize at the read (of an unexpected, non-descriptor value) that caused the failure.
Of course, if a \textit{vexec} is performed but \textit{no nodes were visited}, then \textit{vexec} is the same as \textit{exec}, and is linearized the same way.

The case where a \textit{vexec} is performed after some nodes \textit{were visited} is more nuanced.
Recall that many \textit{helper} threads can participate in a single \textit{vexec} operation $O$ by invoking \textit{help(desc)}, where \textit{desc} is $O$'s descriptor.
The helpers will collaborate to first ``lock'' all addresses, then perform validation, then use CAS to change the descriptor's \textit{status}, then ``unlock'' all addresses.
Only one helper will successfully change the descriptor's \textit{status}, and we call that helper the \textit{decider}.
Once the \textit{status} field is changed, the behaviour of all helpers is dictated by its contents. %
Two cases arise.

If no visited node has its version number changed (or marked) between when it was visited and when the decider rereads its version number during validation, then validation succeeds. %
Given that validation succeeds, \textit{vexec} behaves the same as a successful HFP KCAS (matching the PathCAS semantics).
We linearize just before the decider invokes \textit{validate(desc)}, at which point all added addresses are ``locked'' and no visited node had changed.\footnote{Just as in the HFT KCAS algorithm, at this linearization point, since all added addresses are ``locked,'' and threads cannot read their values without first helping, no thread can read one of the added addresses and obtain an old value. Instead, a new value will be obtained (after helping).} %

However, if some visited node has its version \textit{changed} (or marked) between when it was visited, and when the decider rereads its version number during validation, then \textit{vexec} will behave like a failed HFP KCAS, restoring old values and returning \textit{false} (matching the PathCAS semantics).
We linearize when the value that caused the failure was read.

\paragraph{Progress}
The progress guarantees for PathCAS are subtle.
The \textit{start} and \textit{add} operations are wait-free.
The \textit{visit}, \textit{exec} and \textit{vexec} operations only perform a constant number of steps in addition to an invocation of \textit{help}, but \textit{help} can invoke itself recursively.
The latter is also true in the HFP KCAS algorithm, and it manages to guarantee \textit{lock-free} progress with an assumption that the addresses passed to KCAS are \textit{sorted}.
If we make the same assumption, then it is possible to argue that \textit{visit}, \textit{exec} and \textit{vexec} operations are lock-free. %

However, lock-freedom only guarantees that infinitely many operations will \textit{terminate} in an infinite execution---not that any of them will \textit{succeed}.
To see why this could be a problem, consider a data structure with two nodes, $A$ and $B$.
Suppose thread $t_1$ \textit{visits} $A$ and \textit{adds} $B$ (to change $B$'s value), and thread $t_2$ \textit{visits} $B$ and \textit{adds} $A$ (to change $A$'s value).
If $t_1$ and $t_2$ both ``lock'' their respective \textit{added} nodes, then both perform validation, both will fail validation 
and ``unlock,'' \textit{terminating}, and hence satisfying lock-freedom, but perhaps preventing the data structure \textit{using} PathCAS from making progress. %
The problem here is that both \textit{vexec} operations can fail \textbf{spuriously}, even though the non-descriptor values that are \textit{semantically} contained in $A$ and $B$ have not changed.

\subsection{Avoiding spurious failures}

It is impossible to avoid \textit{vexec} failures altogether.
One can always invoke \textit{vexec} after \textit{adding} addresses with unreasonable old values that they have \textit{never contained}.
However, the above implementation allows every \textit{vexec} to fail \textit{spuriously}, simply because a \textit{visited} node contained a descriptor pointer.
To be able to implement lock-free data structures using PathCAS, we need to change this. %
Without loss of generality, in the following, we focus on \textit{vexec} operations (since \textit{exec} operations are just a special case).

We say a thread $t$ invokes a \textbf{reasonable} \textit{add(addr, old, new)} if the \textit{old} value was read from \textit{addr} at some point \textit{since} the last invocation $S_t$ of \textit{start} by $t$.
If a thread invokes \textit{start} followed by a sequence of \textit{reasonable} \textit{add} operations, followed by a \textit{vexec}, then we call the \textit{vexec} \textit{reasonable}.
With a small modification to \textit{vexec}, we can guarantee the following. %

\noindent\textbf{Property P1.}
If each thread $t$ invokes only \textit{reasonable} \textit{vexec} operations, then whenever a \textit{vexec} $V_t$ fails, another \textit{vexec} has \textit{succeeded} since $V_t$'s \textit{start} operation, $S_t$.

\paragraph{Strong vexec}
In the implementation described previously, a \textit{vexec} fails \textit{validation} simply because it sees a descriptor, and ``unlocks'' all of its nodes.
Rather than failing spuriously, \textit{vexec} can fall back to a slower lock-free code path on which it creates a \textit{new copy} of its descriptor with slightly different contents.
This new descriptor contains all of the \textit{added} fields of the old one, but crucially, all of the \textit{visited} $\langle node, ver \rangle$ \textit{pairs} in the old descriptor are converted into \textit{added} $\langle node.ver, ver, ver \rangle$ triples.
These triples are then \textit{sorted}.
Finally, this new descriptor is passed as the argument to an \textit{exec} operation, which will effectively ``lock'' all of the visited nodes' version numbers rather than simply validating them.

In practice, to reduce overhead, before switching to this slow path, \textit{vexec} can repeatedly try again (a bounded number of times) using an exact copy of its descriptor and performing validation as usual.
Since the slow path has high overhead, the number of retries can be tuned to avoid invoking the slow path except where it is really necessary. One can also try contention management strategies such as bounded exponential backoff to further reduce slow path usage.

Note that the choice of \textit{vexec} or strong \textit{vexec} \textbf{does not affect performance} in our experiments, as spurious failures are sufficiently infrequent that there is no need to switch to the slow path.

\paragraph{How strong vexec helps}
Strong \textit{vexec} is not vulnerable to the progress problem described above.
Suppose thread $t_1$ \textit{visits} $A$ and \textit{adds} $B$, and thread $t_2$ \textit{visits} $B$ and \textit{adds} $A$.
If $t_1$ and $t_2$ both ``lock'' their respective \textit{added} nodes, then both perform validation, both will fail validation and ``unlock,'' \textit{but they will not terminate}.
Rather, they will retry.
They can retry only a bounded number of times before executing the slow path.
Once both are executing the slow path, they will each try to lock $A$ \textit{then} $B$ (because of address sorting), and one of them will succeed.

Let us sketch why P1 is satisfied.
A \textit{reasonable} \textit{vexec} $V_t$ does not fail when it fails validation.
Rather, it fails only if (a) one of its \textit{reasonable} \textit{added} addresses contains an unexpected non-descriptor value, or (b) one of its visited nodes' version numbers has been incremented.
(In both cases, $V_t$ might help one or more other \textit{vexec} operations to complete before it can read a non-descriptor value.)
In case (a), since the \textit{added} address contained its reasonable old value at some time since $S_t$, and it can be changed to a different non-descriptor value only by a successful \textit{vexec}, P1 holds.
Similarly, in case (b), the visited node's version number was read since $S_t$, and a node's version number is incremented only when the node is changed by a successful \textit{vexec}, so P1 holds.

\subsection{Optimizing descriptor management}

Arbel-Raviv and Brown~\cite{trevor-recycle} showed how to transform the HFP algorithm to \textit{eliminate} the need to allocate and free descriptors for DCSS and KCAS.
The same transformation can be applied to PathCAS, allowing each thread to reuse one PathCAS descriptor (and \textit{we do this} in our experiments).
The transformation in~\cite{trevor-recycle} is straightforward and mechanical, but it makes the pseudocode much more difficult to read, so we presented \textit{pre-transformation} code.
Similarly, to avoid complicating the code, we treated DCSS as an atomic primitive.
(In reality it is implemented in software as in~\cite{trevor-recycle}.)

\subsection{Optimizing with hardware TM}
On systems with support for hardware TM, the PathCAS algorithm above can be used as a \textit{fallback} code path, and a faster hardware TM based algorithm can be used as a \textit{fast path}.
In other words, we can use a hardware transaction to perform \textit{vexec/exec} atomically without the overhead of synchronizing via DCSS and CAS. %

Our hardware TM based fast path is simply obtained by taking the software algorithm above, wrapping it in a transaction, and then performing a sequence of sequential optimizations (which do not affect correctness because of the atomicity of hardware transactions).

\subsection{Comparison to transactional memory}
PathCAS is most similar to a lock-free, non-opaque, bounded, object-based TM that is compiled directly into the data structure (rather than being compiled as a library).
Such a highly restricted TM implementation could avoid many of the same traditional TM overheads that we also avoid: incremental validation to guarantee opacity, locks per word (instead of version numbers per node), dynamic data structures such as hash tables with intrusive lists (instead of a simple array for our visited nodes), and function call overhead for reads and writes.
However, such a TM would be no easier to use than PathCAS, and to our knowledge no such TM exists.
Moreover, it would be a substantial undertaking to design an efficient TM with these properties.

\subsection{Design Decision: Manual Version Numbers} \label{a:increment}
We contemplated building the incrementing of version numbers into the abstraction, so that it would be automatic.
However, we decided that requiring addresses passed to \textit{add} to be fields of nodes might be overly restrictive.
We do not want to rule out applications wherein PathCAS could be used to atomically validate a set of nodes, and also modify arbitrary fields that do not belong to a data structure node (such as a size variable).
Therefore, we only require nodes that are \textbf{passed to} \textit{visit} to have version numbers to track changes, and leave it to the programmer to manage them.
Note that our interface supports debugging mechanisms to catch errors in managing version numbers.\footnote{%
For example, \textit{visit} can save the address ranges of all visited nodes, and \textit{exec} can then check for intersections between the \textit{visited} nodes and \textit{added} addresses that do not have a corresponding \textit{node.ver} increment. This introduces overhead, but can be enabled only when debugging.}
In an application where it is acceptable to restrict PathCAS so that it only accesses \textit{nodes}, one could easily change \textit{add()} to also take a node pointer in addition to the field pointer, and automate version increments. %

\section{Application: Lock-free Internal BST} \label{section:bst}

In this section we provide a concrete example of how to create a data structure using PathCAS, namely, a concurrent set implemented as an \textit{internal} binary search tree. 

\paragraph{Operations}
The tree supports the following operations.
\textit{contains(key)} returns \textit{true} if \textit{key} is in the tree, and \textit{false} otherwise.
\textit{insert(key, val)} returns \textit{false} if \textit{key} is in the tree.
Otherwise, it inserts \textit{key} and \textit{value} returns \textit{true}.
\textit{delete(key)} returns \textit{false} if \textit{key} is not in the tree.
Otherwise, it deletes \textit{key} and its associated \textit{value} and returns \textit{true}.

\paragraph{Data structures}
Tree nodes have fields for \textit{left} and \textit{right} children, a \textit{key}, a \textit{value}, and a version number \textit{ver} as required by PathCAS.

To avoid special cases, the tree always contains two \textit{sentinel} nodes with keys $-\infty$ and $+\infty$.
Consequently, every node with key $k \in (-\infty, +\infty)$ always has both \textit{predecessor} and \textit{successor} nodes.
The sentinel with key $+\infty$, which we call the \textit{maxRoot}, is the root of the entire tree.
The sentinel with key $-\infty$, which we call the \textit{minRoot}, is the left child of \textit{maxRoot}.
No field of \textit{maxRoot} is ever changed.
All keys in $(-\infty, +\infty)$ are always found in the \textit{right subtree} of \textit{minRoot}.

\paragraph{Implicit read()}
Our pseudocode exemplifies a feature of our PathCAS implementation in C++: implicit \textit{read} invocations.
Whereas \textit{KCASRead()} calls must be explicitly added by the programmer, in C++, templates and operator overloading can be used to invoke PathCAS \textit{read()} calls automatically.
\footnote{The programmer need only annotate the \textit{types} of fields that can be modified by KCAS in the data structure \textit{node} type definition, by \textit{wrapping} each field's type in a special PathCAS template type.
For example, \texttt{int key} becomes \texttt{casword<int> key}.
This requires very little effort, and can even help us catch some types of programmer errors, such as unsafe writes to fields that can be modified with PathCAS.} 
\textit{Thus, in our BST pseudocode we do not explicitly invoke the PathCAS read function, but the reader should note: any field that is ever modified by PathCAS is accessed using read.}
\begin{algorithm}[t]
    \caption{BST::search(\textit{key})}\label{search}
    \begin{algorithmic}[1]
   \While{\textit{true}}
    \State \textit{parent $=$ maxRoot}
    \State \textit{parentVer $=$ visit(parent)}
    \State \textit{curr $=$ minRoot}
    \State \textit{currVer $=$ visit(curr)}
    \While{\textit{curr $\neq$ NIL}}
       \State \textit{currKey $=$ curr.key}
       \If{\textit{key $==$ currKey}}
         \State \textbf{return} \textit{$\langle$true, curr, currVer, parent, parentVer$\rangle$}
       \EndIf
       \State \textit{parent $=$ curr}
       \If{\textit{key $>$ currKey}}
         \textit{curr $=$ curr.right}
       \Else{ \textit{key $<$ currKey}}
         \textit{curr $=$ curr.left}
       \EndIf
       \State \textit{parentVer $=$ currVer}
       \State \textit{currVer $=$ visit(curr)}
    \EndWhile
    \State \textbf{return} \textit{$\langle$false, curr, currVer, parent, parentVer$\rangle$}
    \EndWhile
    \end{algorithmic}
\end{algorithm}

\paragraph{Search}
The \textit{search(key)} procedure (Algorithm \ref{search}) is invoked by \textit{contains}, \textit{insert} and \textit{delete}.
It performs a traditional BST search until it encounters a \textit{NIL} pointer, or finds a node containing \textit{key}.
\textit{search} returns a tuple of five items with types: \textit{$\langle$Boolean, node, version, node, version$\rangle$}.
If the key \textit{key} is found, this tuple contains \textit{true}, followed by the node that contains \textit{key} and its version number (observed during \textit{search}), followed by its parent and its version number. %
If \textit{key} was not found, then search returns \textit{false}, followed by the final node it encountered (before seeing a NIL pointer) and the version number of that node.
The remaining two fields are ignored in this case. %
We return these two nodes (and their versions) to be used by \textit{insert} and \textit{delete}. %
The key difference between this search and a \textit{sequential} BST search is that each node is passed to an invocation of \textit{visit}.

\paragraph{Contains} %
The \textit{contains(key)} operation %
invokes \textit{search(key)}, followed by \textit{validate()}.
If validation succeeds, then the entire search was effectively \textit{atomic} (since the entire path was \textit{visited}), so we return \textit{true} if \textit{search} found \textit{key} and \textit{false} otherwise.
If validation fails, we retry the \textit{contains} from scratch.

\begin{figure}
\vspace{-2mm}
\includegraphics[width=0.8\linewidth]{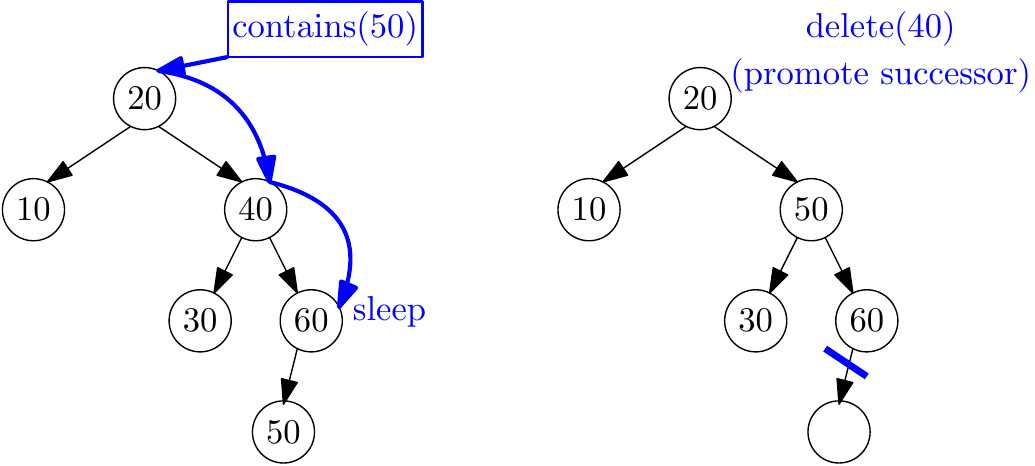}
\vspace{-4mm}
\caption{Error in \textit{contains} operation without validation}
\label{fig:contains-error-without-validation}
\end{figure}

One might wonder why \textit{contains} performs validation.
Figure~\ref{fig:contains-error-without-validation} depicts an error that can occur without validation in an internal BST with atomic updates.
In that execution, \textit{contains(50)} reaches node 60 then sleeps.
Then, \textit{delete(40)} atomically deletes 40, promoting the successor key 50 in its place.
When \textit{contains(50)} wakes up and continues its search, it will conclude that 50 is not in the tree and return \textit{false}.
This is incorrect, as 50 has been in the tree throughout the entirety of the \textit{contains(50)} operation.
Validation would catch this error, since \textit{delete(40)} changes a node that \textit{contains(50)} has already visited.

Validation makes arguing correctness trivial (validated searches are \textit{atomic}), and only incurs a small amount of overhead.
We discuss an optimized implementation that performs \textit{less} validation in Section~\ref{sec:bst-optimizations}.

\paragraph{Insert}
The \textit{insert(key, val)} operation (Algorithm \ref{insert}) first invokes \textit{search(key)} to determine whether \textit{key} is already in the tree, and to locate the \textit{parent} whose child pointer should be changed to insert a new node containing \textit{key} (if \textit{key} is not already in the tree). %

If the \textit{search} finds \textit{key}, then \textit{validate} is invoked to determine whether any the nodes visited by \textit{search} changed since they were visited.
If validation succeeds, it establishes a time $t$ during the \textit{insert} operation when \textit{key} was already in the tree, so \textit{false} is returned (the \textit{insert} is linearized at time $t$).

\begin{algorithm}[t]
    \caption{BST::insert(\textit{key, val})}\label{insert}
    \begin{algorithmic}[1]
        \While{\textit{true}}
        \State \textit{start()}
        \State\textit{$\langle$foundKey, -, -, parent, parentVer$\rangle =$ search(key)}
        \IIf {\textit{foundKey} \textbf{and} \textit{validate()}} \textbf{return} \textit{false};
        \State \textit{newLeaf $=$ createNode(key, val)} \label{line:insert:node-creation}
        \State \textit{parentKey $=$ parent.key}
        \State \textit{ptrToChange $=$ (key $<$ parentKey) ? \&parent.left : \&parent.right}
    	\State \textit{add(ptrToChange, NIL, newLeaf)}
        \State \textit{add(\&parent.ver, parentVer, parentVer + 2)} \Comment{Increment version}
        \If {\textit{vexec()}} \label{line:insert:kcas}
        	 \textbf{return} \textit{true} \label{line:insert:return-true}
    	 \EndIf
        \EndWhile
    \end{algorithmic}
\end{algorithm}

If \textit{search()} does not find \textit{key}, then a new node containing \textit{key} and \textit{val} is created, and \textit{add} is invoked so that the appropriate child pointer of \textit{parent} will be changed (by a subsequent \textit{vexec}) to point to this new node.
Since we are trying to change \textit{parent}, \textit{add} is invoked to cause the parent's version number to be incremented. %

Finally, \textit{vexec} is invoked to (attempt to) atomically change the \textit{added} addresses \textbf{only if} none of the visited nodes have changed since they were visited.
If it succeeds, we linearize at the \textit{vexec}.
Otherwise, we retry the \textit{insert} from scratch.

\paragraph{Delete}
The \textit{delete(key)} operation (Algorithm \ref{delete}) first searches for the key to be deleted, similar to \textit{insert}. 

If \textit{search} does not find \textit{key}, the path followed in \textit{search} is validated.
If this validation is successful, \textit{false} can be returned as a time has been established when the entire path traversed by \textit{search}, which did \textit{not} contain \textit{key}, was atomically contained in the tree.
Thus, the tree did not contain \textit{key} at some time during the \textit{delete}, and we can linearize at that time.
If validation fails, \textit{delete} is retried from scratch.

\begin{algorithm}[tb]
\caption{BST::getSuccessor(\textit{start, startVer})}\label{successor}
\begin{algorithmic}[1]
  \State \textit{succP $=$ start}
  \State \textit{succPVer $=$ startVer}
  \State \textit{succ $=$ startNode.right}
  \State \textit{succVer $=$ visit(succ)}
  \While {\textit{true}}
    \State \textit{next $=$ succ.left}
    \If {next $==$ NIL}
      \textit{\textbf{return} \textit{$\langle$succ, succVer, succP, succPVer$\rangle$}}
    \EndIf
    \State \textit{succP $=$ succ}
    \State \textit{succPVer $=$ succVer}
    \State \textit{succ = next}
    \State \textit{succVer = visit(next)}
  \EndWhile
\end{algorithmic}
\end{algorithm}

If \textit{search} finds \textit{key}, the node \textit{curr} containing \textit{key} is returned, along with its \textit{parent}.
\textit{delete()} then checks whether these nodes are marked (and hence deleted already).
If either is marked, the \textit{delete} is retried from scratch.

Next, \textit{delete} reads \textit{curr.left} and \textit{curr.right} to determine how many children \textit{curr} has.
It does not matter if the number of children is counted incorrectly, for example, because \textit{curr.left} is changed between these two reads.
If \textit{curr} changes, our subsequent \textit{vexec} will fail and the \textit{delete} will retry.
We would not have to consider this possibility at all if we were using opaque transactional memory instead of PathCAS, but we would argue that this reasoning is not onerous to avoid the overheads that come along with opacity.

As in a \textit{sequential} internal BST, three cases arise.
In each case, we use \textit{vexec} to perform the sequential update \textit{atomically}.
If \textit{vexec} succeeds, \textit{delete} returns \textit{true}, and we linearize at the \textit{vexec}.
If \textit{vexec} fails, the \textit{delete} is retried from scratch.

\begin{algorithm}[t]
    \caption{BST::delete(\textit{key})}\label{delete}
    \begin{algorithmic}[1]
      \While{true}
        \State \textit{start()}
        \State \textit{$\langle$foundKey, curr, currVer, parent, parentVer$\rangle =$ search(key)}
        \If {\textbf{not} \textit{foundKey}}
            \IIf{validate()} \textbf{return} \textit{false};
            \ElseI\ \textbf{continue}
        \EndIf
        \IIf {\textit{currVer \& 1} \textbf{or} \textit{parentVer \& 1}} \textbf{continue} \Comment{if \textbf{marked}} \label{delete:markedcheck}
        \State \textit{currLeft $=$ curr.left}
        \State \textit{currRight $=$ curr.right}
        \If {\textit{currLeft $==$ NIL} \textbf{and} \textit{currRight} $==$ NIL} \Comment{\textbf{Leaf} deletion}
          \State \textit{ptrToChange $=$ (curr $==$ parent.left) ? \&parent.left : \&parent.right}
          \State \textit{add(ptrToChange, curr, NIL)}
          \State \textit{add(\&parent.ver, parentVer, parentVer + 2)}
          \State \textit{add(\&curr.ver, currVer, currVer + 1)} \Comment{\textbf{mark} curr}
          \If {\textit{vexec()}} \textbf{return} \textit{true} \EndIf
        \ElsIf {\textit{currLeft $==$ NIL \textbf{or} currRight $==$ NIL}} \Comment{\textbf{One child}}
          \State \textit{childToKeep $=$ (currLeft $==$ NIL) ? currRight : currLeft}
          \State \textit{ptrToChange $=$ (curr $==$ parent.left) ? \&parent.left : \&parent.right}
          \State \textit{add(ptrToChange, curr, childToKeep)}
          \State \textit{add(\&parent.ver, parentVer, parentVer + 2)}
          \State \textit{add(\&curr.ver, currVer, currVer + 1)} \Comment{\textbf{mark} curr}
          \If {\textit{vexec()}} \textbf{return} \textit{true} \EndIf
        \Else \Comment{\textbf{Two child} deletion}
          \State \textit{$\langle$succ, succVer, succP, succPVer$\rangle =$ getSuccessor(curr, currVer)}
          \IIf {\textit{succ $==$ NIL} \textbf{or} (\textit{succVer \& 1}) \textbf{or} (\textit{succPVer \& 1})}\ \textbf{continue}
          \State \textit{succR $=$ succ.right} \Comment{succ has at most one child}
          \If{\textit{succR $\neq$ NIL}} \Comment{if it does}
            \State \textit{succRVer $=$ visit(succR)}
            \IIf {\textit{succRVer \& 1}} \textbf{continue}
          \EndIf
          \State \textit{ptrToChange $=$ (succP.right $==$ succ) ? \&succP.right : \&succP.left}
          \State \textit{add(ptrToChange, succ, succR)} \label{line:deleteSimpleStart}
          \State \textit{add(\&curr.val, curr.val, succ.val)}
          \State \textit{add(\&curr.key, key, succ.key)}
          \State \textit{add(\&succ.ver, succVer, succVer + 1)} \Comment{\textbf{mark} succ}
          \State \textit{add(\&succP.ver, succPVer, succPVer + 2)}  \label{line:deleteSimpleEnd}
          \If {\textit{succP $\neq$ curr}} %
            \textit{add(\&curr.ver, currVer, currVer + 2)}
          \EndIf
          \If {\textit{vexec()}} \label{line:deleteTwo:kcas}
             \textbf{return} \textit{true} \label{line:delete:two-return-success}
          \EndIf
        \EndIf
      \EndWhile
    \end{algorithmic}
\end{algorithm}

\paragraph{Leaf deletion}
If \textit{curr} has no children, \textit{vexec} is invoked to unlink and mark it, and to increment the \textit{parent.ver}.

\paragraph{One child deletion}
If \textit{curr} only has a single child, \textit{vexec} is invoked to replace \textit{curr} by its child, marking \textit{curr} and incrementing \textit{parent.var}.

\paragraph{Two child deletion}
If \textit{curr} has two children, we will try to replace its key and value with those of its \textit{successor} \textit{succ}, then delete the node \textit{succ} (exactly as one does in a sequential internal BST).
We first locate the successor \textit{succ} and its parent \textit{succP} using \textit{getSuccessor}, which \textit{visits} each node it traverses and returns the version numbers it saw.
Note that the successor cannot have a left child (or else it is not the successor).
So, \textit{succ} has at most one child, which means it can be deleted using one of the previous two cases.

If it has a child, \textit{succR}, then we change the appropriate pointer in the parent \textit{succP} to \textit{succR}.
If it has no children, \textit{succR} is NIL, so changing the appropriate pointer in \textit{succP} to \textit{succR} simply unlinks \textit{succR}.
We \textit{mark} \textit{succR} since it is being removed, and \textit{increment} the versions of \textit{curr} and \textit{succP} since they are being changed.
(If the successor happens to be the right child of \textit{curr}, then \textit{succP} and \textit{curr} are the same node, so we only need to increment one of \textit{succP} and \textit{curr}.)
Note that the success of \textit{vexec} implies that \textit{succ} actually \textit{is} the successor of \textit{curr} when the \textit{delete} is linearized.

\subsection{Optimizing to reduce validation}\label{sec:bst-optimizations} 

In \textit{contains}, if \textit{foundKey} is \textit{true}, then it is unnecessary to \textit{validate}, because the key can only be found if it was actually in the tree at some time during the contains, and we can linearize the contains at that time.
(If a node was unlinked \textit{before} \textit{contains} \textit{began}, then \textit{contains} cannot reach it.)

Readers familiar with the lazy linked list~\cite{LazyList} might wonder why we do not have to consider the case where a node is \textit{marked} as \textit{logically deleted} \textbf{before} our \textit{contains} began, but not \textit{unlinked} until later.
Note that, unlike typical concurrent data structures with marking, where nodes are marked \textit{before} they are removed, in our tree nodes are unlinked and marked in the \textit{same atomic PathCAS} operation.
Thus, reachability in our tree is \textit{equivalent} to being unmarked (and hence logically contained in the tree).

This optimization can also be applied to \textit{insert} and \textit{delete} operations that return \textit{false}.
Additionally note that validation is not required in the leaf deletion and one child deletion cases, so \textit{exec} can be used instead of \textit{vexec}.
A detailed explanation is deferred to the full version of this paper. %

\vspace{-1mm}

\begin{figure*}[th]
\vspace{-3mm}
{
\newcommand{\plotwidth}{0.32\linewidth}
\newcommand{\legendwidth}{0.6\linewidth}
\begin{tabular}{p{\plotwidth}p{\plotwidth}p{\plotwidth}}
\centering\noindent\textbf{1\% updates} & \centering\noindent\textbf{10\% updates} & \centering\noindent\textbf{100\% updates}
\end{tabular}

\vspace{-0.5mm}
\rotatebox{90}{\textbf{Unbalanced BSTs}}
\includegraphics[width=\plotwidth]{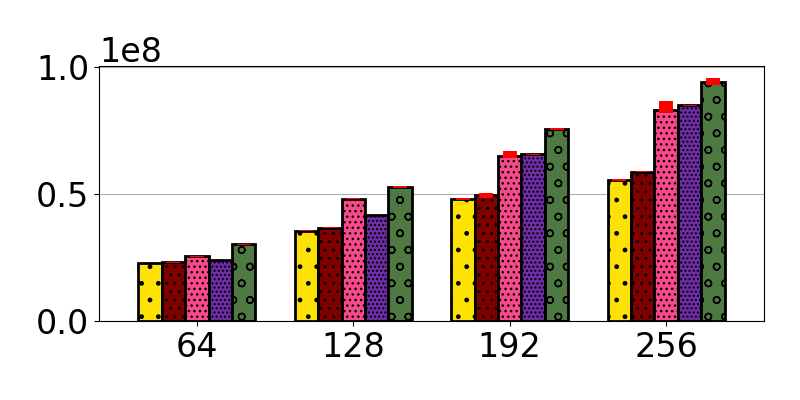}
\includegraphics[width=\plotwidth]{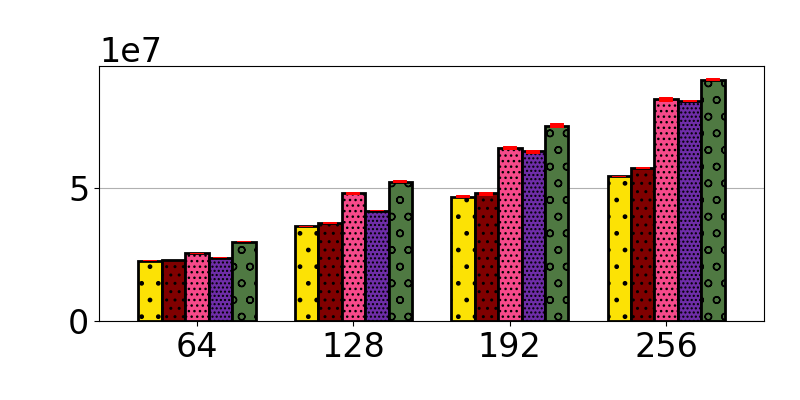}
\includegraphics[width=\plotwidth]{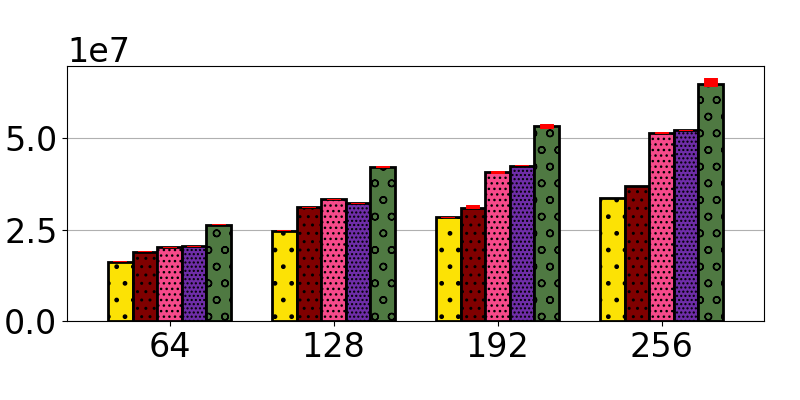}

\vspace{-3mm}
\includegraphics[width=\legendwidth]{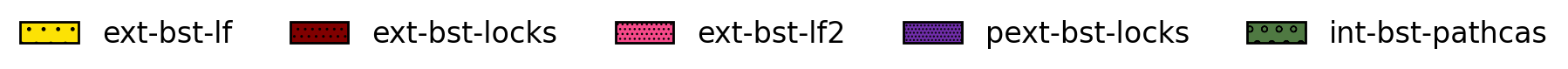}
}

{
\newcommand{\plotwidth}{0.32\linewidth}
\newcommand{\legendwidth}{0.6\linewidth}
\vspace{-2mm}
\rotatebox{90}{\hspace{2mm}\textbf{Balanced BSTs}}
\includegraphics[width=\plotwidth]{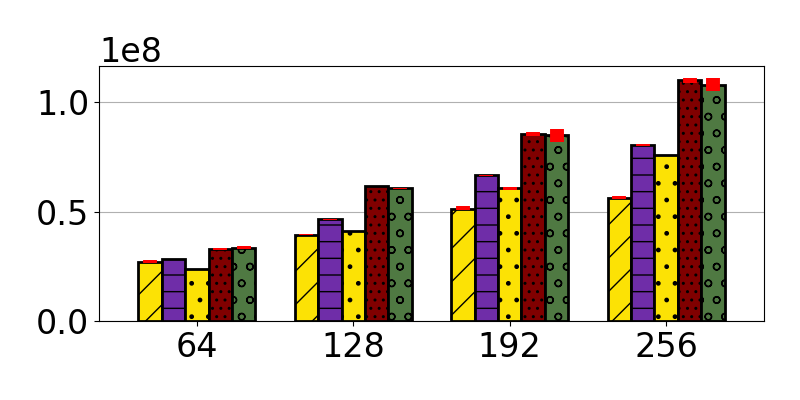}
\includegraphics[width=\plotwidth]{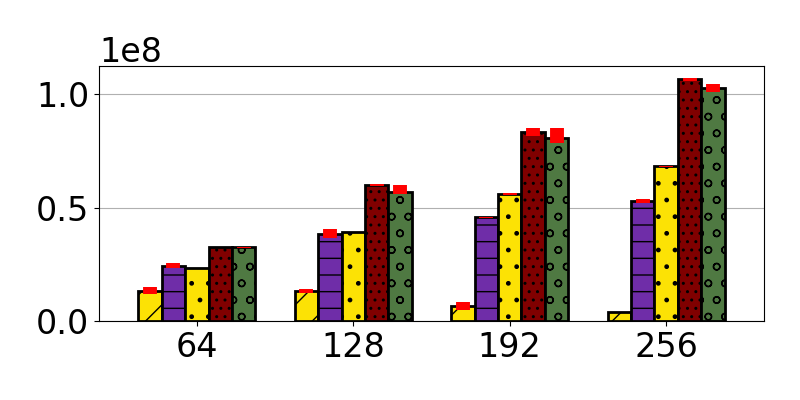}
\includegraphics[width=\plotwidth]{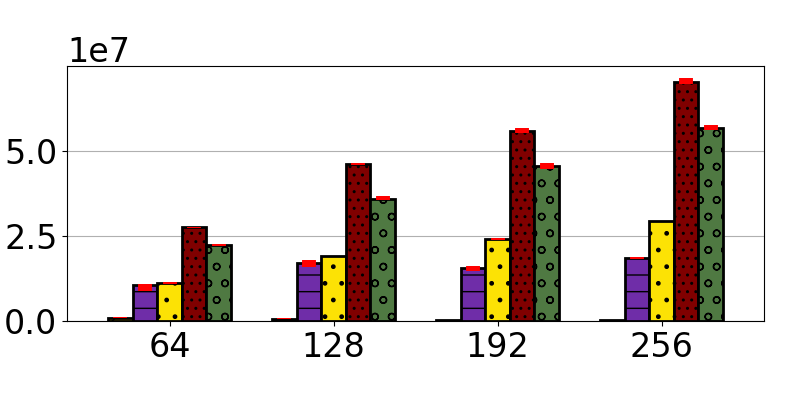}

\vspace{-3mm}
\includegraphics[width=\legendwidth]{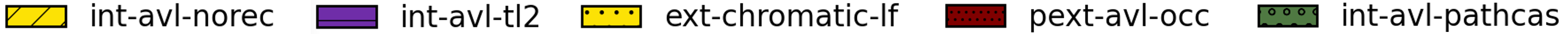}
}

\vspace{-4mm}
\caption{Selected results: \textbf{AMD} system, 10 million keys. %
x-axis = \# of threads. y-axis = operations/sec.
}
\label{fig:nasus-results}
\end{figure*}

\subsection{Extension: AVL Trees with PathCAS}

As a second example application of PathCAS, we extend our internal BST to perform relaxed AVL tree balancing.
Due to lack of space, we only give an overview here.
Complete details appear in the full version of the paper. %

We augment the BST nodes described in the previous section with two new fields: \textit{parent} and \textit{height}.
The former points to the node's parent %
and the latter contains the \textit{logical height} of the node, which can differ from the node's \textit{actual height} when the tree is unbalanced.

In \textit{insert}, we initialize newly created nodes' \textit{parent} pointers to point to the node they are being inserted under.
In \textit{delete}, whenever we perform a \textit{vexec} that removes an internal node, and hence changes the \textit{parent} of a child, we \textit{add} that child's \textit{parent} pointer to the \textit{vexec} as appropriate. %

A balance \textit{violation} exists at node \textit{n} when:
\begin{compactitem}
    \item n.left.height - n.right.height > 2; \textbf{or} 
    \item n.left.height - n.right.height < -2; \textbf{or} 
    \item n.height $\neq$  1 + max(n.left.height, n.right.height)
\end{compactitem}
A \textit{violation} can be created by any operation that causes node to \textit{gain or lose children}. %
Whenever an operation creates a \textit{violation}, it performs \textit{rebalancing steps} to fix the violation.

We implement the relaxed AVL tree rebalancing steps of Boug{\'e}~\cite{bouge}.
Boug{\'e}'s proved that, starting from an arbitrarily unbalanced tree, after performing a bounded number of atomic rebalancing steps (wherever they can be applied in the tree, and in any order), the tree will converge to a balanced state. %
Rebalancing steps are local modifications that affect a small number of nodes, and they do not need to be performed atomically at the same time as a search.
The rebalancing steps, namely \textit{rotateLeft}, \textit{rotateRight}, \textit{rotateLeftRight}, \textit{rotateRightLeft} and \textit{fixHeight}, are very similar to the familiar AVL tree rotations.
Rebalancing steps eliminate violations, or move them towards the root where they will be eliminated.
A tree with no violations is balanced.

Whenever a thread creates a violation at a node, it takes responsibility for repairing that violation, \textit{and any subsequent violations it creates while repairing that violation}.
More specifically, after performing a successful \textit{insert} or \textit{delete}, a thread traverses towards the root, fixing any violations it sees until either: it fixes a violation at the root, it observes a node on the path towards the root that has no violation, or it encounters a \textit{deleted} node on the path to the root (which means another thread has taken responsibility for any violations further along the path to the root). 
This is why we augment nodes with parent pointers: they allow us to easily ``follow'' violations up the tree.

\subsection{Freeing data structure nodes}

In unmanaged languages like C++, PathCAS manages its own memory, but the programmer must still manually reclaim memory for the \textit{data structures they implement} using PathCAS.
Reclaiming nodes that are deleted by a \textit{vexec} (or \textit{exec}) is quite simple using an algorithm such as DEBRA or NBR~\cite{DEBRA, nbr}.
The C++ implementation of DEBRA used in~\cite{IST} offers operations \textit{getGuard()} and \textit{retire(node)}.
The former is invoked at the beginning of each data structure operation.
The latter can be invoked on any \textit{node} after it is unlinked using \textit{vexec} (or \textit{exec}).
This will perform a \textit{delayed free} once no thread has a pointer to \textit{node}.
We use DEBRA to reclaim memory for all data structures in our experiments.

Using DEBRA is so mechanical that %
the necessary invocations of \textit{retire} could even be integrated directly into \textit{vexec} (and \textit{exec}), by having a successful \textit{vexec} \textit{retire} each node whose \textit{version} it \textit{marks} just before returning.

\vspace{-0.5em}
\section{Evaluation} \label{eval} %

\begin{figure}[t]
\vspace{-3.5mm}
   \begin{center}
   \scriptsize
 \begin{tabular}{| c | c | c|} 
\multicolumn{3}{c}{\textbf{Unbalanced BSTs}} \\

 \hline
 ext-bst-lf & External, lock-free, CAS  & \cite{Ellen10} \\ 
 \hline
 ext-bst-lf2 & External, lock-free, CAS and BTS  & \cite{Natarajan14} \\ 
 \hline
   ext-bst-locks & External, ticket locks  & \cite{david-tree} \\ 
 \hline
   pext-bst-locks & Partially-external, locks, logical ordering  & \cite{dana-tree} \\ 
 \hline
    int-bst-pathcas & Internal, lock-free, PathCAS & this work \\ 
 \hline
     int-bst-pathcas+ & Internal, lock-free, PathCAS + HTM fast-path & this work \\ 
 \hline
\multicolumn{3}{c}{\textbf{Balanced BSTs}} \\

 \hline
 ext-chromatic-lf & External, LLX/SCX  & \cite{Brown14} \\ 
 \hline
   pext-avl-occ & Partially-external, locks, logical ordering  & \cite{bronson-tree} \\ 
 \hline
    int-avl-pathcas & Internal, lock-free, PathCAS & this work \\ 
 \hline
     int-avl-pathcas+ & Internal, lock-free, PathCAS + HTM fast-path & this work \\ 
 \hline
\multicolumn{3}{c}{\textbf{Transactional Memory Algorithms}} \\
    \hline
 norec & STM  & \cite{NoREC} \\ 
 \hline
    hynorec & HTM fast-path + STM slow-path  & \cite{HybridNoRec} \\ 
 \hline
 rhnorec & HTM fast-path + HTM/STM slow-path  & \cite{rhnorec} \\ 
 \hline
   tle & HTM + Global Lock fallback  & this work \\ 
 \hline
    tl2 & STM & \cite{tl2} \\ 
 \hline

\end{tabular}
\end{center}
    \vspace{-4.5mm}
    \caption{The list of algorithms in our experiments.}
    \label{fig:algos_table}
\end{figure}

Our experiments follow the methodology of \cite{IST}, and we use the authors' publicly available benchmark, \textit{Setbench}.
We compare against state-of-the-art hand-crafted trees, as well as several TM-based trees (see Figure~\ref{fig:algos_table}).
We experimented with update rates (1\%, 10\% and 100\%) and uniform key ranges ($2\times10^{5}$, $2\times10^{6}$, $2\times10^{7}$).
Each trial pre-filled the data structure to contain half of the keyrange, then ran for 10 seconds.
Data is averaged over six trials with min/max bars shown in red.

Our AMD system has two EPYC 7662 CPUs, each with 64 cores and two hardware threads per core, for a total of 256 hardware threads, and a 256MB shared L3 cache. %
Threads are \textit{pinned} such that thread counts up to 128 run on one socket. %
We used GCC 10.1.0 \texttt{-O3}, the fast allocator jemalloc, and interleaved memory pages across sockets with \textit{numactl}.

Many more experiments can be found in the full version of this paper, including Intel results and additional algorithms and workloads. %

\begin{figure}[t]

\centering
\scriptsize
\begin{tabular}{|l|lllll|} 
\hline
                 & LLC miss & Cycles & Instr. & Avg. Key & Peak Mem.  \\ 
                 & per op & per op & per op & Depth & Usage  \\ 
\hline
ext-bst-lf       & 26.6            & 25872   & 2047         & 32.9           & 1137 MiB              \\
ext-bst-locks    & 24.6            & 23790   & 1527         & 30.5           & 1038 MiB              \\
ext-bst-lf2      & 15.3            & 15226   & 1739         & 30             & 720 MiB               \\
pext-bst-locks   & 16.3            & 16184   & 1011         & 28.6           & 2047 MiB              \\
\textbf{int-bst-pathcas}  & \textbf{13.7 }           & \textbf{12371}   & \textbf{2661}         & \textbf{28.7 }          & \textbf{539 MiB }              \\
\hline
int-avl-norec    & 141             & 8677653 & 789316       & 21.7           & 620 MiB               \\
ext-chromatic-lf & 31.5            & 27123   & 3204         & 24.3           & 2441 MiB              \\
int-avl-tl2      & 44.2            & 55514   & 7935         & 21.7           & 3642 MiB              \\
pext-avl-occ     & 11.6            & 11234   & 2398         & 23.2           & 844 MiB               \\
\textbf{int-avl-pathcas}  & \textbf{12.9}            & \textbf{13752}   & \textbf{3976}         & \textbf{21.7}           & \textbf{717 MiB}               \\
\hline
\end{tabular}

\vspace{-4mm}
\caption{Detailed analysis for 100\% updates, 256 threads.} %
    \label{fig:algos_stats}
\vspace{-3mm}
\end{figure}

\paragraph{Comparing unbalanced trees}
The top three plots in Figure~\ref{fig:nasus-results} present results comparing our unbalanced BST (\textit{int-bst-pathcas}) to a variety of leading hand-crafted unbalanced BST implementations.
Our code includes all fixes recommended by Arbel-Raviv et~al. for obtaining reliable BST performance results~\cite{BSTRoot}.
Our PathCAS BST often significantly outperforms its competitors.
We explain why using Figure~\ref{fig:algos_stats}.

As expected, out of the unbalanced BSTs, \textit{int-bst-pathcas} performs the largest number of instructions per tree operation.
However, it has the smallest cycle count per operation, suggesting that its instructions can be pipelined more efficiently. %
Crucially, \textit{int-bst-pathcas} incurs the smallest number of last-level cache misses, because of its low average key depth and peak memory usage.
The \textit{ext-bst-*} BSTs are \textit{external}: the keys that are semantically contained in the dictionary are stored in the \textit{leaves}, and internal nodes contain dummy \textit{routing} keys.
External trees contain more nodes than internal trees and are taller.
\textit{pext-bst-locks} is \textit{partially} external, which means it is somewhat closer to an internal tree.
Despite its low average key depth, its nodes are quite large, containing multiple locks and many pointers, as they participate in both a tree and a doubly linked list.
It underperforms because of its peak memory usage and LLC misses.

\paragraph{Comparing balanced trees}
The bottom three plots in Figure~\ref{fig:nasus-results} compare our AVL tree (\textit{int-avl-pathcas}) with other balanced BSTs, including \textit{pext-avl-occ}~\cite{bronson-tree}, which is known to be the fastest concurrent BST in many workloads~\cite{BSTRoot}.
In read-mostly workloads, \textit{int-avl-pathcas} is competitive with \textit{pext-avl-occ}, and outperforms the other trees. %

In the 100\% update workload, \textit{int-avl-pathcas} is at most 20\% slower than the fastest algorithm.
The fact that \textit{int-avl-pathcas} is not \textit{drastically} outperformed by the highly tuned and intricate \textit{pext-avl-occ} tree is remarkable.\footnote{The \textit{pext-avl-occ} tree is intricate, and carefully engineered, using sophisticated sequence locks that encode whether ongoing rebalancing operations are \textit{shrinking} or \textit{expanding} the key range, and allowing a key that was marked as deleted to be reinserted simply by changing a bit. This may inflate its performance in the types of workloads used in our experiments.}
The \textit{ext-chromatic-lf} tree, which is implemented using the LLX and SCX primitives, does not fare nearly as well.

\paragraph{TM-based trees}
It should be noted that in an attempt to be generous to these TM approaches, in our implementations we compiled each TM in the same compilation unit as the data structure (rather than as a linked library), and force-inlined all TM code, eliminating the overhead of function calls to the TM code from the data structure. This optimization would be unrealistic in practice, however should give the TM implementations the best performance possible for comparison. 
Despite this, the TM based algorithms in Figure~\ref{fig:nasus-results} still suffer from high instruction counts and LLC miss rates (Figure~\ref{fig:algos_stats}).
In particular, the extremely high instruction counts for \textit{int-avl-norec} are due to contention on the global version lock and repeated read set validation to guarantee opacity. %

The results in the introduction (Figure~\ref{fig:example_perf}), from our Intel system, include more algorithms, since the system has hardware transactional memory support.
In those results, our trees outperform the next fastest algorithm, TLE, by nearly 2x.
Moreover, those results are ``generous'' to TLE, since its global locking fallback code path degrades performance dramatically in workloads with more updates.

\subsection{Comparison with MCMS}
To compare PathCAS with MCMS, we extended Timnat's original C++ code for the MCMS linked list.
We note that we found some bugs in his implementation of MCMS, one of which only affects the source code, and one of which affects the algorithm in the MCMS paper.
Details are to appear in the full version of this paper.
We fixed these bugs, and applied the same lock-free descriptor optimizations that we use, and implemented an internal BST using MCMS to validate the entire search path (similarly to how we validate the entire search path using PathCAS).
The implementation is optimized to the best of our ability: for instance, it avoids performing MCMS in cases where searches return true or inserts return false.
Moreover, deletes that return true perform their modifications in small MCMS operations that do not include the search path.
The hardware TM code path in our MCMS implementation is faithful to the MCMS paper in that it does not write to nodes on the search path unless the operation is an update that intends to modify that node.

\begin{figure}[t]

        \scriptsize
        \noindent\begin{tabular}{|l|lll|lll|} 
            \hline
            & \multicolumn{3}{c|}{\textbf{100\% update}} & \multicolumn{3}{c|}{\textbf{100\% search}} \\
            \hline
            \# threads & PathCAS & MCMS+ & MCMS- & PathCAS & MCMS+ & MCMS- \\ 
            \hline
            1	&2.50	&0.99	&1.12   &3.19	&1.25	&1.54\\
            2	&4.42	&0.69	&0.70	&6.33	&0.78	&0.77\\
            4	&7.94	&0.56	&0.54	&12.44	&0.63	&0.62\\
            8	&13.97	&0.50	&0.50	&24.07	&0.55	&0.54\\
            48	&46.67	&0.35	&0.31	&77.23	&0.40	&0.31\\
            190	&79.71	&0.04	&0.03	&263.00	&$<0.1$	&$<0.1$\\
            \hline
        \end{tabular}

    \vspace{-4mm}
    \caption{Intel results for an internal BST implemented with PathCAS, MCMS with HTM (MCMS+), and MCMS without HTM (MCMS-), respectively. Tree initially contains 100,000 keys. Results are in millions of operations per second.} %
        \label{fig:exp-mcms}
    \vspace{-3mm}
\end{figure}

Results appear in Figure~\ref{fig:exp-mcms}.
We find that the resulting tree is orders of magnitude slower than our PathCAS tree.
For example, for 100\% searches and 100,000 keys, on a system with hardware TM, the throughput of the MCMS tree is 0.4M at 48 threads, whereas the throughput of PathCAS is 77M.
According to \texttt{perf stat -e tx\_commit,tx\_abort}, half of MCMS transactions abort at 8 threads, 84\% abort at 48 threads, and there are many capacity aborts even with a single thread.
As we suspected, even transient aborts in MCMS can quickly turn into cascading aborts and software path executions.
At 190 threads, PathCAS is thousands of times faster than MCMS.

\subsection{Comparison with Elastic TM}

\begin{figure}[t]
        \centering
        \small
        \noindent\begin{tabular}{|l|lllll|} 
            \hline
            \# threads: & 24 & 48 & 96 & 144 & 190 \\
            \hline
            \textit{ext-bst-elastic}    &17.24	&29.54	&52.25	&71.59	&90.43\\ 
            \textit{ext-bst-lf2}        &29.28	&53.07	&129.02	&195.71	&267.44\\
            \hline
        \end{tabular}

    \vspace{-4mm}
    \caption{Intel results for elastic transactions with 1\% updates and 99\% searches. The trees initially contain 10,000,000 keys. Results are in millions of operations per second.}
        \label{fig:exp-elastic}
    \vspace{-3mm}
\end{figure}

We also compared PathCAS with \textit{elastic transactions}, which use sophisticated synchronization techniques to split transactions into smaller atomic pieces to improve performance. %
Specifically, we compared with \textit{ext-bst-elastic}, the ``speculation-friendly BST'' from Synchrobench~\cite{synchrobench} that is built using elastic transactions. %
We were unable to port the Elastic STM library that \textit{ext-bst-elastic} depends on into Setbench, as it would require overhauling Setbench to support the background rebalancing threads used by \textit{ext-bst-elastic}, and would require us to completely redesign its memory reclamation. %

So, we obtained performance numbers for \textit{ext-bst-elastic} using Synchrobench, instead of Setbench.
We also obtained performance numbers for \textit{ext-bst-lf2} on the same workload, as it is included in Synchrobench as well as setbench.
Results appear in Figure~\ref{fig:exp-elastic}.
This gives us a sort of limited point of comparison between our results in Setbench and our results in Synchrobench.
Although the results do not allow for a rigorous comparison between \textit{ext-bst-elastic} and the other trees in our experiments, we note that it is \textit{much} slower than \textit{ext-bst-lf2}, which is in the middle of the pack in our Setbench experiments.
And, we also note that \textit{ext-bst-elastic} is being evaluated in an environment that is presumably most favourable to it, as \textit{ext-bst-elastic} was developed by the authors of Synchrobench, and integrated therein by the authors.
Additionally observe that this 1\% update workload is quite favourable to \textit{ext-bst-elastic}, as its performance degrades faster than ext-bst-lf2 as the update rate increases.

\section{Conclusion}

This paper introduced PathCAS, a primitive used to implement efficient concurrent data structures while maintaining lower complexity than hand-crafted techniques.
PathCAS utilizes an HTM fast path combined with an efficient fallback path that relies on KCAS, version numbers, and search path validation.

We implemented a set of historically difficult data structures using PathCAS, including an internal balanced tree, and have shown them to achieve competitive performance with the best fine-grained variants. 
We compared the performance of our data structures with both TM-based and hand-crafted variants of these structures, showing that PathCAS surpasses TM based approaches, and can rival the state-of-the-art in hand-crafted designs.

While in this work we focus on the two trees presented, we emphasize that one can use PathCAS in a direct way to implement many data structures wherein an operation consists of a read phase followed by a write phase.
The construction is similar to our trees: \textit{visit} each node that will be read or written, then \textit{add} and \textit{exec} the necessary modifications.
Using this approach, we were able to implement the following: (a,b)-trees, lists, hash-tables, hash-lists, stacks, queues, skip-lists and dynamic graph connectivity data structures.
We chose to focus on trees in the paper, since they have traditionally been difficult to implement and prove correct. However, we plan to share more of our implementations as part of future work.

\begin{acks}
This work was supported by: the Natural Sciences and Engineering Research Council of Canada (NSERC) Collaborative Research and Development grant: CRDPJ 539431-19, the Canada Foundation for Innovation John R. Evans Leaders Fund with equal support from the Ontario Research Fund CFI Leaders Opportunity Fund: 38512, Waterloo Huawei Joint Innovation Lab project ``Scalable Infrastructure for Next Generation Data Management Systems'', NSERC Discovery Launch Supplement: DGECR-2019-00048, NSERC Discovery Program under the grants: RGPIN-2019-04227 and RGPIN-04512-2018, and the University of Waterloo. We would also like to thank the reviewers for their insightful comments. %
\end{acks}

\bibliography{ref}

\section{Artifact}

We have provided a docker container in which you can download and start running the experimental benchmark. 
\begin{itemize}
    \item The docker container was prepared on Ubuntu 20.04, which is also the OS the actual container is running
    \item In particular, we build our container using \texttt{Docker version 20.10.2, build 20.10.2-0ubuntu1~20.04.2}. %
    \item Our experiments are run on a machines with 384 GB of RAM and 192 threads total, while possible to run these experiments on a different machines with less memory or threads:
    \begin{itemize}
        \item If you have less memory, you may need to use smaller data structure sizes 
        \item If you have fewer threads, you will simply need to use less threads
    \end{itemize}
\end{itemize}

Inside the docker container, in the \texttt{experiments} folder, you will find a script called \texttt{run.sh} that builds all data structures and runs a simple test (using each data structure).
This script is driven by the experimental configuration described in \texttt{experiments/}{\texttt{\_common.py}.
By default, the contents of \texttt{experiments}\texttt{/\_common.py} %
cause \texttt{run.sh} to reproduce Figure~3.
More details on how to modify this experimental configuration appear below. 

Note that \texttt{run.sh} performs experiments for both transactional memory and non-transactional memory data structures.
This can take quite some time, so you may want to limit the number of trials in \texttt{\_common.py} to a relatively small number.

\subsection{Step By Step Instructions}

\begin{enumerate}
    \item Install docker (ideally the same version we use) on your system with your preferred package manager
    \item Download the docker image from Zenodo (\url{https://zenodo.org/record/5728166}).
    For example, on Linux you might do this by executing:

    \texttt{\$ wget https://zenodo.org/record/5728166/fil}\\
    \texttt{es/ppopp2022pathcas.tar.gz}
    
    \item Add the image to your local Docker images (running in the same directory as step 2):
    
    \texttt{\$ docker load -i ppopp2022pathcas.tar.gz}
    \item Launch the docker container
    
    \texttt{\$ docker run -i -t --privileged ppopp2022pathc}\\
    \texttt{as /bin/bash}
    
    Note that \texttt{privileged} is \textit{required} in order to ascertain the proper thread pinning strategy for the experiments, and to record performance counters for, e.g., cache statistics.
    You may also need to set \texttt{kernel.perf\_event\_} \texttt{paranoid} to \texttt{-1} on Linux.
    
    \item Go to the experiments folder \\ 
    \texttt{\$ cd experiments}
    
    \item The experiments can be configured in \texttt{\_common.py}.
    Be sure to edit the HOST CONFIGURATION and EXPERIMENTAL CONFIGURATION sections of \texttt{experiments/} \texttt{\_common.py} to match the machine you are running on, and to reflect the types of experiments you would like to run.
    Note that the various configuration parameters are described in the comments.
    
    By default, \texttt{\_common.py} is configured to reproduce Figure 3 in the paper.
    Running such experiments can take many hours (at least 12).
    \textbf{If you would like to run a shorter test to get started,} please uncomment the more restrictive \textit{testing} values of \texttt{ins\_del\_fractions}, \texttt{max\_keys}, \texttt{exp\_duration\_millis}, \texttt{thread\_counts} and \texttt{num\_trials} provided alongside the default values in \texttt{\_common.py} before proceeding.
    (Once you have run your small test, you can comment those testing values out again.)

    Additionally, Smaller example test values are provided if you wish to run a quick proof of concept test (commented out next to the actual set values). 

    \item To run the experiments described in \texttt{\_common.py}, simply execute \texttt{\$ ./run.sh}.
    It will run experiments, store the results in a sqlite3 database, process those results using various SQL queries, and produce several text data files and plot images (more on this below).
    
    If the \textit{testing} values in \texttt{\_common.py} are \textbf{uncommented}, this should take less than five minutes.
    
    If the \textit{testing} values are \textbf{commented out}, then \texttt{run.sh} will perform all necessary experiments to reproduce Figure 3 in the paper.
    \textbf{This should take approximately 12 hours.}
    
    While experiments are being performed, the scripts will produce output of the form ``\texttt{step 000001 / 000336...}''. Each such line contains the command being executed, as well as a very rough \textit{estimate of the remaining running time} for all experiments.
    \item
        After the runs are complete, \textbf{PNG format plot images} for the experiments can be found in directories \texttt{experiments/data\_tm} and \texttt{experiments/data\_non\_tm}.
        Note that, by default, a large super set of the plots in the paper are generated.
        
        To more easily view these images, you may want to copy them from the docker container to your host machine using \hyperlink{https://docs.docker.com/engine/reference/commandline/cp/}{docker cp}.
        For example, \textit{while the docker container is running}, you can run the following command \textit{on the host machine}, where \texttt{<CONTAINER\_NAME>} is the name of the running container.\footnote{To obtain the name of the running container, run \texttt{\$ docker container list} and look in the \texttt{NAMES} column.}
        
        \texttt{\$ docker cp <CONTAINER\_NAME>:/root/tmbench/ex}\\
        \texttt{periments .}
        
        This will copy all experimental results to the host machine.
        Alternatively, you can browse the results in text format directly inside the docker container:
        
        \begin{enumerate}
            \item
                A detailed summary of the numerical data can be obtained by starting in the \texttt{experiments} directory and running \texttt{./basic\_info.sh} to produce a pretty-printed table, with columns for update rate, data structure name, key range size, thread count and throughput.
                If you ran both TM and Non-TM experiments, this script will print a table for each type of results.

            \item
                Note that you can also perform your own arbitrary SQL queries on the sqlite database, either by modifying the queries in \texttt{basic\_info.sh}, or by entering an interactive sqlite console in either of the \texttt{data\_*} directories in \texttt{experiments/}:
                    
                \texttt{\$ sqlite3 output\_database.sqlite}
            \item
                If you want to view the complete raw text data (which contains considerably more information than the sqlite database), you can find files for each experimental ``step'' stored in \texttt{experiments/data\_tm} and \texttt{experiments/data\_non\_tm}, with file names of the form: \texttt{data0*<STEP>.txt}.
            
                There are also CSV-format tables of data that are converted into various plots in file names of the form:
                \texttt{\{DATA\_STRUCTURE\_TYPE\}\_\{METRIC\}\_u\{UPDATE RA}\\
                \texttt{TES\}-k\{MAX\_KEY\}.\{FILE\_TYPE\}}
            
                For example, \texttt{data\_tm/bst\_tm\_total\_throughput}\\
                \texttt{\_sec-u50.0\_50.0-k20000000.txt} contains a table of data that is used to produce a throughput comparison plot for TM-based binary search trees, with 50\% inserts and 50\% deletes, and a key range size of 20 million (i.e., initial data structure size of 10 million).
                This example is of particular note, as it corresponds to one of the plots in Figure 3. %
        \end{enumerate}
        
    \item
        Note that, due to the nature of docker containers, \textbf{all data will be lost if you \texttt{exit} the \texttt{docker run}.}
        If you want to save any of your generated data, you can do so using \hyperlink{https://docs.docker.com/engine/reference/commandline/cp/}{docker cp} from a different terminal on the host machine to copy the relevant data from the container to the host.
    
        Alternatively, you can use \hyperlink{https://docs.docker.com/engine/reference/commandline/commit/}{docker commit} to save a new version of the docker image that includes all of the data you've generated.
        You can then launch \textit{that} image in a docker container to access your data again.
    
    \item
        Once you have saved any data you want to keep, you can \textit{exit} the docker container by running \texttt{exit} in the container, or \texttt{docker stop <CONTAINER\_NAME>} on the host.
\end{enumerate}

\paragraph{LLC misses on AMD}
Note the artifact fails to track LLC misses on recent AMD processors.
Instead, they must be tracked manually using \texttt{perf stat}. %

For example, on our machine, we used a command of the form:
\texttt{[usual command prefix up until the binary] perf stat -e l3\_comb\_clstr\_state.request\_miss [act}\\
\texttt{ual binary being run] [args]}.

Also note that will produce a raw number of LLC misses for the entire execution, so one will need to divide by the number of operations completed.
This number will be somewhat inflated because it includes experimental setup and tear-down as well as prefilling.
To get around this, %
use \texttt{perf record -k CLOCK\_MONOTONIC}, which timestamps all cache misses with a clock that is compatible with our benchmark.
Then one can take a pair of timestamps emitted by our benchmark, \texttt{REALTIME\_START\_PERF\_FORMAT} and \texttt{REALTIME\_END\_P}\\
\texttt{ERF\_FORMAT}, and plug them into \texttt{perf report --ns --time <start>,<stop>} where \texttt{<start>} and \texttt{<stop>} are our timestamps. %

\paragraph{Generating Figure~5}
Figure~5 is generated by running the following SQL queries on the results databases (starting in the same directory as \texttt{\_common.py} and \texttt{run.sh}). %
\begin{verbatim}
for exp in tm non_tm ; do                             \
  ../setbench/tools/data_framework/run_experiment.py  \
  _exp_${exp}.py -q "select maxkey, TOTAL_THREADS,    \
  alg, mem_maxresident_kb, PAPI_L3_TCM, PAPI_TOT_CYC, \
  PAPI_TOT_INS, tree_stats_avgKeyDepth from data      \
  where ins_del_frac = '50.0 50.0'                    \
  order by maxkey desc, total_threads desc, alg" ; done
\end{verbatim}

\clearpage
\appendix
\noindent\huge{\textbf{Supplementary Material}}
\normalsize

\section{Additional PathCAS Details}

\section{Hardware TM based PathCAS} \label{a:htm}

\begin{algorithm}[h]
\caption{PathCAS::execute()}\label{kcas:exe}
\begin{algorithmic}[1]
   \State \textit{desc = getDescriptor()} \Comment{Gets the current thread's descriptor}
   \For{tries = 0; tries $<$ 5; tries++}
     \If{((status = \_xbegin()) == \_XBEGIN\_STARTED)}
     \IIf{\textbf{not} \textit{validate(desc)}}  \textit{\_xabort(OLD)} \Comment{Validation}
     \For{\textbf{each} (addr, old, new) \textbf{in} desc.entries}
       \State \textit{val = *addr} \Comment{Check if addresses contain old values}
       \If{\textit{val $\neq$ old}}
         \IIf{\textit{isDescriptor(val)}} \textit{\_xabort(DESCRIPTOR)}
         \ElseI
            \State \textit{\_xabort(OLD)} \Comment 
       \EndIf
     \EndFor
     \For{\textbf{each} (addr, old, new) \textbf{in} desc.entries}
       \State \textit{*addr = new} \Comment{Write new values}
     \EndFor
     \State \textit{\_xend()}
     \State \textbf{return} \textit{true}
     \ElsIf{status $\&$ \_XABORT\_EXPLICIT}
     \IIf{\textit{\_XABORT\_CODE(status) == DESCRIPTOR}} \textbf{break}
     \ElsIIf{\textit{\_XABORT\_CODE(status) == OLD}} \textbf{return}  \textit{false}
     \EndIf
   \EndFor
   \State \textbf{return} \textit{help(desc)} \Comment{Fallback code path}
\end{algorithmic}
\end{algorithm}

The HTM implementation (Algorithm \ref{kcas:exe}) first checks all the fields added to the KCAS descriptor to ensure that all the fields hold the old value from the descriptor. If any of these fields contain the incorrect old value, the transaction is explicitly aborted. If the incorrect value read was a KCAS descriptor, the slow path is taken, we cannot simply return false, as we do not know the logical value of this field. If the value was not a KCAS descriptor, then the KCAS can simply return false, as it contained a value other than the one in the descriptor. 

If all the fields are read to contain the old values in the descriptor, they are now in the \textit{read set} of the transaction. Any change to any of these fields will result in an abort, and the transaction will retry up to a fixed number of attempts. If the transaction succeeds in writing all the new values to the fields with no conflicts, the transaction will commit, and \textit{true} can be returned.

\section{Lock-freedom of PathCAS}

Without loss of generality, in the following, we focus on \textit{vexec} operations (since \textit{exec} operations are just a special case).

The \textit{start} and \textit{add} operations are wait-free.
The \textit{read}, \textit{visit} and \textit{vexec} operations only perform a bounded number of steps in addition to an invocation of \textit{help}, but \textit{help} can invoke itself recursively.
We prove that these operations are lock-free.
Suppose, to obtain a contradiction, that threads take infinitely many steps in \textit{read}, \textit{visit} or \textit{vexec} operations, but only finitely many such operations terminate.
Clearly, threads that take infinitely many steps in such an operation must take infinitely many steps in \textit{help}.

The only way a thread can take infinitely many steps in \textit{help} is if it continues performing new recursive invocations of help at line~8 of the pseudocode for \textit{help}.
As in the HFP algorithm, (1) a \textit{vexec} can only be helped at line~8 if it has ``locked'' some node, and (2) because all addresses are sorted and ``locked'' in order, if a \textit{vexec} $O_1$ is helping another \textit{vexec} $O_2$, then all of $O_2$'s ``locked'' nodes come strictly after all of $O_1$'s ``locked'' nodes in the sort order.
Since there are infinitely many invocations of \textit{help} and only finitely many invocations of \textit{vexec}, eventually some thread's call stack must contain a \textit{helping cycle}: an invocation of \textit{help(desc1)} that (possibly indirectly) invokes \textit{help(desc2)} that (possibly indirectly) invokes \textit{help(desc1)}.
This implies that the nodes ``locked'' for \textit{desc1} come \textit{strictly after themselves} in the sort order---a contradiction.
So, no such cycle can exist, which implies that \textit{help} can only be invoked finitely many times at line~8 of the pseudocode for \textit{help}.
Therefore, \textit{read}, \textit{visit} and \textit{vexec} must be lock-free.

\subsection{Design Decision: Manual Incrementing of Version Numbers} \label{a:increment}
We contemplated building the incrementing of version numbers into the abstraction, so that it would be automatic.
However, we decided that requiring addresses passed to \textit{add} to be fields of nodes might be overly restrictive.
We do not want to rule out applications wherein PathCAS could be used to atomically validate a set of nodes, and also modify arbitrary fields that do not belong to a data structure node (such as a size variable).
Therefore, we only require nodes that are \textbf{passed to} \textit{visit} to have version numbers to track changes, and leave it to the programmer to manage them.
Note that our interface supports debugging mechanisms to catch errors in managing version numbers.\footnote{%
For example, \textit{visit} can save the address ranges of all visited nodes, and \textit{exec} can then check for intersections between the \textit{visited} nodes and \textit{added} addresses that do not have a corresponding \textit{node.ver} increment. This introduces overhead, but can be enabled only when debugging.}
In an application where it is acceptable to restrict PathCAS so that it only accesses \textit{nodes}, one could easily change \textit{add()} to also take a node pointer in addition to the field pointer, and automate version increments. %

\subsection{Using Version Numbers} \label{a:wrap}

Version number based validation can theoretically cause ABA problems if version numbers wrap around. However, the probability of an actual failure is extremely low. Suppose version numbers are 64 bits long, and a slow helper is attempting to use DCSS to change an address from $\langle old, version = 17 \rangle$ to $\langle new, 18 \rangle$. For an ABA problem to cause this DCSS to erroneously succeed, the helper must sleep for exactly the right amount of time for this specific address to undergo \textit{precisely} $k \cdot 2^{64}$ changes for $k \in \mathbb{Z}^+$ before waking up and attempting the DCSS.
\section{AVL Extension} \label{a:avl}
Within the main body of this work, we emulated a relatively {simple example} of using PathCAS by outlining a internal, unbalanced, BST.
However PathCAS can also be leveraged for the creation of much more complex structures, and one such example is a similar tree, but augmented to be self-balancing. 
Rebalancing steps are based on the relaxed AVL tree of Boug{\'e}~\cite{bouge}.

\begin{figure}[H]
    \includegraphics[width=.15\textwidth]{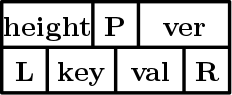}
\centering
\caption{AVL Node Structure}
\label{fig:node}
\end{figure}

In order to facilitate rebalancing, the structure of the node must change in two key ways: by tracking the \textit{logical height} of the node, and containing a pointer to its parent.
Node store \textit{logical} heights (stored in the \textit{height} field), and also have an \textit{actual} height (the number of nodes below this node within the tree to the farthest leaf + 1), which can differ. 
We say a node is \textbf{balanced} if the logical heights of its children differ by less than 2. 

These two new fields can be see within Figure \ref{fig:node}, where \textit{P} is the parent pointer and \textit{height} is the height field. The other fields are shared with the BST shown in the main body of this work.

Five rebalancing steps are proposed: \textit{rotateLeft(node)}, \textit{rotateRight(node)}, \textit{rotateLeftRight(node)}, \textit{rotateRightLeft(node)}, and \textit{fixHeight(node)}. 
These rotations are very similar to sequential AVL tree rotations, and those outlined by \cite{bouge}.
\textit{fixHeight(node)} updates the \textit{height} field of a node
to $1+max(hL,hR)$, where $hL$ (resp., $hR$) is the \textit{height} field of its left (resp., right) child (with the eventual goal of propagating accurate height throughout the tree).
Boug\'e proves that applying these rebalancing steps in a tree, wherever they apply, and in any order, until no more \textit{can be applied}, yields a strict AVL tree.

The operations from the previous section carry over with two changes:
After a successful insertion just before line~\ref{line:insert:return-true} of \textit{insertIfAbsent}, we invoke \textit{rebalance(p)} (Algorithm \ref{rebalance}) and just before line~\ref{line:delete:two-return-success} of \textit{delete} we invoke \textit{rebalance(sp)}.
\textbf{Generally, after any tree modification, \textit{rebalance} is invoked on any node that (may have) gained or lost children.}

The \textit{rebalance} operation determines whether to perform a rotation on \textit{n}, update its height, or do nothing based on its \textit{apparent balance}.
\textit{n}'s apparent balance is calculated by checking the heights of its children (Line 12). 
If \textit{n} requires rebalancing (the apparent balance is $\geq +2$ or is $\leq -2$), a direction is determined: a positive apparent balance indicates that \textit{n}'s left subtree is larger and requires a rotation in the opposite direction, and vice-versa. 
Depending on the balance of \textit{n}'s children, a single rotation may be insufficient to repair the imbalance of \textit{n}. If this is the case, a double rotation is used, which involves applying a single rotation to one of \textit{n}'s children, and another one to \textit{n}. 
For performance reasons, we combine double rotations into a single large PathCAS. (By doing this we can also avoid updating version numbers, etc., twice.)
A visualization of tree states that lead to the possible rotation is shown in Figure \ref{fig:rotations}. 
\begin{figure}[H]
    \includegraphics[width=.45\textwidth]{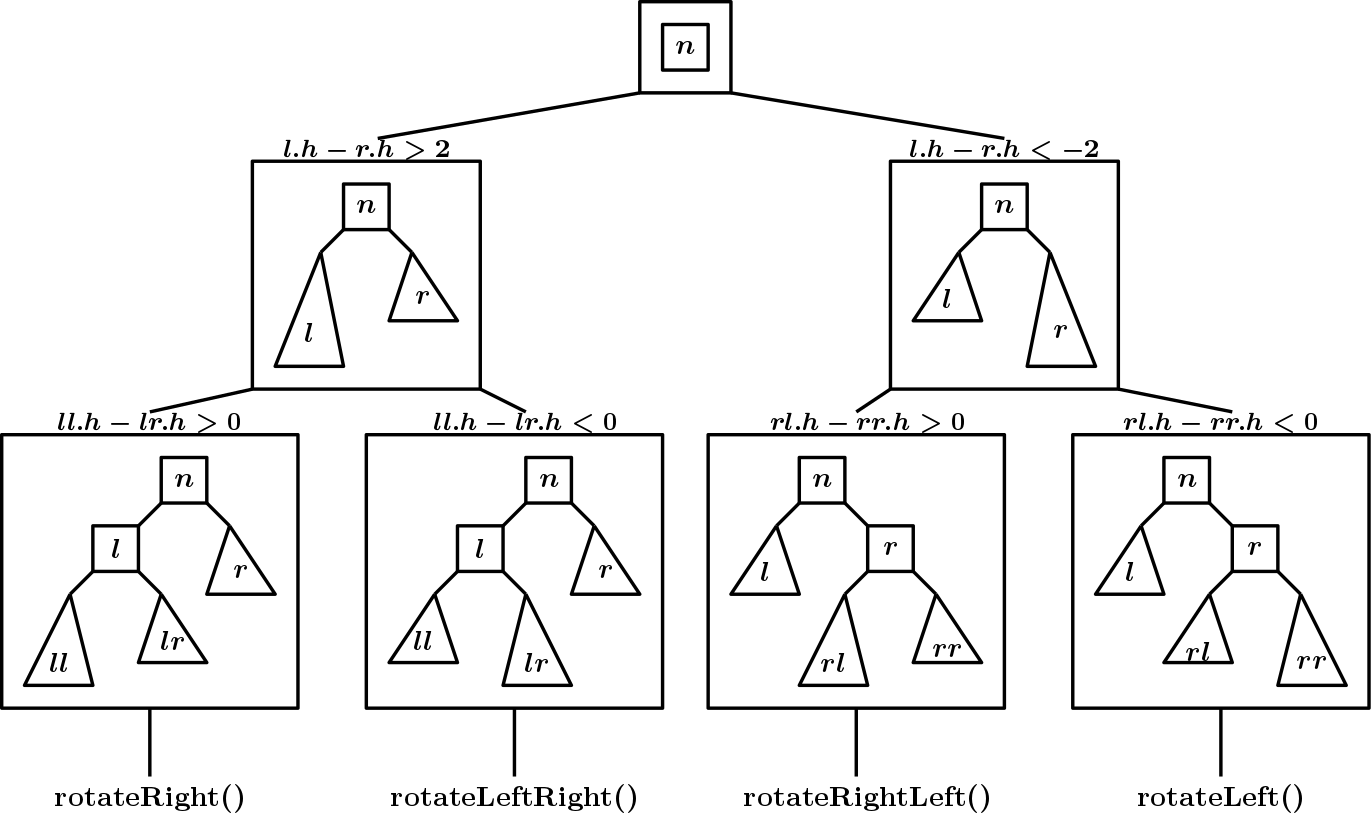}
\centering
\caption{Possible Rotations}
\label{fig:rotations}
\end{figure}

The simple case in which only a single rotation is required (Line 24) is addressed by \textit{rotateRight} (Algorithm \ref{rotateRight}, Figure \ref{fig:rrbst}) or \textit{rotateLeft} (symmetric).
In our example, \textit{rotateRight}, we move \textit{n}'s left child, \textit{l}, to \textit{n}'s position. 
Additionally, \textit{n} replaces its left child pointer to \textit{l} with \textit{l}'s right child, \textit{lr}.  
New heights for all nodes involved are calculated based on this rotation (except for \textit{parent}, as it will be checked for imbalance after this operation anyways) and are updated as part of the PathCAS. 
Note that in Figure \ref{fig:rrbst} the version numbers of \textit{ll} and \textit{r} are blue as they nodes that are part of the sequential rotation but are not modified, but are validated as part of \textit{vexec}. 
This is because the heights of these nodes are used to calculate the new heights for other nodes, and changes to these nodes could result in a rotation that does not improve the balance of the tree. 
If \textit{n} does not need rebalancing, we call \textit{fixHeight} to ensure its height is accurate with respect to its children.

Note that rotations may move nodes off of this parent pointer path towards the root, potentially making threads unable to reach nodes they need to rebalance. 
Thus, after a rotation is successful, \textit{rebalance} is called recursively on every node that had its \textit{left} or \textit{right} fields changed.

\begin{figure}[h]
\centering
    \includegraphics[width=.9\linewidth]{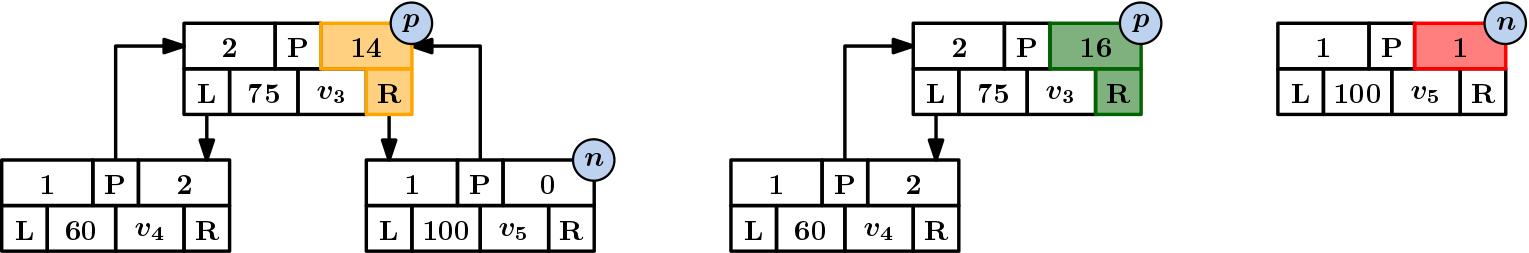}
\centering
\caption{AVL delete Leaf}
\label{fig:bstdelete0}
\end{figure}

\begin{figure}[h]

    \includegraphics[width=.9\linewidth]{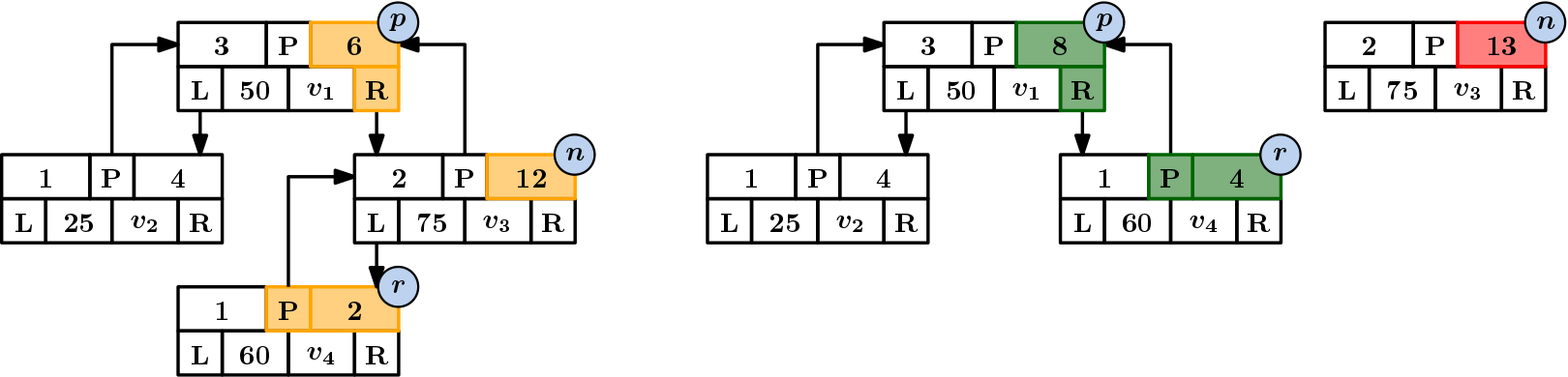}
\centering
\caption{AVL delete one child}
\label{fig:bstdelete1}
\end{figure}

\begin{figure}[h]

    \includegraphics[width=.9\linewidth]{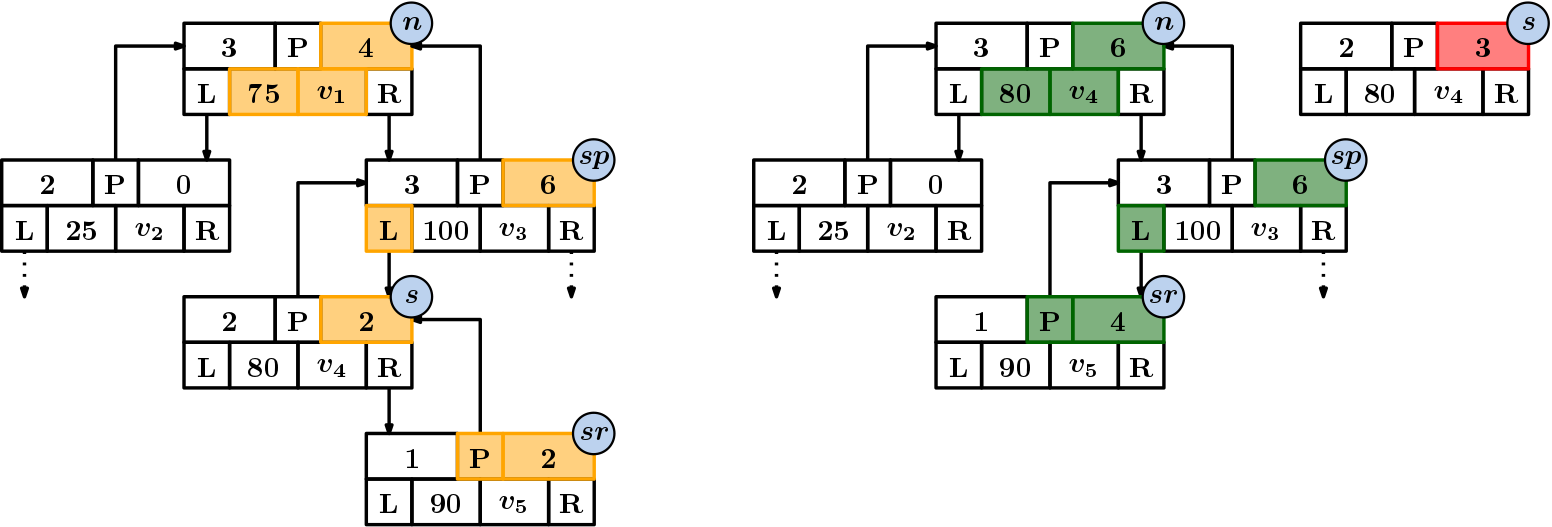}
\centering
\caption{AVL delete two children}
\label{fig:bstdelete2}
\end{figure}

\begin{figure}[h]

    \includegraphics[width=.9\linewidth]{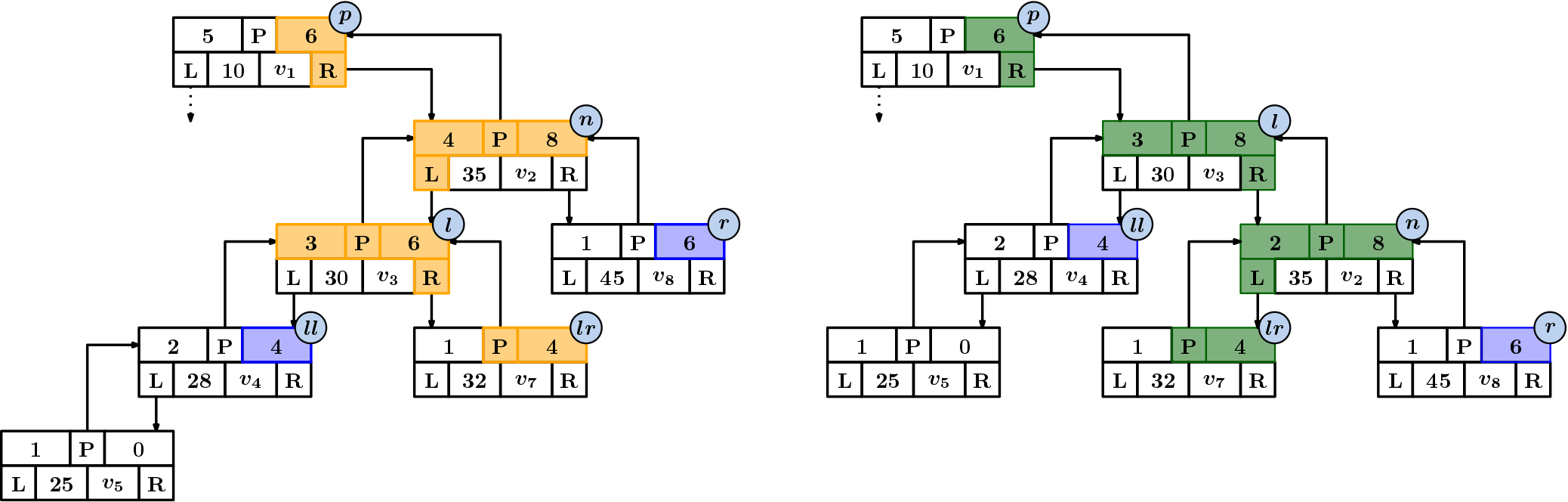}
\centering
\caption{AVL right rotation}
\label{fig:rrbst}
\end{figure}

\begin{algorithm}[h]
\caption{fixHeight(\textit{n, nVer})}\label{fixHeight}
\begin{algorithmic}[1]
\State \textit{l = n.left}
\State \textit{r = n.right}

\State \textit{rVer = r.ver}

\If{\textit{l $\neq$  NIL}}
    \State \textit{lVer = visit(l)}
	\State \textit{add($\&$l.ver, lVer, lVer)}
\EndIf
\If {\textit{r $\neq$ NIL}}
    \State \textit{rVer = visit(r)}
    \State \textit{add($\&$l.ver, rVer, rVer)}
\EndIf

\State \textit{oldHeight = n.height}
\State \textit{newHeight = 1 + max(l.height, r.height)}

\If{\textit{oldHeight == newHeight}} \Comment{height is correct}
    \If {\textit{n.ver == nVer} \textbf{and} \textit{(l == NIL} \textbf{or} \textit{l.ver == lVer)} \textbf{and} \textit{(r == NIL} \textbf{or} \textit{r.ver == rVer)}}
    \State \textbf{return} \textit{UNNECESSARY}
    \Else\ \textbf{return} \textit{FAILURE}
    \EndIf
\EndIf

\State \textit{add($\&$n.height, oldHeight, newHeight)}
\State \textit{add($\&$n.ver, nVer, nVer + 2)}

\IIf {\textit{vexec()}} \textbf{return} \textit{SUCCESS}
\State \textbf{return} \textit{FAILURE}
\end{algorithmic}
\end{algorithm}

\begin{figure}[t]
    \includegraphics[width=0.45\textwidth]{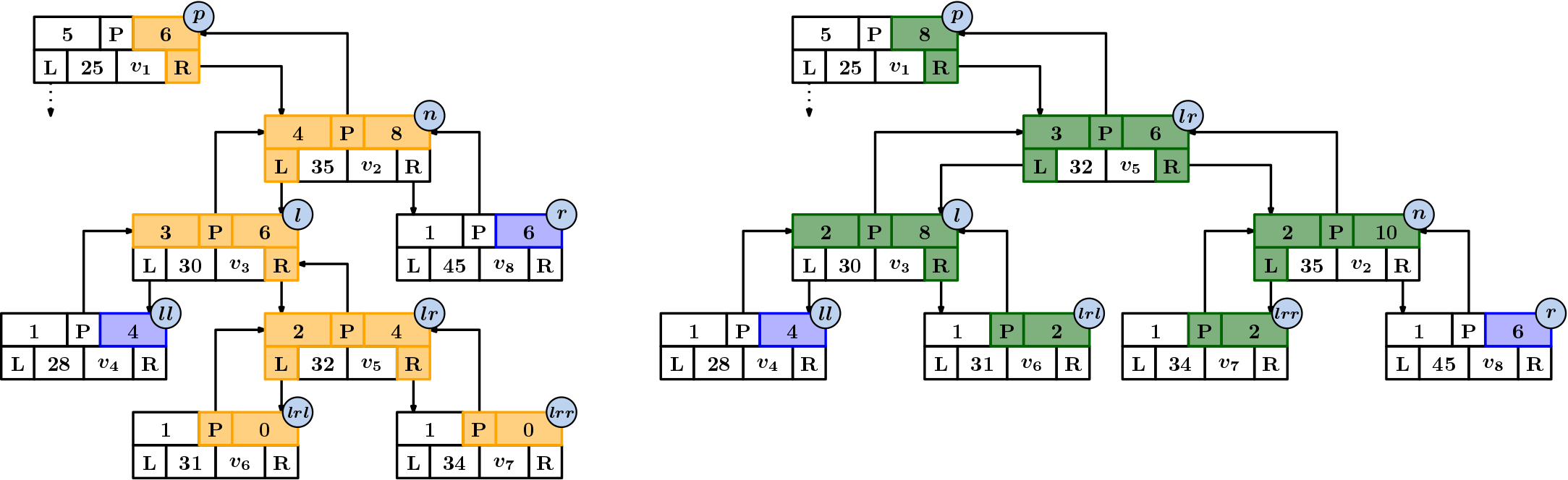}
\centering
\caption{Left Right Rotation BST}
\label{fig:rlrbst}
\end{figure}

\begin{algorithm}

\caption{rotateLeftRight(\textit{p, pVer, n, nVer, l, lVer, r, rVer, lr, lrVer})}\label{rotateLeftRight}
\begin{algorithmic}[1]
\IIf{\textit{p.right = n}} \textit{add($\&$p.right, n, lr)}
\ElsIIf{\textit{p.left = n}} \textit{add($\&$p.left, n, lr)}
\ElseI\ \textbf{return} \textit{RETRY}

\State \textit{lrl =  lr.left}; \textit{lrlHeight =  0};
\If{\textit{lrl $\neq$  NIL}}
\State \textit{lrlVer =  visit(lr)}; \textit{lrlHeight =  lrl.height};
\IIf{\textit{lrlVer \& 1}} \textbf{return} \textit{RETRY}
\State \textit{add($\&$lrl.parent, lr, l,}
\State \textit{add($\&$lrl.ver, lrlVer) lrlVer + 2)}
\EndIf

\State \textit{lrr =  lr.right}; \textit{lrrHeight =  0};
\If{\textit{lrr $\neq$  NIL}}
\State \textit{lrrVer =  visit(llr)}; \textit{lrrHeight =  lrr.height};
\IIf{\textit{lrrVer \& 1}} \textbf{return} \textit{RETRY}
\State \textit{add($\&$lrr.parent, lr, n,}
\State \textit{add($\&$lrr.ver, lrrVer) lrrVer + 2)}
\EndIf

\State \textit{rHeight =  0}
\If{\textit{r $\neq$  NIL}}
\State \textit{rVer =  visit(r)}; \textit{rHeight =  r.height};
\EndIf

\State \textit{ll =  l.left}; \textit{llHeight =  0};
\If{\textit{ll $\neq$  NIL}}
\State \textit{llVer =  visit(ll)}; \textit{llHeight =  ll.height};
\IIf{\textit{llVer \& 1}} \textbf{return} \textit{RETRY}
\EndIf

\State \textit{oldNHeight =  n.height} \Comment{\textbf{Same changes as sequential AVL}}
\State \textit{oldLHeight =  l.height}
\State \textit{oldLRHeight =  lr.height}

\State \textit{newNHeight =  1 + max(lrrHeight, rHeight)}
\State \textit{newLHeight =  1 + max(llHeight, lrlHeight)}
\State \textit{newLRHeight =  1 + max(newNHeight, newLHeight)}

\State \textit{add($\&$lr.parent, l, p)}
\State \textit{add($\&$lr.left, lrl, l)}
\State \textit{add($\&$l.parent, n, lr)}
\State \textit{add($\&$lr.right, lrr, n)}
\State \textit{add($\&$n.parent, p, lr)}
\State \textit{add($\&$l.right, lr, lrl)}
\State \textit{add($\&$n.left, l, lrr)}
\State \textit{add($\&$n.height, oldNHeight, newNHeight)}
\State \textit{add($\&$l.height, oldLHeight, newLHeight)}
\State \textit{add($\&$lr.height, oldLRHeight, newLRHeight)}
\State \textit{add($\&$lr.ver, lrVer, lrVer + 2)} \Comment{\textbf{New changes in PathCAS}}
\State \textit{add($\&$p.ver, pVer, pVer + 2)}
\State \textit{add($\&$n.ver, nVer, nVer + 2)}
\State \textit{add($\&$l.ver, lVer, lVer + 2)}
\IIf {\textit{vexec()}} \textbf{return} \textit{SUCCESS}
\State \textbf{return} RETRY

\end{algorithmic}
\end{algorithm}

\begin{algorithm}[t]
\caption{rebalance(\textit{n})}\label{rebalance}
\begin{algorithmic}[1]
\While{\textit{n $\neq$ minRoot}}
\State \textit{start()}
\State \textit{nV = node.ver}
\IIf {\textit{nV \& 1}} \textbf{return}
\State \textit{p = n.parent}
\State \textit{pV = visit(p)}
\IIf {\textit{pV \& 1}} \textbf{continue} 
\State \textit{l = n.l}
\IIf {\textit{l $\neq$ NIL}} \textit{lV = visit(l)}

\State \textit{r = n.right}
\IIf {\textit{r $\neq$ NIL}} \textit{rV = r.ver }
\IIf{\textit{lV \& 1} \textbf{or} \textit{rV \& 1}} \textbf{continue} 

\State \textit{lh = (left == NIL ? 0 : l.height)}
\State \textit{rh = (right == NIL ? 0 : r.height)}
\State \textit{nBalance = lh - rh}
\If {\textit{nBalance $\geq$ 2}}
\State \textit{ll = l.left}
\IIf {\textit{ll $\neq$ NIL}} \textit{llV = visit(ll)}

\State \textit{lr = l.right}
\IIf {\textit{lr $\neq$ NIL}} \textit{lrV = visit(lr)} 
\IIf{\textit{llV \& 1} \textbf{or} \textit{lrV \& 1}} \textbf{continue} 

\State \textit{llh = (ll == NIL ? 0 : ll.height)}
\State \textit{lrh = (lr == NIL ? 0 : lr.height)}

\State \textit{lBalance = llh - lrh}
\If {\textit{lBalance $<$ 0}}
\If {\textit{rotateLeftRight(p, pV, n, nV, l, lV, r, rV, lr, lrV)}}
    \State \textit{rebalance(n); rebalance(l); rebalance(lr);}
    \State \textit{n = p}
    \EndIf
\ElsIf {\textit{rotateRight(p, pV, n, nV, l, lV,  r, rV)}}
    \State \textit{rebalance(n); rebalance(l);}
    \State \textit{n = p}
\EndIf
\ElsIf {\textit{nBalance $\leq$ -2}}
\State \{...\} \Comment {Reverse case}
\Else 
\State \textit{res = fixHeight(n, nV)}
\IIf {\textit{res == FAILURE}} \textbf{continue}
\ElsIIf{\textit{res == SUCCESS}}  \textit{n = n.parent}
\ElseI\ \textbf{return}
\EndIf

\EndWhile
\end{algorithmic}
\end{algorithm}

\begin{algorithm}[t]

\caption{rotateRight(\textit{p, pV, n, nV, l, lV,  r, rV})}\label{rotateRight}
\begin{algorithmic}[1]

\IIf{\textit{p.right = n}} 
    \textit{add($\&$p.right, n, l)}
\ElsIIf{\textit{p.left = n}} 
    \textit{add($\&$p.left, n, l)}
\ElseI\ \textbf{return} \textit{false}

\State \textit{lr =  l.right}
\State \textit{lrH =  0}
\If{\textit{lr $\neq$  NIL}}
\State \textit{lrV =  visit(lr);} \textit{lrH =  lr.height;}
\IIf{\textit{lrV \& 1}} \textbf{return} \textit{false}
\State \textit{add($\&$lr.parent, l, n)}
\EndIf

\State \textit{ll =  l.left}
\State \textit{llh =  0}
\If{\textit{ll $\neq$  NIL}}
\State \textit{llV =  visit(ll);} \textit{llh =  ll.height;}
\IIf{\textit{llV \& 1}} \textbf{return} \textit{false}

\EndIf

\State \textit{rh =  0}
\If{\textit{r $\neq$  NIL}}
\State \textit{rV =  r.ver;} \textit{rH =  r.height;}

\State \textit{add($\&$r.ver, rV, rV)}
\EndIf

\State \textit{oldNH =  n.height} \Comment{\textbf{Same changes as sequential AVL}}
\State \textit{oldLH =  l.height}
\State \textit{newNH =  1 + max(lrH, rH)}
\State \textit{newLH =  1 + max(llH, newNH)}

\State \textit{add($\&$l.parent, n, p,}
\State \textit{add($\&$n.left, l, lr)}
\State \textit{add($\&$l.right, lr, n)}
\State \textit{add($\&$n.parent, p, l) }
\State \textit{add($\&$n.height, oldNH, newNH)}
\State \textit{add($\&$l.height, oldLH, newLH)}
\State \textit{add($\&$p.ver, pV, pV + 2)} \Comment{\textbf{New changes in PathCAS}}
\State \textit{add($\&$n.ver, nV, nV + 2)}
\State \textit{add($\&$l.ver, lV, lV + 2)}

\IIf{\textit{vexec()}} \textbf{return} \textit{true} 
\State \textbf{return} \textit{false}
\end{algorithmic}
\end{algorithm}

\section{Unbalanced Tree Correctness Proof}  \label{a:avl-proof-correctness}
\begin{defn}
The \textbf{search path} to a key $k$ is the path an atomic search($k$) would traverse.
\end{defn}

\begin{defn}
A node is \textbf{in the data structure} if it is reachable from maxRoot. 
\end{defn}

\begin{lemma} \label{lemma:marked}
No successful PathCAS ever modifies the fields of marked nodes.
\end{lemma}{}
\begin{proof}
This follows from inspection of the code. Before reading any other field of a node, the version number of this node is read. If this node is to be involved in the PathCAS, this version number is checked to ensure the node is not marked, and the version number is incremented as part of the PathCAS. If the node is marked, the operation is retried. If the node is not marked when it is checked, but is marked between this check and the PathCAS execution, the PathCAS will fail (the marking of a node directly involves the changing of its version number). 
\end{proof}{}
\begin{lemma} \label{lem:avl}
Our implementation of an unbalanced binary search tree satisfies the following claims. 
\begin{enumerate}
\item The node \textit{maxRoot} always has \textit{minRoot} as its left child, and NIL as its right child. The node \textit{minRoot} has NIL as its left child. (Remark: the right child points to the rest of the tree.) \label{lem:avl:root}
\item Consider any search, where \textit{$r_{1}...r_{l}$} is the sequence of nodes visited by it so far. For each $r_{i}$ in $r_{1}...r_{l}$, there is a time during the search when $r_{i}$ is in the data structure. \label{lem:avl:search-nodes}
\item Consider an invocation $I$ of \textit{PathCAS::validate()} in an \textit{insertIfAbsent(key, val)}, \textit{delete(key)} or \textit{contains(key)}. If $I$ returns true, then \textit{path} is the search path to $k$ just before $I$.\label{lem:avl:validate}
\item
    \begin{enumerate}[leftmargin=2\parindent]
        \item The tree rooted at the right child of \textit{minRoot} is always a valid binary searh tree. \label{lem:avl:is-relaxed-avl}
        \item Any \textit{insertIfAbsent} or \textit{delete} operation that performs a successful PathCAS returns the same value it would if it were performed atomically at its linearization point (the successful PathCAS). \label{lem:avl:lin-kcas}
        \item Any \textit{insertIfAbsent} or \textit{delete} operation that terminates without performing a successful PathCAS returns the same value it would if it were performed atomically at its linearization point. \label{lem:avl:lin-no-kcas}
    \end{enumerate}

\end{enumerate}
\end{lemma}
\begin{proof}
Consider an arbitrary execution $E$.
We prove these claims together by induction on the sequence of steps $s_1, s_2, ...$ (which can be shared memory reads, atomic KCASRead operations, or atomic KCAS operations) in $E$.

Base case: Before any PathCAS is successful, the tree is in its initial state where two nodes exist: \textit{minRoot} and \textit{maxRoot}. \textit{maxRoot} has the key +$\infty$, its left child is \textit{minRoot}, and its right child is NIL. \textit{minRoot} has the key -$\infty$, and both its children are NIL. 

Inductive step: suppose the claims all hold before step $s$. We prove they hold after step $s$.

\noindent \textsc{Claim \ref{lem:avl:root}}. This configuration is the initial state of the tree, except the right child of \textit{minRoot} could no longer be NIL. No operation can modify this state. \textit{minRoot} and \textit{maxRoot} contain keys that will never be part of any operation: no operation will search for, attempt to remove, or attempt to insert these keys. Additionally, these nodes will never be rebalanced: it is explicit in the code that rebalancing stops when it reaches the \textit{minRoot}, which also means \textit{maxRoot} will also never be rebalanced. \\

\noindent \textsc{Claim~\ref{lem:avl:search-nodes}}. The only step that can impact this claim is a KCASRead in search that traverses to a new node, by reading a pointer from $r_{l}$ to some $r_{l+1}$. (This pointer is then added it to \textit{path[]}, the sequence of nodes visited so far.)
So, $s$ is a KCASRead reading the left or right pointers of a node to move to another. %

By the inductive hypothesis $r_{l}$ was in the data structure at some time before $s$ during the search. If $r_{l}$ still points to $r_{l+1}$ at step $s$, then since \textit{$r_{l}$} is in the data structure and points to $r_{l+1}$, so is $r_{l+1}$.
Otherwise, \textit{$r_{l}$} was deleted before $s$ during the search, and right before this deletion $r_{l}$ pointed to $r_{l+1}$ (Lemma \ref{lemma:marked}).
Therefore, $r_{l+1}$ was in the data structure at the time of this deletion (which was during the search). \\

\noindent \textsc{Claim \ref{lem:avl:validate}}. Only \textit{successful} PathCAS operations can affect this claim, as reads do not modify the structure of the tree. A search reads and stores the version numbers of all nodes encountered during a search. If any of these nodes change between when they were read as part of the search and when they are validated as part of \textit{validatePath}, \textit{validatePath} will return false. If \textit{validatePath} returns true, it is guaranteed that no modification occurred to \textbf{any} node along the path between when it was read and when it was validated as part of \textit{validatePath}. At the time \textit{t} just before \textit{I}, all nodes were read as part of the search but are yet to be validated. Hence, if \textit{I} returns true, all the nodes still had the same version numbers read as part of the search at \textit{t} and this path was an \textit{atomic snapshot} of the search path to \textit{k}, just before the invocation of \textit{I} (before any nodes were validated).
\\

\noindent \textsc{Claim \ref{lem:avl:is-relaxed-avl}}. Only \textit{successful} PathCAS operations that change the layout of the tree can affect this claim. KCAS operations are performed by \textit{insertIfAbsent} and \textit{delete}.%
We proceed by cases.

\textit{Case 1}: Suppose \textit{s} is a successful KCAS at line~\ref{line:insert:kcas} of \textit{insertIfAbsent}.
This KCAS will only be completed if a successful path validation occurred.
Let \textit{t} be the time just before the successful invocation of \textit{validateDesc} started.
Claim \ref{lem:avl:validate} means that this path \textbf{is} the search path to \textit{k} at time \textit{t}, and no node has changed since they were read during the operation.
The node we modify to insert this new node is the final node on the search path at time \textit{t}, denoted by $\textit{last}_t$. 
In order for the search path to \textit{k} to change (meaning $\textit{last}_t$ to no longer the correct parent for a new node containing \textit{k}) a concurrent operation must modify at least one node along the search path to \textit{k}.
However, if one of these nodes were to be modified, the KCAS would fail: validation would not have succeeded.

\textit{Case 2:} Suppose \textit{s} is a successful vexec at line \ref{line:deleteTwo:kcas} of \textit{delete} (we omit the PathCAS of the single or no child delete as it is strictly easier). 
This proof is similar to the previous case where \textit{s} is a PathCAS within \textit{insertIfAbsent}.
Let \textit{t} be the time just before the successful invocation of \textit{validatePath} started.
In \textit{delete}, search locates a node \textit{n} that contains the key \textit{k}, and the version number of \textit{n} is added to the PathCAS to be incremented. 
\textit{getSuccessor} locates the successor of \textit{k} by traversing the tree from \textit{n} and validating the path taken. 
Because the the invocation of \textit{validateDesc()} in \textit{vexec()} must return \textit{true} for this operation to be successful, the node returned by \textit{getSuccessor} is the successor of \textit{k} at \textit{t}, by traversing the tree from \textit{n} and validating both the path taken to \textit{k} and the path taken to the successor. 
In order for this operation to be invalid, either \textit{n} must not contain \textit{k}, or the successor returned by \textit{getSuccessor} must be incorrect.
However, for this to be the case, a concurrent operation must modify at least \textit{n} or a node along the path from \textit{n} to the successor found.
\textit{n} cannot be modified, as it is protected by the PathCAS.
If one of the nodes along the path is modified, the PathCAS would fail: validation would not have succeeded.

\noindent \textsc{Claim \ref{lem:avl:lin-kcas}}. In Claim \ref{lem:avl:is-relaxed-avl} we actually proved that all operations that execute a successful PathCAS are atomic at time \textit{s}, which is the time when the PathCAS was executed. \\

\noindent \textsc{Claim \ref{lem:avl:lin-no-kcas}}. \textit{insertIfAbsent} only returns without performing a PathCAS if \textit{search} locates a node that has the key \textit{k}. From Claim \ref{lem:avl:search-nodes} the key \textit{k} that we are trying to insert was in the data structure at some time \textit{t} during the search. \textit{t} is during \textit{insertIfAbsent}, so we linearize this operation at \textit{t}, returning false.

\textit{delete} only returns without performing a PathCAS if \textit{search} does not locate a node that has the key \textit{k}. By Claim \ref{lem:avl:validate}, the key \textit{k} that we are trying to remove was not in the data structure at some time \textit{t} during the search. \textit{t} is during \textit{delete}, so we linearize this operation at \textit{t}, returning false. \\

The linearization points for the operations are as follows:
\begin{itemize}
    \item \textit{insertIfAbsent}
        \begin{itemize}
        \item returning \textit{true}: at the successful PathCAS at line \ref{line:insert:kcas} of \textit{insertIfAbsent}
        \item returning \textit{false}: the time \textit{t} during search where \textit{k} was in the data structure, from Claim \ref{lem:avl:search-nodes}
        \end{itemize}
    \item \textit{delete}
      \begin{itemize}
        \item returning \textit{true}: at the successful PathCAS of at line \ref{line:delete:two-return-success} of \textit{delete} or at line successful PathCAS for the single or no child deletion in \textit{delete} (whichever is used)
        \item returning \textit{false}: the time \textit{t} just before the invocation of \textit{validatePath} in search, from Claim \ref{lem:avl:validate}
        \end{itemize}
    \item \textit{contains} 
         \begin{itemize}
        \item returning \textit{true}: the time \textit{t} during \textit{search} where \textit{k} was in the data structure, from Claim \ref{lem:avl:search-nodes}
        \item returning \textit{false}: the time \textit{t} just before the invocation of \textit{validatePath} in search, from Claim \ref{lem:avl:validate}
        \end{itemize}
\end{itemize}

\end{proof}

\begin{theorem}
Our binary search tree implements a linearizable dictionary.
\end{theorem}

\begin{proof}
Lemma \ref{lem:avl} proves that the all the operations on the tree are atomic and do not violate any relaxed binary search tree properties. Searches are equivalent to an atomic search at some time \textit{t} during \textit{search}. \textit{contains} is simply linearized at this time \textit{t} during \textit{search}. 
\end{proof}

\section{Motivation: The Hardness of Proving Fine-Grained Data Structures Correct}
\subsection{Overview} \label{a:motivation}
Correctness bugs exist even in peer reviewed works, going undetected for long periods of time. In this section we provide an overview of previously discovered issues in other work, and one new bug we found during our investigation. 

Michael and Scott \cite{MSTech} discovered two race conditions in the lock-free concurrent queue by Valois \cite{ValoisThesis} which lead to incorrect memory reclamation. These issues could corrupt the data structure in two ways: by freeing the same nodes multiple times, or freeing nodes that are still reachable.  The memory reclamation scheme attempts to avoid ABA problems by relying on reference counters in data structure nodes, where threads increase the reference counter of a node when they read a pointer to it, and decrement the counter when they no longer require it. The thread that moves this reference counter to 0 after a node has been removed will reclaim this node. There exists a time, however, between when a pointer to a node is read by a thread and that thread increments the reference counter of the node. During this time another thread could move the reference counter of the node to 0 and free it. The thread about to increment the reference counter will be unaware of this, will increment the reference counter of the node to 1, then potentially back to 0, resulting in a double free. A similar issue can also result in nodes being freed despite still being in the data structure. 

Shafiei \cite{ShafieiThesis} found an execution that dereferences a null pointer in the lock-free doubly-linked-list by Sundell and Tsigas \cite{TsigasDLL} by running the Java PathFinder (JPF) model checker on the implementation. She describes how the sheer complexity of the algorithm makes it very difficult to reason about its correctness, and that the proofs provided in the work are insufficient. This correctness issue was fixed in subsequent versions of the same work.

As part of this work, we discovered a new bug in the balanced BST by Drachsler et al. \cite{dana-tree}. This tree uses additional pointers in nodes to track the predecessor and successor of a node, which are used to recover from searches that end up in an incorrect location due to a concurrent rotation. An execution exists, however, wherein searches fail to find keys that are present during the entire search operation. We notified the authors of this work, who confirmed the bug via private communication and plan to release errata on the topic at a later date. A full explanation of the incorrect execution is in Appendix \ref{a:dana}.

There are two existing lock-free internal BSTs in the literature, due to Howley and Jones~\cite{howley} (HJ12) and Ramachandran et~al.~\cite{ram-tree} (RM14).
Publicly available implementations of both trees \textit{fail} our experimental validation checks, which test for consistency between the tree contents at the end of an experiment and the return values of all updates recorded throughout the experiment.
The HJ12 tree has no accompanying proofs, although the authors did a model checker on some very small trees.
We have been unable to resolve the correctness problems in RM14 with the authors, and it is unclear whether the algorithm is correct.
(The paper cites a proof in a technical report, but the technical report does not appear to be available.)
Both papers use highly intricate synchronization mechanisms: e.g., RM14 is described in \textbf{315 lines of pseudocode} \textit{excluding} code comments.
By contrast, our algorithm is arguably short and easy to reason about.
Each operation \textit{visits} every node that it accesses, before accessing any of the node's fields, and performs either a \textit{validate} or \textit{vexec} to establish a time when the entire operation takes effect atomically.

\subsection{Drachsler Tree Bug}\label{a:dana}
\subsubsection{Overview}

Consider the following tree structure: 
\begin{center}
    \includegraphics[width=.45\textwidth]{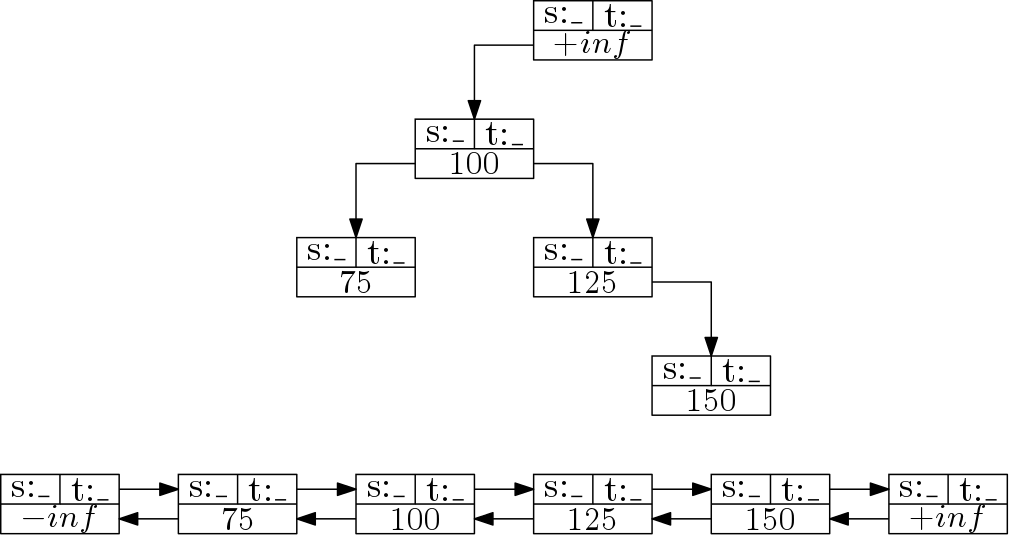}
\end{center}

This shows both the tree and logical ordering structure as shown in the paper. The bottom values are keys, and \textit{s} and \textit{t} fields represent the succLock and treeLocks respectively, a \_ represents no thread holds the lock, but this will be filled in when a thread holds the lock. Assume now that a thread \textit{p} inserts the value 175 to both the logical order and the tree, but has not rebalanced the tree as of yet (i.e. just before line 9 of algorithm 5 in the original paper) leaving the tree in a state as follows:
\begin{center}
    \includegraphics[width=.45\textwidth]{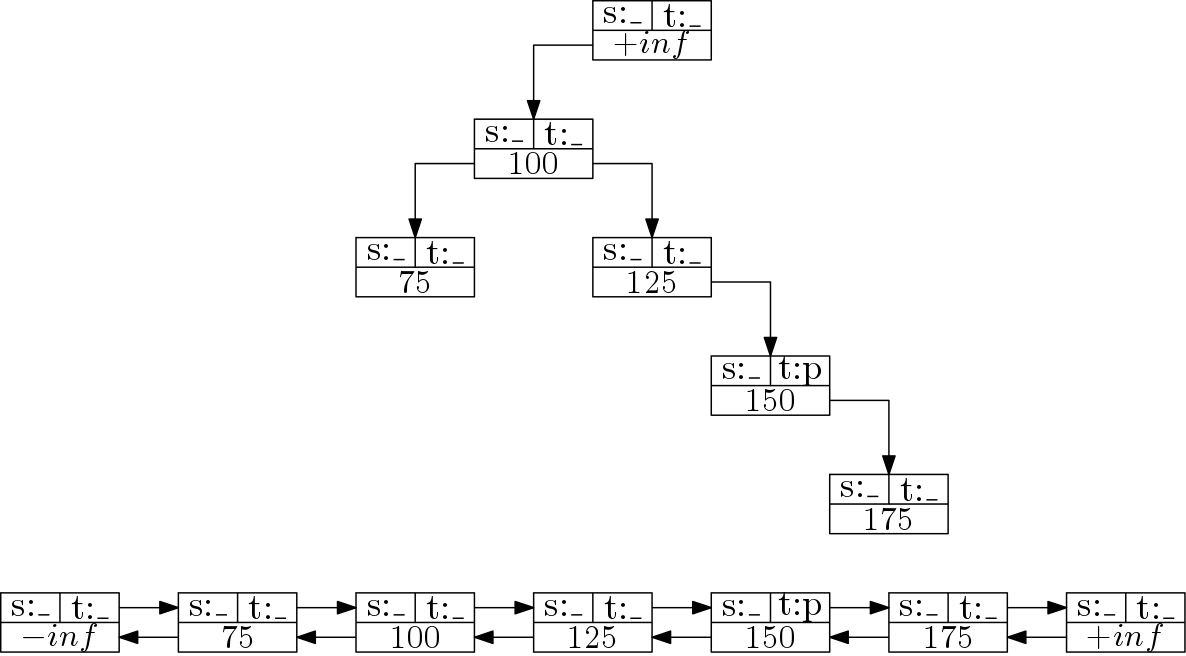}
\end{center}
Thread \textit{p} now goes to sleep. Note that the treeLock for the node 150 is held by thread \textit{p}, but the succLock is released as that operation is completed. A rotation should occur on the node 125. Before this rotation can occur, imagine a thread \textit{q} performs an insertion of 160 up until line 15 of algorithm 3. This can occur because the succLock of 150 has been released by p (only the treeLock is held) and the treeLock of 175 will be the one that is acquired for this operation. This will leave the tree in the following state:
\begin{center}
    \includegraphics[width=.45\textwidth]{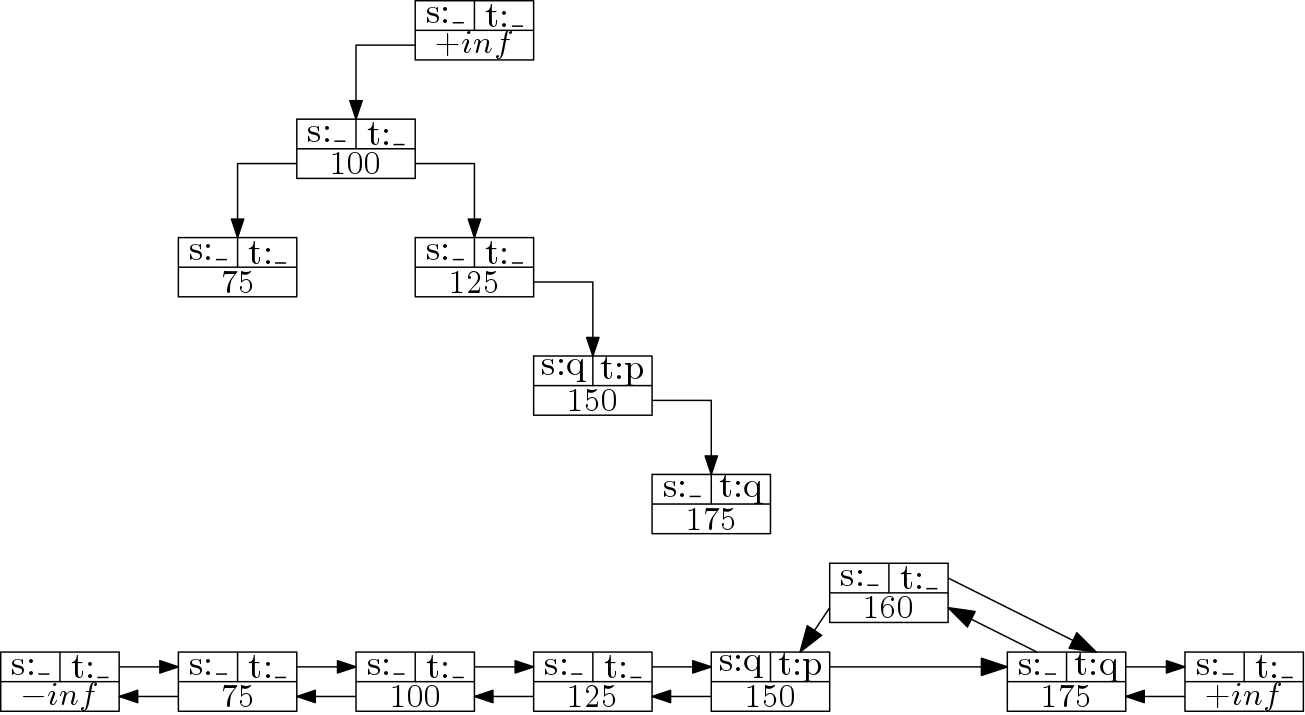}
\end{center}
Thread \textit{q} now goes to sleep. A third thread, \textit{r} does a contains operation for the key 160 to completion. It will search the tree as per algorithm 1, then will backtrack and find 160 as per algorithm 2. This operation will ignore all locks and return true. This implies the insertion of 160 must already have been linearized. Hence, the linearization point for insert described in the paper is incorrect (line 16, where pred.succ is updated) as the change is observed here, it must be before this line (line 15 is logical). \\

Next, a fourth thread, thread \textit{s}, does the same operation as thread \textit{r} (contains(160)), however it goes to sleep after it traverses to node 125.

\begin{center}
    \includegraphics[width=.45\textwidth]{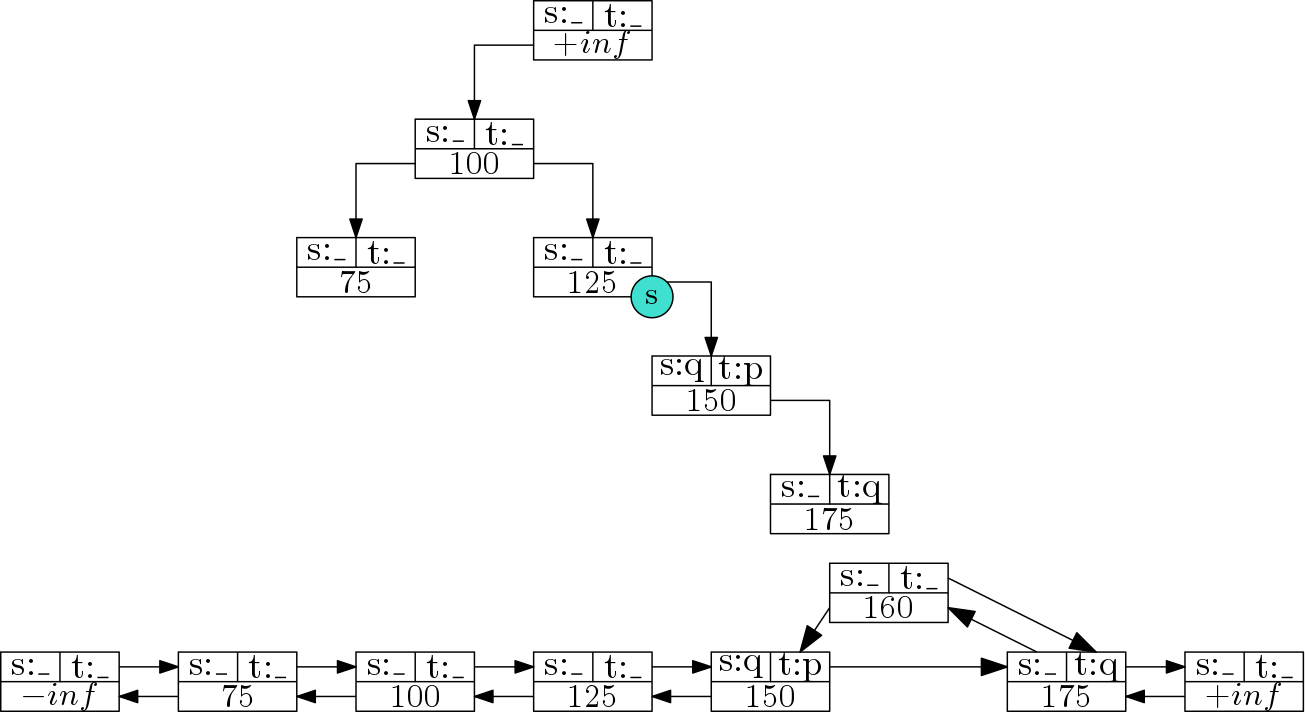}
\end{center}

Thread \textit{p} now wakes up and will start a rotation, this will result in the function rebalance(node, child) being called with the node containing 125 as the ``node'' argument, and the node containing 150 as the ``child'' argument to the \textit{rebalance(node, child)} function. Thread \textit{p} already has the treeLock for 150, and can freely acquire the treeLocks for 125 and 100 (the other ones required for this rotation). Note here that thread \textit{q} only holds the treeLock for 175. No threads hold succLocks at this point. \textit{p} can proceed (is not blocked by any currently held locks).

\begin{center}
    \includegraphics[width=.45\textwidth]{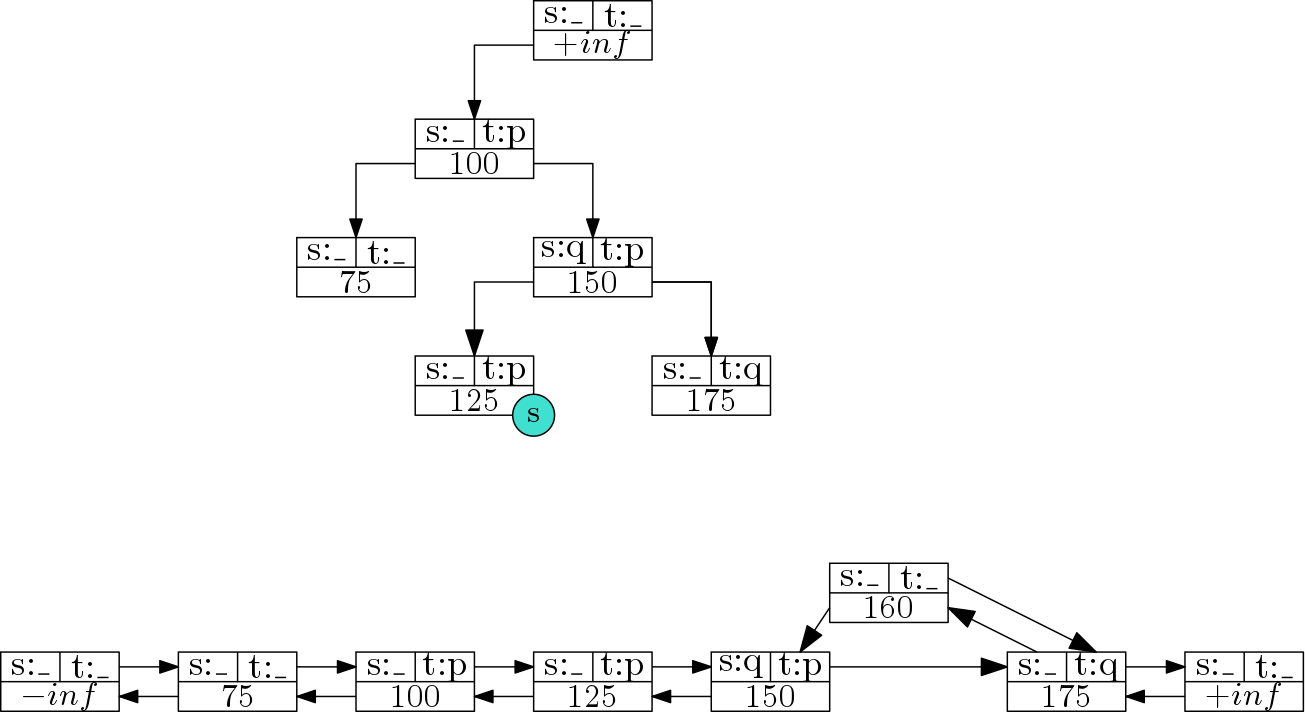}
\end{center}

Thread \textit{s} now wakes up after this rotation, realizes it's at a leaf node and completes its traversal of the tree. From there, it will attempt to follow the logical order to discover if it missed an update. Line 2 of algorithm 2 will not traverse the list (as \textit{node.key $>$ k} is false), but line 3 will traverse (as \textit{node.key $<$ k}) until the node containing 175, see that this key is not 160 and return false. This is invalid, as the linearization point for the insert(160) of thread \textit{q} has passed or else the result of the previous search for 160 must be incorrect. The following thread schedule represents the execution of the threads.  

\begin{center}
    \includegraphics[width=.45\textwidth]{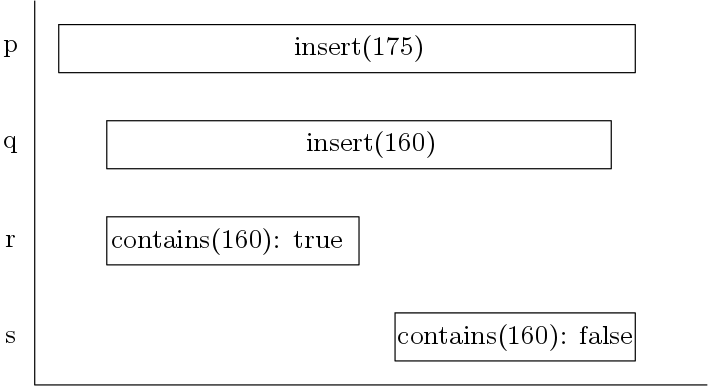}
\end{center}

The two contains cannot be linearized. 

\subsubsection{Solution: Search Direction Swap}

Say we reverse line 2 and 3 from the contains() operation (Algorithm 2 in the original work), making it look like the following:
\begin{algorithm}[H]
\caption{contains(key)}
\begin{algorithmic}[1]

\State node = search(key)
\State \textbf{while} \textit{node.key $<$ k} \textbf{do} node = node.succ
\State \textbf{while} \textit{node.key $>$ k} \textbf{do} node = node.pred
\State \textbf{return} \textit{(node.key == k \textbf{and} !(node.mark))}

\end{algorithmic}
\end{algorithm}

Now consider the example above, up to this state. 

\begin{center}
    \includegraphics[width=.45\textwidth]{images/dana-tree/5.png}
\end{center}

Thread s will now find the key 160, and the insert operation can be linearized. Actually, \textit{regardless of what node you end up at after this search, the contains operation will find 160}. If you are left of the partially inserted node, you will traverse succ pointers until you pass it, then follow a single pred pointer to the node in question. If you are to the right of the partially inserted node, you will follow no succ pointers, but follow pred pointers until you reach the node in question. \\ 

Consider the reverse case, before a right rotation occurs:

\begin{center}
    \includegraphics[width=.5\textwidth]{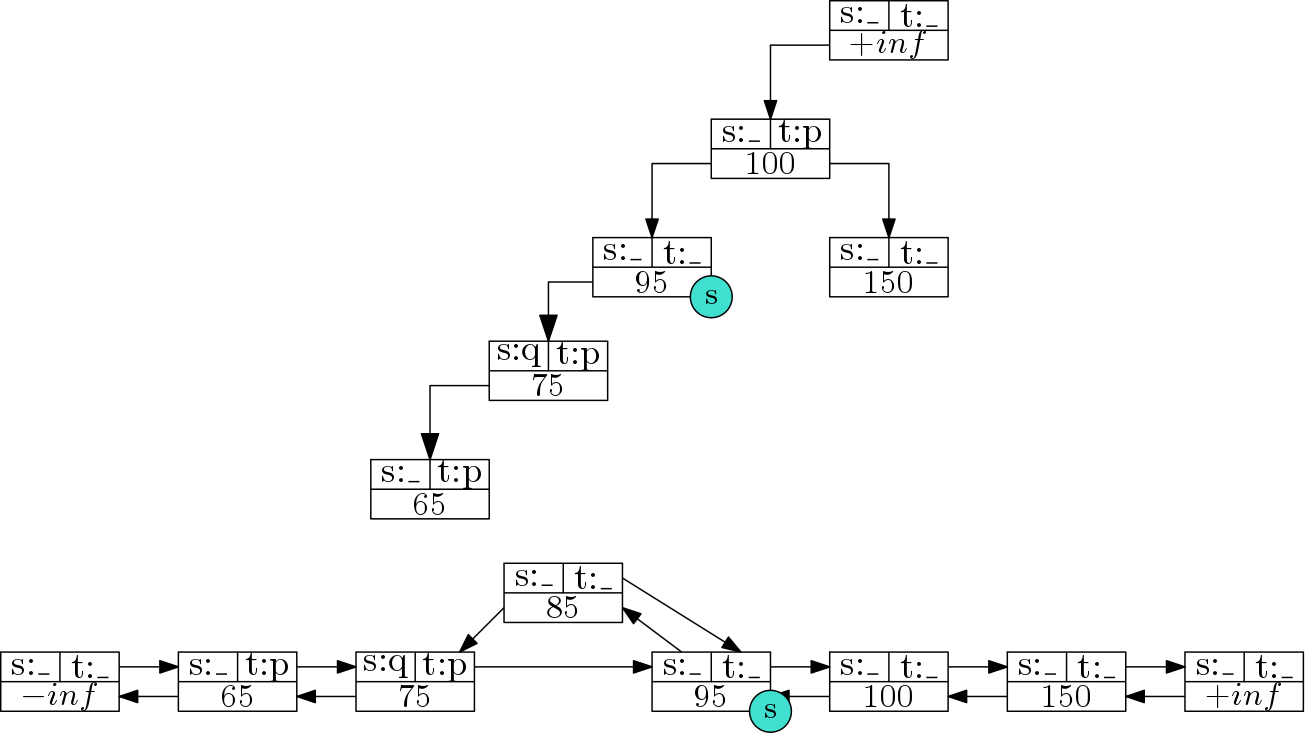}
\end{center}

After the rotation occurs:

\begin{center}
    \includegraphics[width=.45\textwidth]{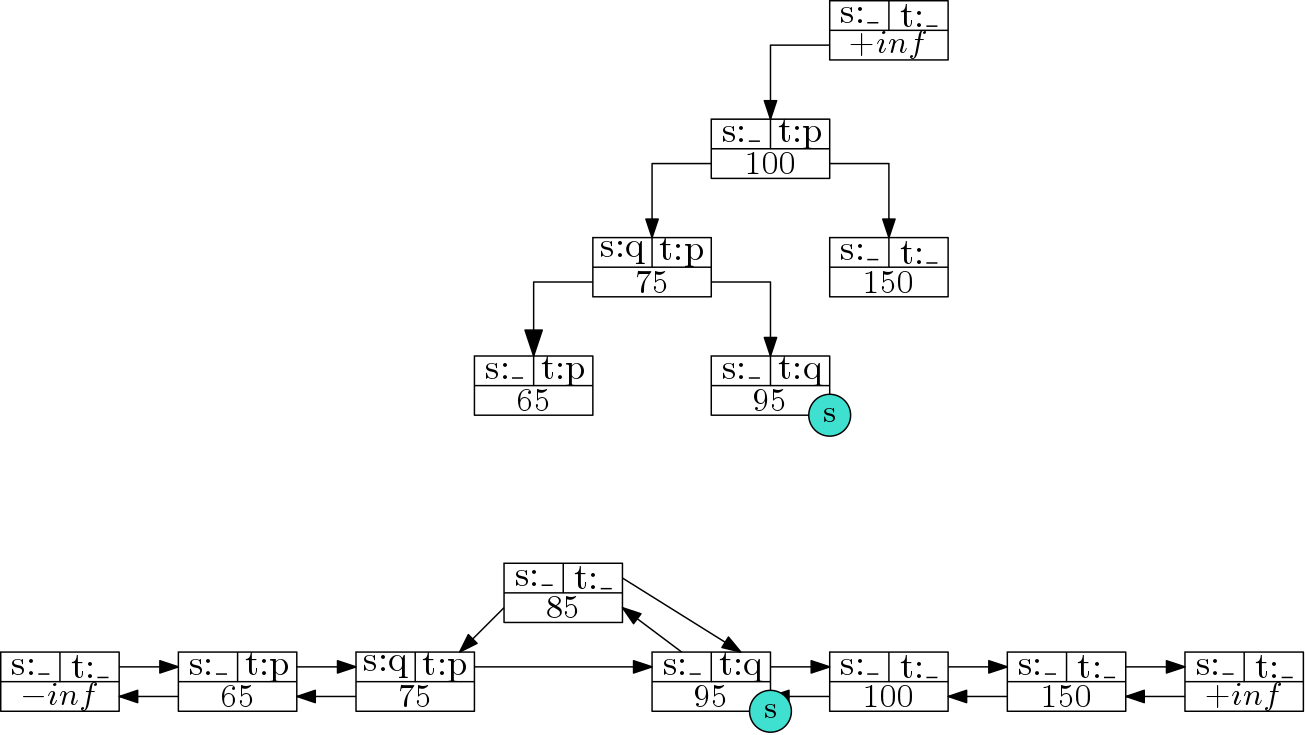}
\end{center}

Here, \textit{s} is searching for 85, arrives at 95 but then gets rotated down. If it searches right then left, it will still find 85. It is the same as the previous example, \textit{regardless of what node you end up at after this search, the contains operation will find 85.}

\section{Full Experimental Results} \label{a:experiments}
Our other system has four Intel Xeon Platinum 8160 CPUs, each with 24 cores and two hardware threads per core, for a total of 192 hardware threads.
Cores on the same CPU share a 33MB L3 cache. %
Threads are \textit{pinned} such that thread counts up to 48 run on one physical CPU, thread count 96 runs on two CPUs, and so on. %
Our code was compiled with GCC 7.4.0-1 using flag \texttt{-O3}.
We used \texttt{numactl} to interleave memory pages evenly across CPUs.
The fast allocator jemalloc 5.0.1-25 was for all algorithms, and memory was reclaimed using DEBRA, a fast epoch-based reclamation algorithm~\cite{DEBRA}.  %

\begin{figure}[H]

\begin{subfigure}{1\linewidth}
        \includegraphics[width=1\linewidth]{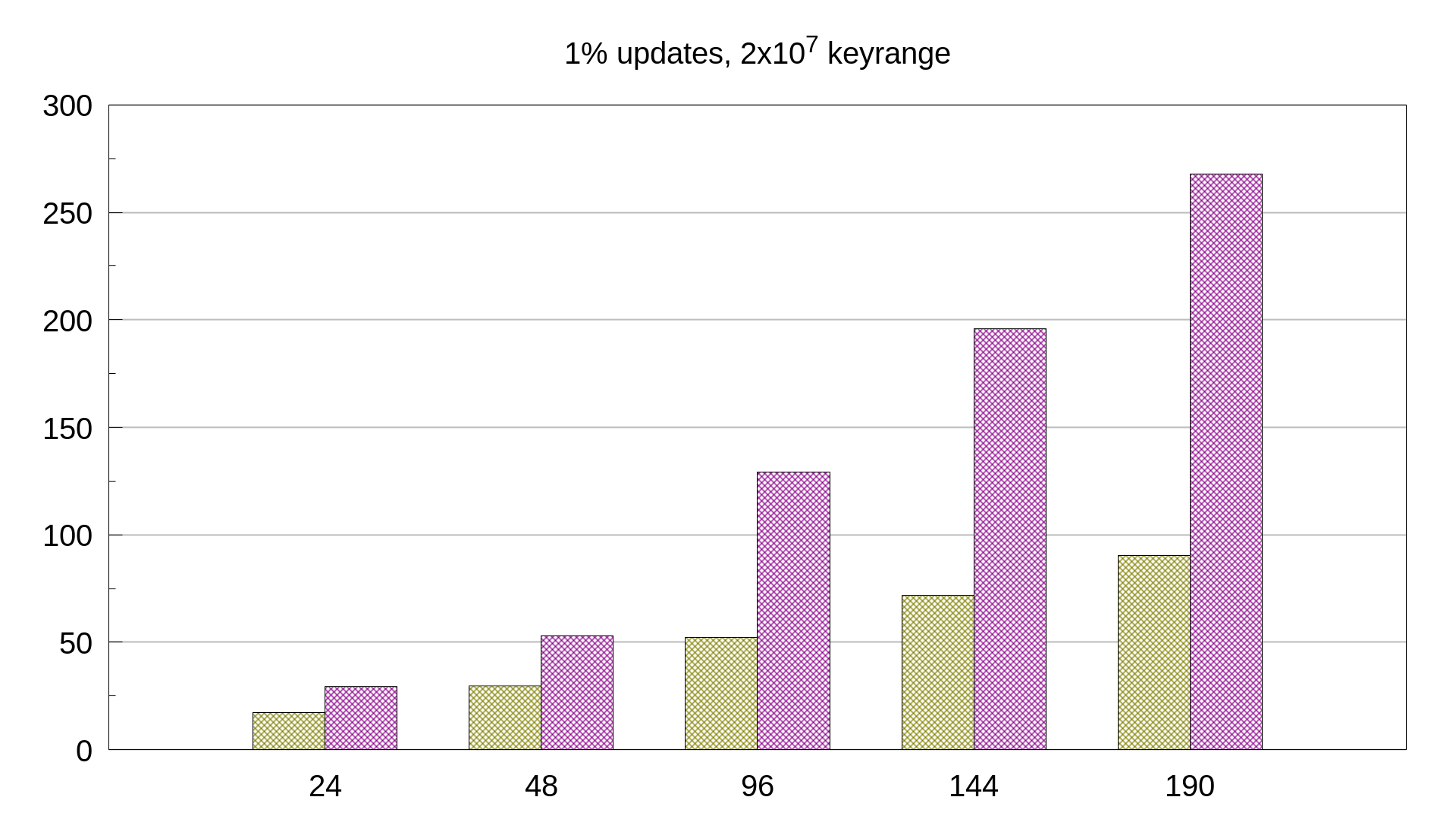}
\centering
\end{subfigure}
\begin{subfigure}{1\linewidth}
        \includegraphics[width=0.6\linewidth]{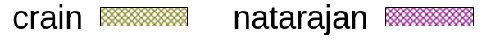}
\centering

\end{subfigure}
\vspace{-4mm}
\caption{\textbf{Throughput} evaluation of the speculation-friendly BST, in synchrobench.  Total throughput in millions of operations per second, full comparison, higher is better }
\label{fig:synchro}
\end{figure}

One other interesting comparison would be with \textit{crain}, the speculation-friendly BST from Synchrobench~\cite{synchrobench}, shown above in Figure \ref{fig:synchro}. While we attempted to port this data structure to our testing suite, due to time constraints we were unable to (as it would require porting the Elastic STM algorithm, integrating support for background threads that perform data structure maintenance asynchronously, and completely reimplementing its memory reclamation).

So, we obtained performance numbers for \textit{crain} using Synchrobench, instead of Setbench. We also obtained performance numbers for \textit{natarajan} in Synchrobench, to include some sort of comparison point between our results in Setbench and our results in Synchrobench. Although the resulting graph does \textit{not} allow for a \textit{rigorous} comparison between \textit{crain} and the other trees, we note that it is \textit{much} slower than \textit{natarajan}, which is in the middle of the pack in our Setbench experiments. And, we also note that \textit{crain} is being evaluated in an environment that is presumably most favourable to it, as \textit{crain} was developed by the authors of Synchrobench, and integrated therein by the authors. We aim to have properly integrated \textit{crain} into Setbench time for the final paper.

\begin{figure*}[H]
\newcommand{\plotwidth}{0.32\linewidth}
\newcommand{\legendwidth}{0.8\linewidth}
\begin{tabular}{p{\plotwidth}p{\plotwidth}p{\plotwidth}}
\centering\noindent\textbf{1\% updates} & \centering\noindent\textbf{10\% updates} & \centering\noindent\textbf{100\% updates}
\end{tabular}
\rotatebox{90}{\hspace{5mm}\textbf{10M keys}}
\includegraphics[width=\plotwidth]{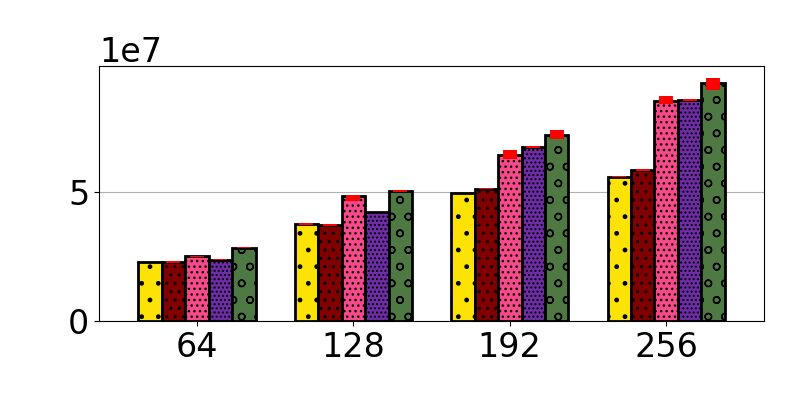}
\includegraphics[width=\plotwidth]{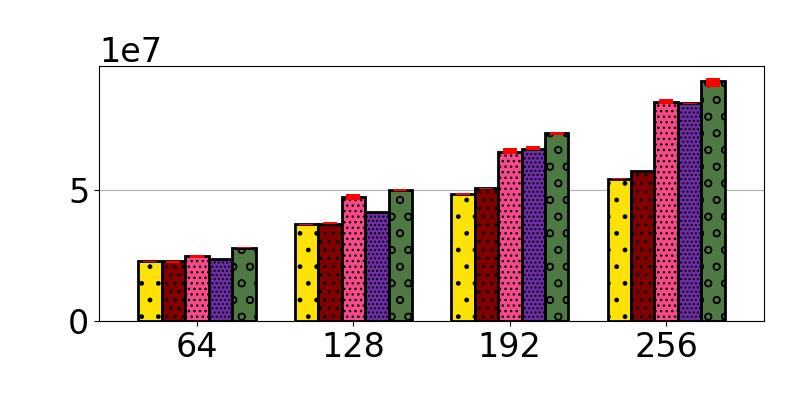}
\includegraphics[width=\plotwidth]{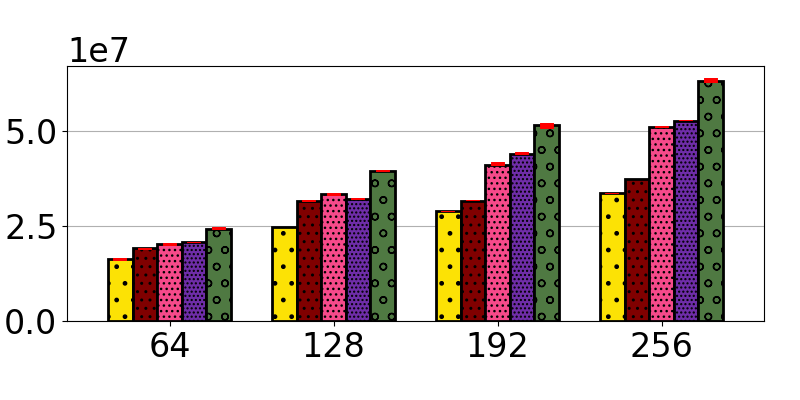}
\rotatebox{90}{\hspace{5mm}\textbf{1M keys}}
\includegraphics[width=\plotwidth]{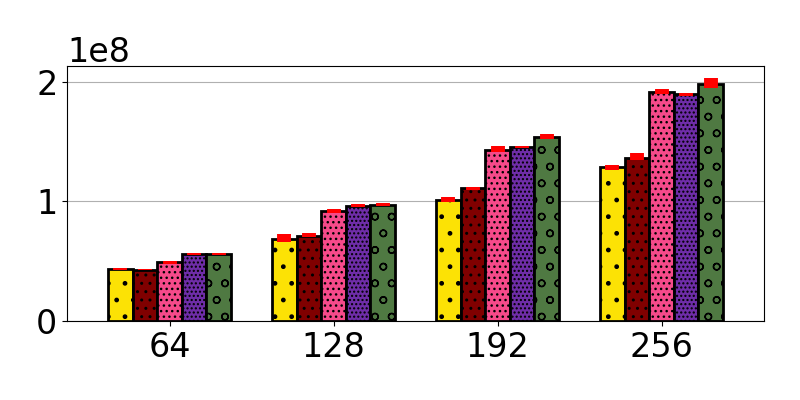}
\includegraphics[width=\plotwidth]{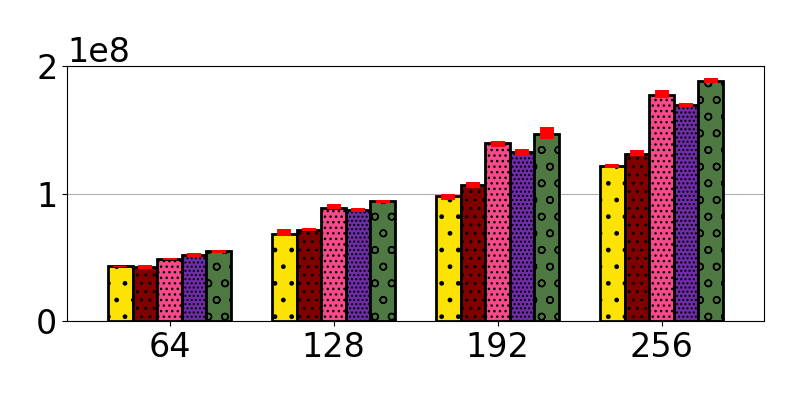}
\includegraphics[width=\plotwidth]{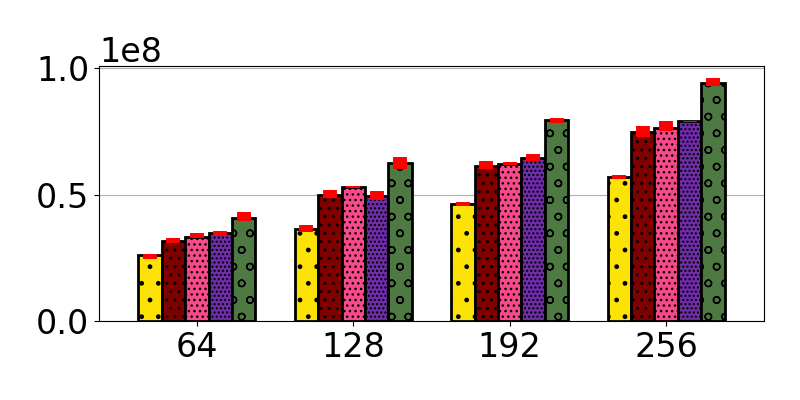}
\rotatebox{90}{\hspace{5mm}\textbf{100k keys}}
\includegraphics[width=\plotwidth]{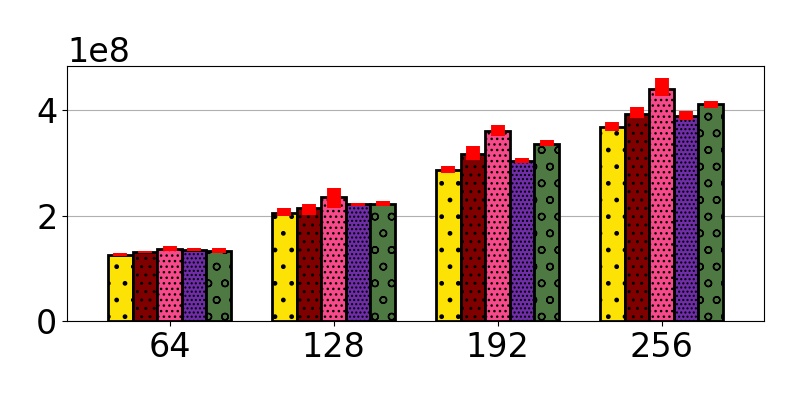}
\includegraphics[width=\plotwidth]{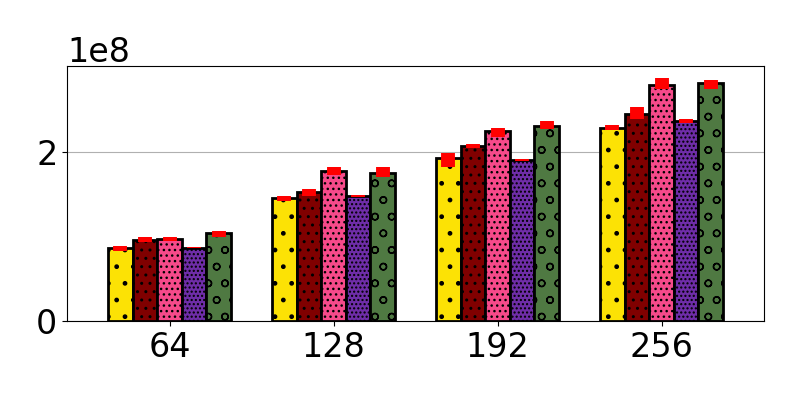}
\includegraphics[width=\plotwidth]{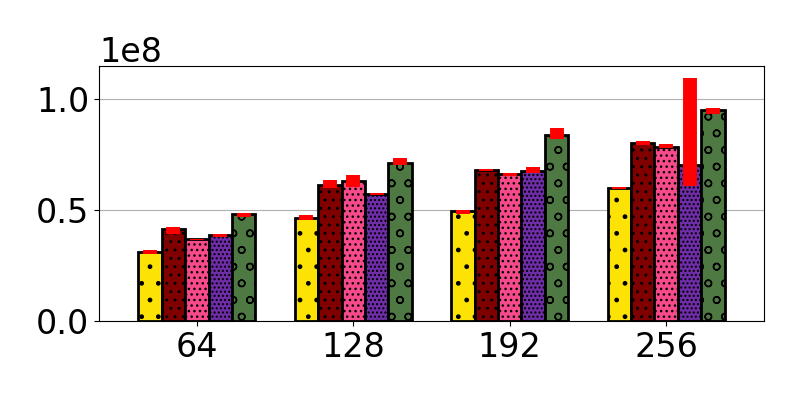}
\includegraphics[width=\legendwidth]{images/ppopp21/nasus_non_tm/bst_hc-legend.png}
\vspace{-4mm}
\caption{Comparing with \textbf{Handcrafted} \textbf{Unbalanced BSTs} on the \textbf{AMD} system. Operations per microsecond vs \# threads.}
\label{fig:nasus-non-tm-bst}
\end{figure*}

\begin{figure*}[H]
\newcommand{\plotwidth}{0.32\linewidth}
\newcommand{\legendwidth}{0.6\linewidth}
\begin{tabular}{p{\plotwidth}p{\plotwidth}p{\plotwidth}}
\centering\noindent\textbf{1\% updates} & \centering\noindent\textbf{10\% updates} & \centering\noindent\textbf{100\% updates}
\end{tabular}
\rotatebox{90}{\hspace{5mm}\textbf{10M keys}}
\includegraphics[width=\plotwidth]{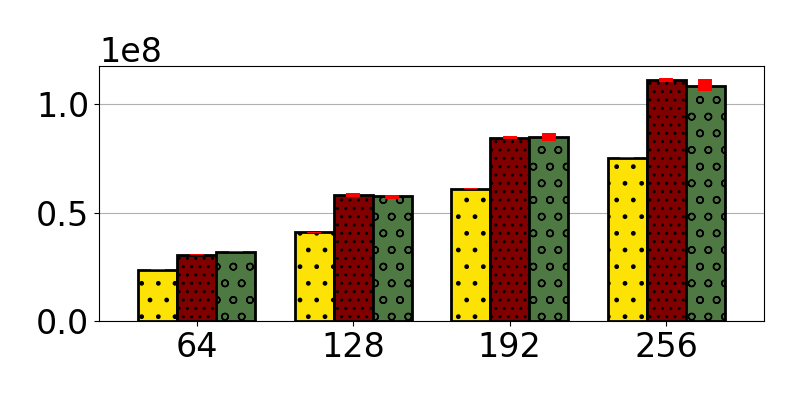}
\includegraphics[width=\plotwidth]{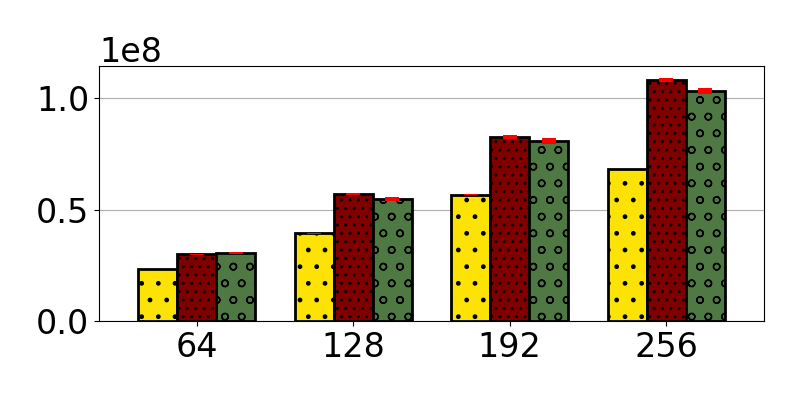}
\includegraphics[width=\plotwidth]{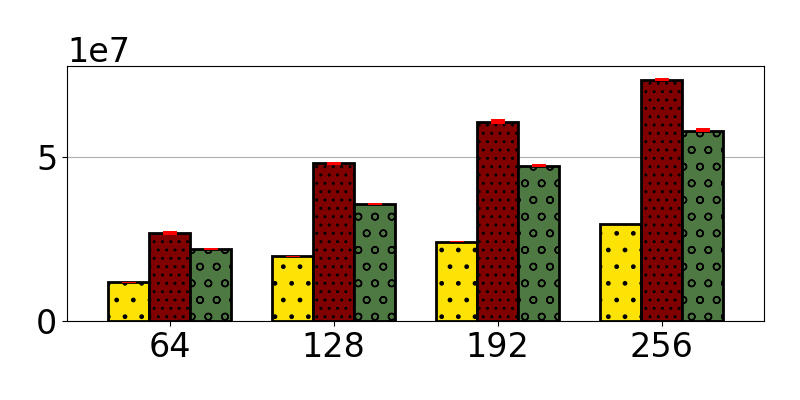}
\rotatebox{90}{\hspace{5mm}\textbf{1M keys}}
\includegraphics[width=\plotwidth]{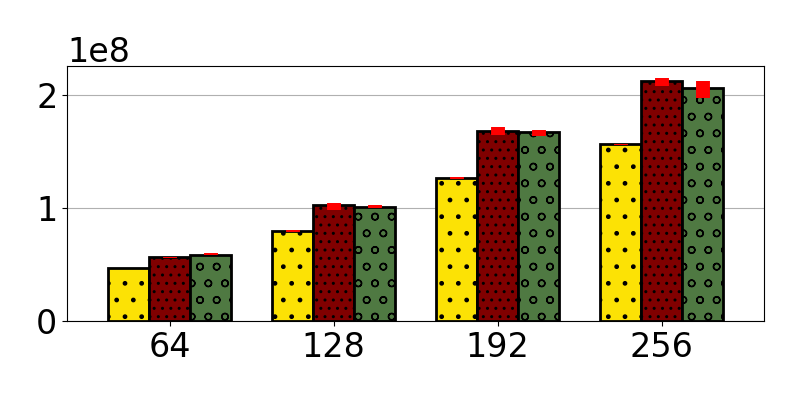}
\includegraphics[width=\plotwidth]{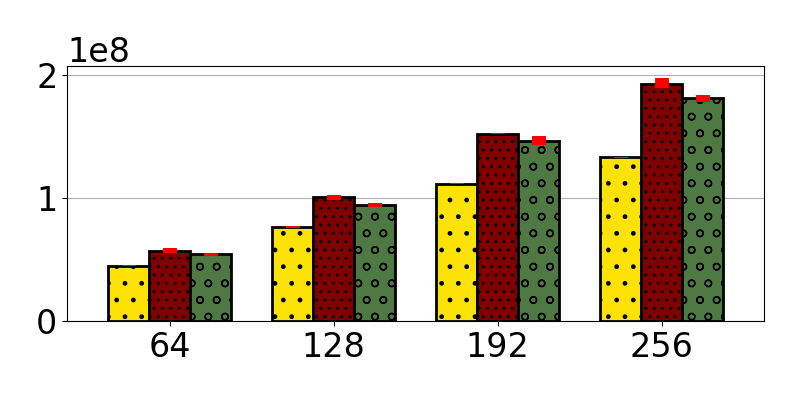}
\includegraphics[width=\plotwidth]{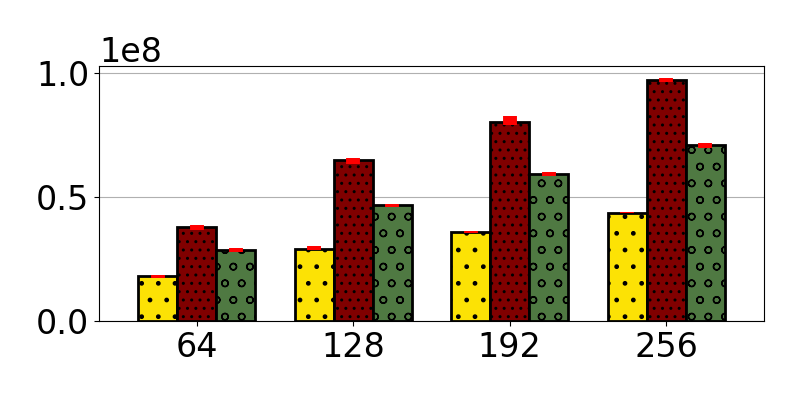}
\rotatebox{90}{\hspace{5mm}\textbf{100k keys}}
\includegraphics[width=\plotwidth]{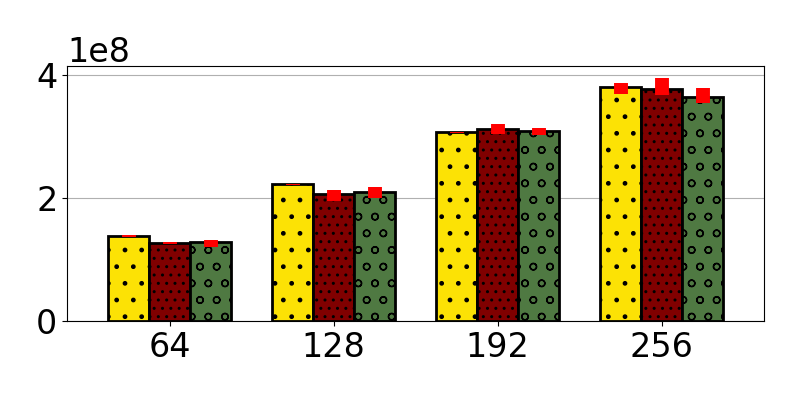}
\includegraphics[width=\plotwidth]{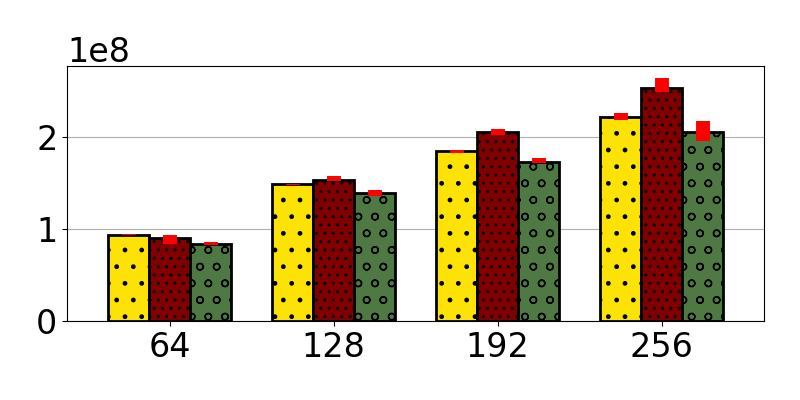}
\includegraphics[width=\plotwidth]{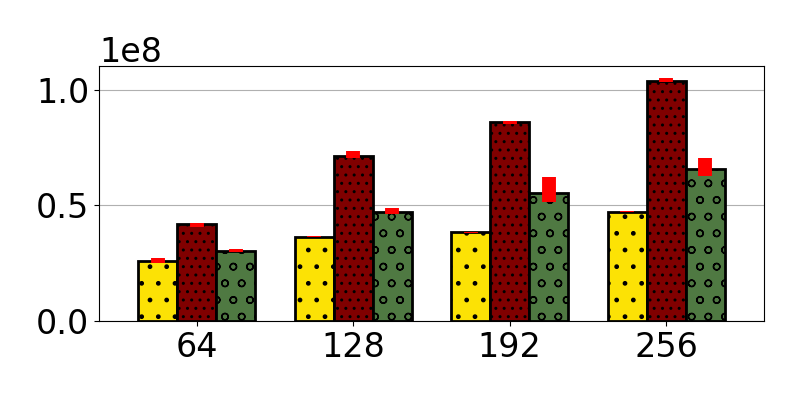}
\includegraphics[width=\legendwidth]{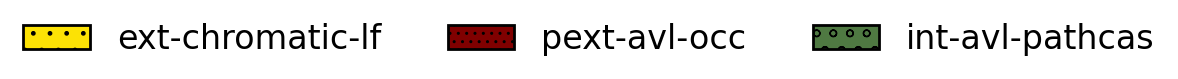}
\vspace{-4mm}
\caption{Comparing with \textbf{Handcrafted} \textbf{Balanced BSTs} on the \textbf{AMD} system. Operations per microsecond vs \# threads.}
\label{fig:nasus-non-tm-avl}
\end{figure*}

\begin{figure*}
\newcommand{\plotwidth}{0.32\linewidth}
\newcommand{\legendwidth}{0.5\linewidth}
\begin{tabular}{p{\plotwidth}p{\plotwidth}p{\plotwidth}}
\centering\noindent\textbf{10M keys, 10\% updates} & \centering\noindent\textbf{1M keys, 10\% updates} & \centering\noindent\textbf{100k keys, 10\% updates}
\end{tabular}
\rotatebox{90}{\hspace{1mm}\textbf{ops/microsecond}}
\includegraphics[width=\plotwidth]{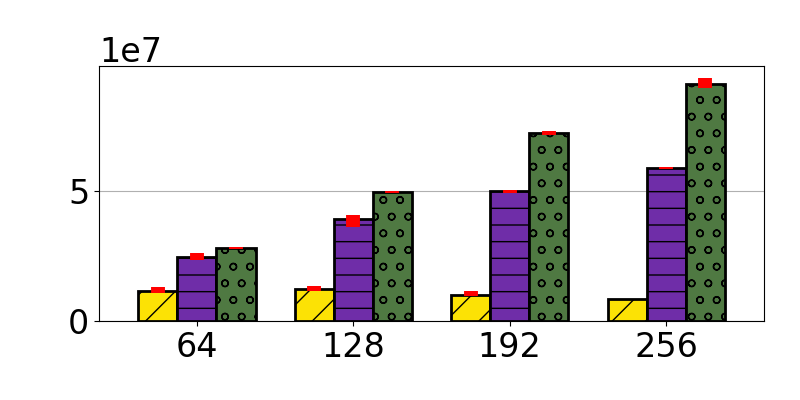}
\includegraphics[width=\plotwidth]{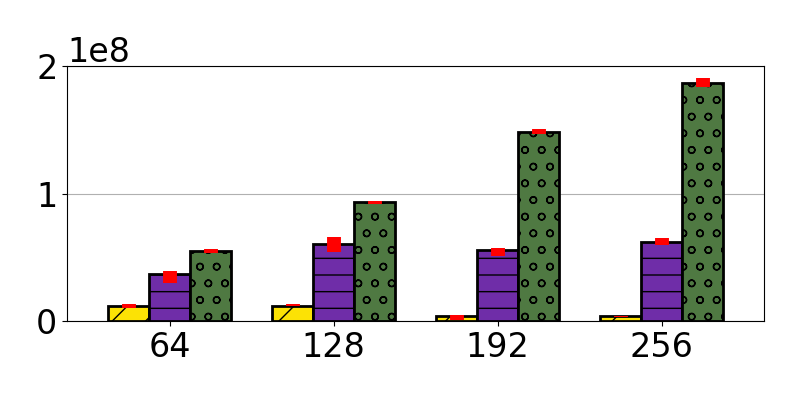}
\includegraphics[width=\plotwidth]{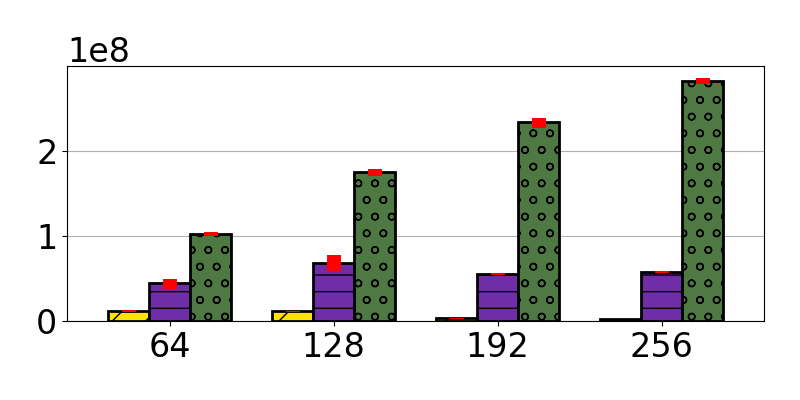}
\rotatebox{90}{\hspace{3mm}\textbf{abort rate (\%)}}
\includegraphics[width=\plotwidth]{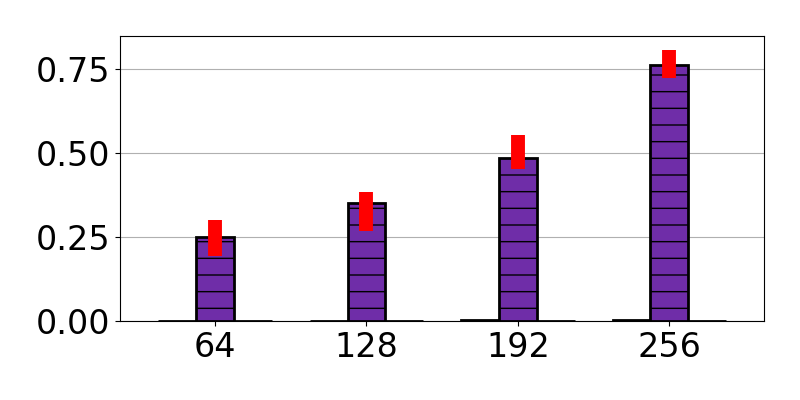}
\includegraphics[width=\plotwidth]{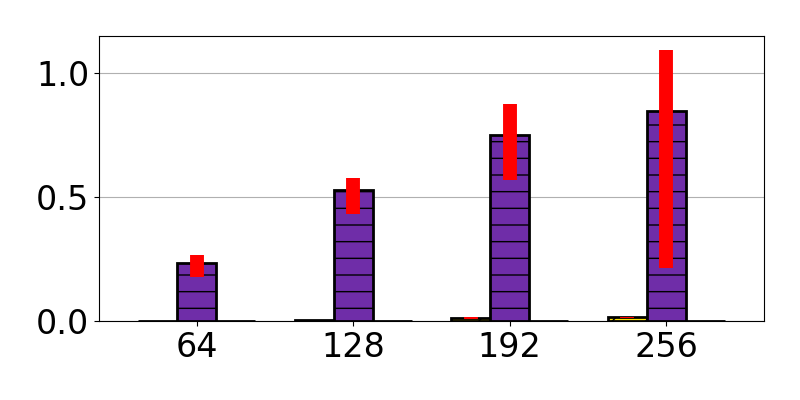}
\includegraphics[width=\plotwidth]{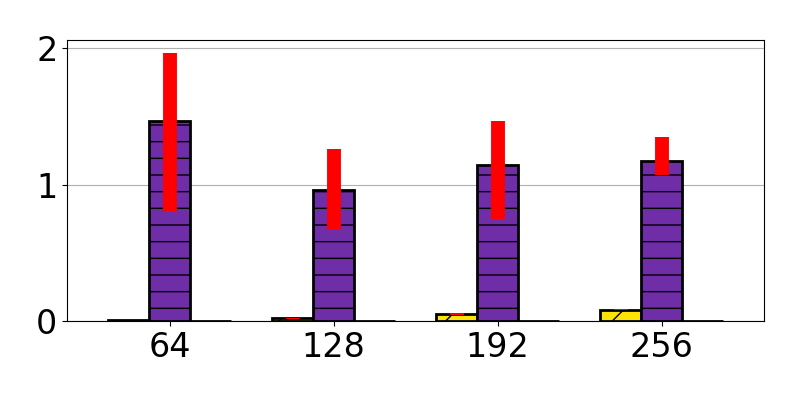}
\rotatebox{90}{\hspace{2mm}\textbf{time locked (s)}}
\includegraphics[width=\plotwidth]{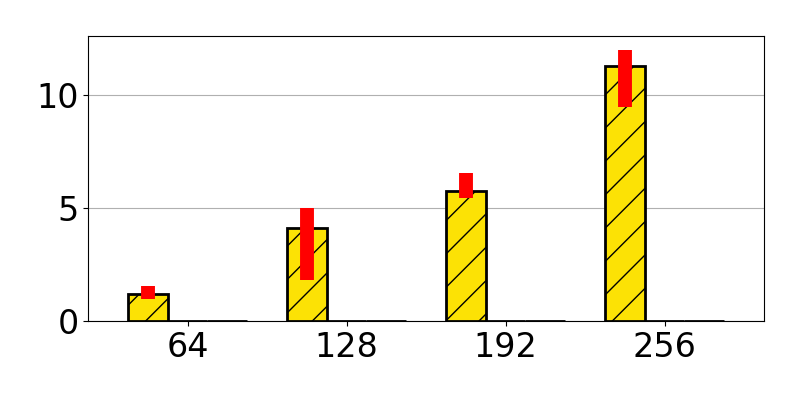}
\includegraphics[width=\plotwidth]{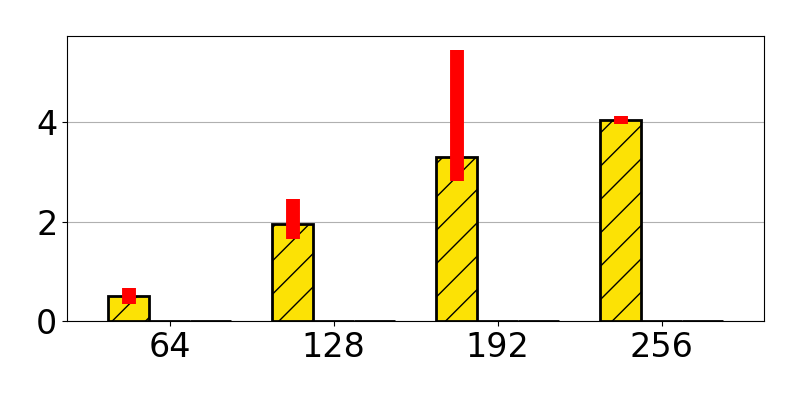}
\includegraphics[width=\plotwidth]{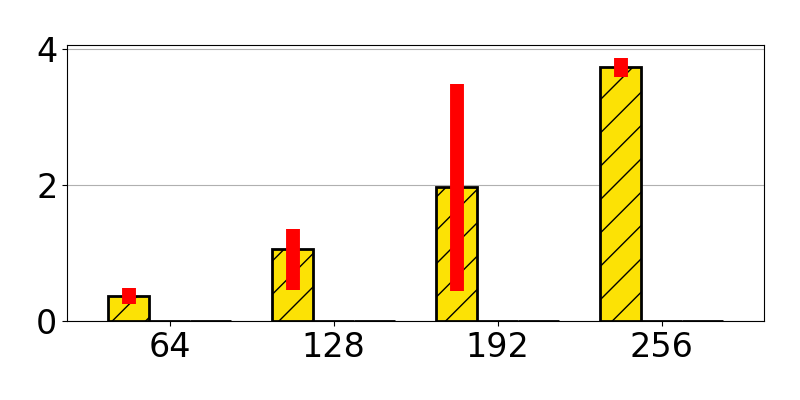}
\includegraphics[width=\legendwidth]{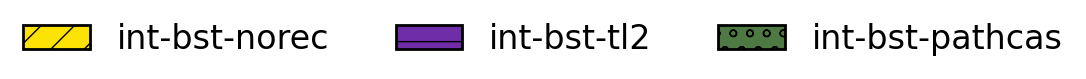}
\vspace{-4mm}
\caption{Comparing with \textbf{TM-based} \textbf{Unbalanced BSTs} on the \textbf{AMD} system. \textbf{Note varying y-axes} (x-axis = \# threads).}
\label{fig:nasus-tm-bst}
\end{figure*}

\begin{figure*}
\newcommand{\plotwidth}{0.32\linewidth}
\newcommand{\legendwidth}{0.5\linewidth}
\begin{tabular}{p{\plotwidth}p{\plotwidth}p{\plotwidth}}
\centering\noindent\textbf{10M keys, 10\% updates} & \centering\noindent\textbf{1M keys, 10\% updates} & \centering\noindent\textbf{100k keys, 10\% updates}
\end{tabular}
\rotatebox{90}{\hspace{1mm}\textbf{ops/microsecond}}
\includegraphics[width=\plotwidth]{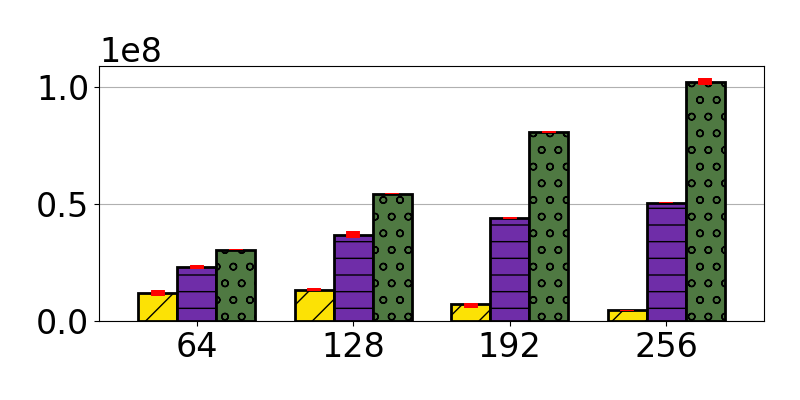}
\includegraphics[width=\plotwidth]{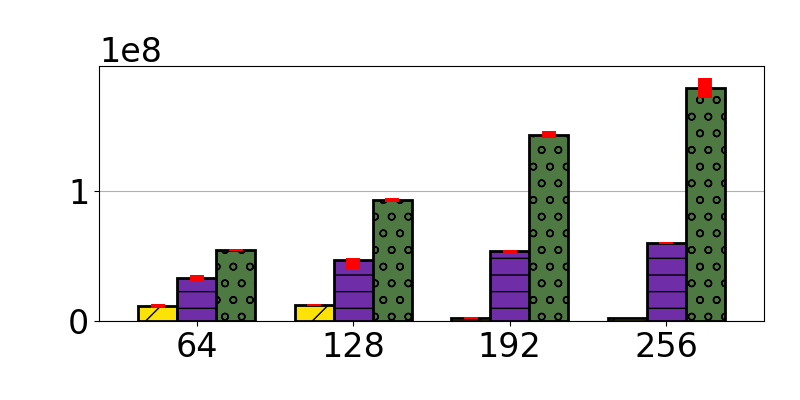}
\includegraphics[width=\plotwidth]{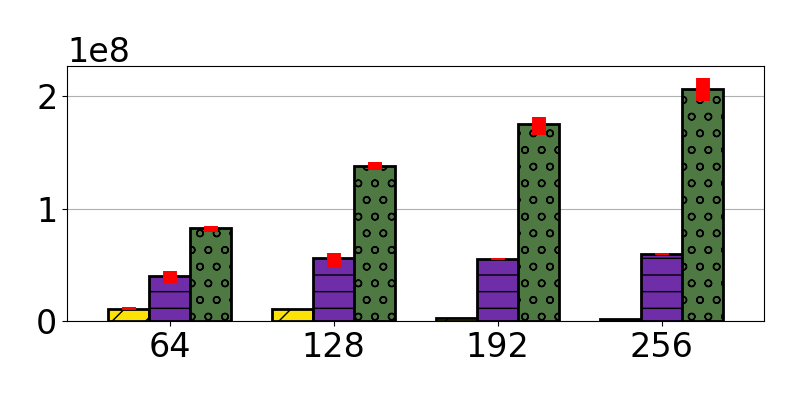}
\rotatebox{90}{\hspace{3mm}\textbf{abort rate (\%)}}
\includegraphics[width=\plotwidth]{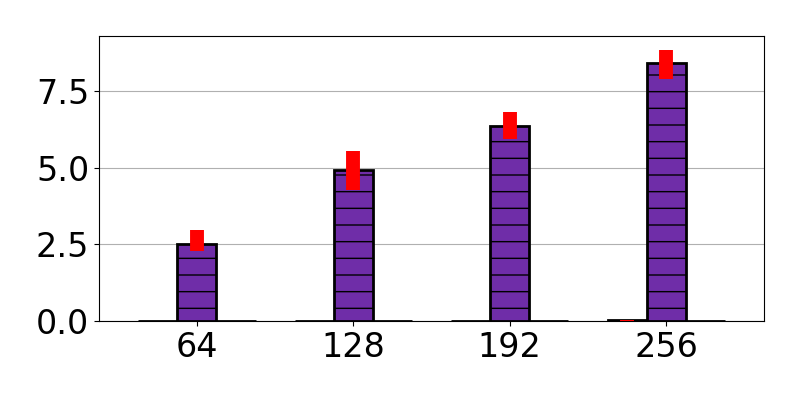}
\includegraphics[width=\plotwidth]{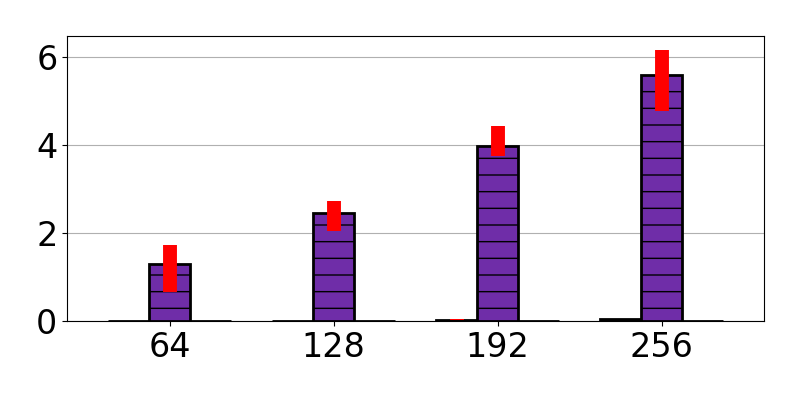}
\includegraphics[width=\plotwidth]{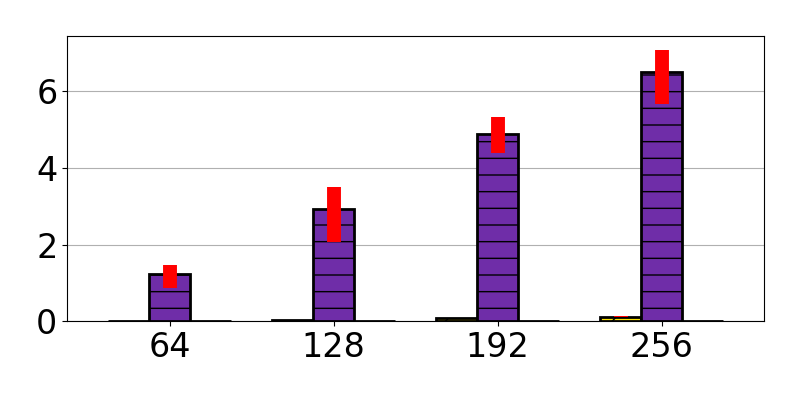}
\rotatebox{90}{\hspace{2mm}\textbf{time locked (s)}}
\includegraphics[width=\plotwidth]{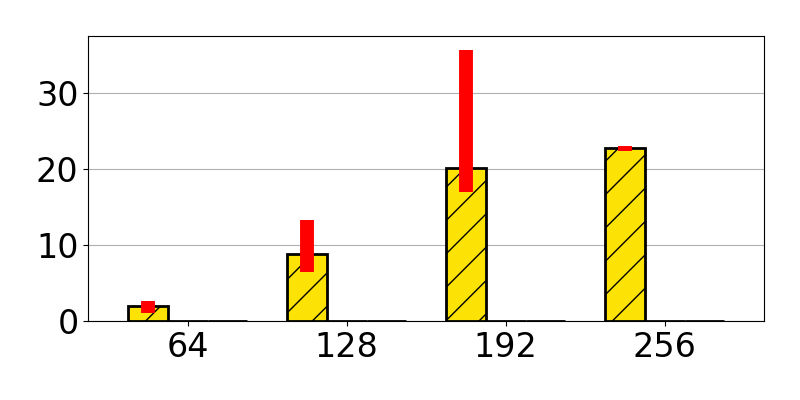}
\includegraphics[width=\plotwidth]{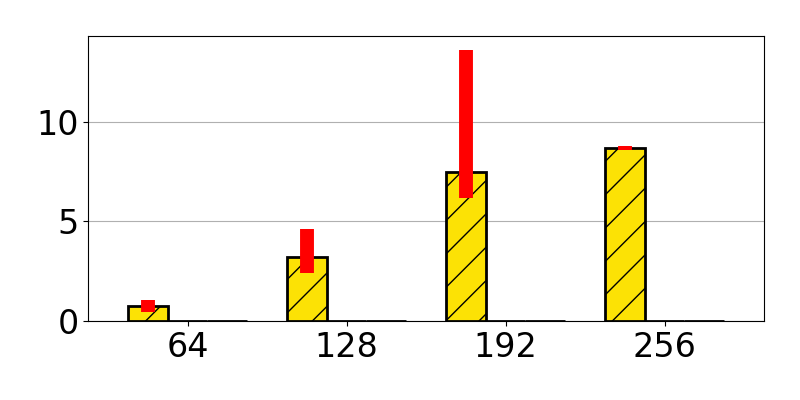}
\includegraphics[width=\plotwidth]{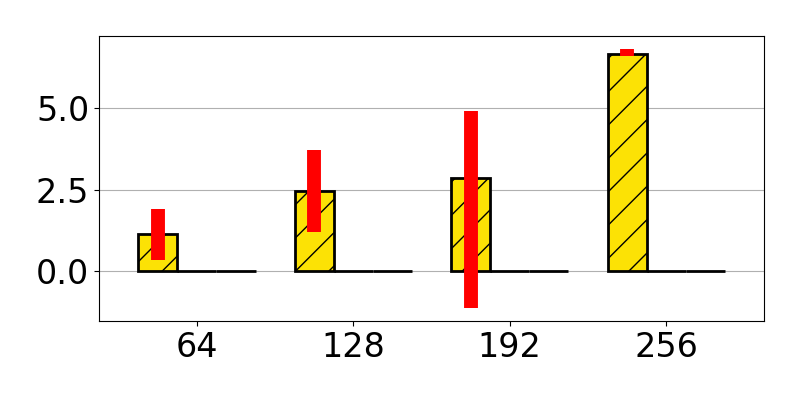}
\includegraphics[width=\legendwidth]{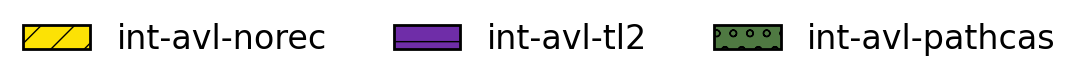}
\vspace{-4mm}
\caption{Comparing with \textbf{TM-based} \textbf{Balanced BSTs} on the \textbf{AMD} system. \textbf{Note varying y-axes} (x-axis = \# threads).}
\label{fig:nasus-tm-avl}
\end{figure*}

\begin{figure*}
\newcommand{\plotwidth}{0.32\linewidth}
\newcommand{\legendwidth}{0.5\linewidth}
\begin{tabular}{p{\plotwidth}p{\plotwidth}p{\plotwidth}}
\centering\noindent\textbf{1\% updates} & \centering\noindent\textbf{10\% updates} & \centering\noindent\textbf{100\% updates}
\end{tabular}
\rotatebox{90}{\hspace{5mm}\textbf{10M keys}}
\includegraphics[width=\plotwidth]{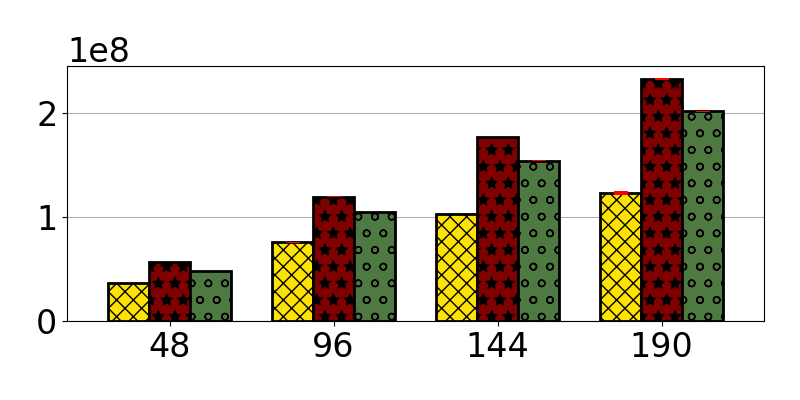}
\includegraphics[width=\plotwidth]{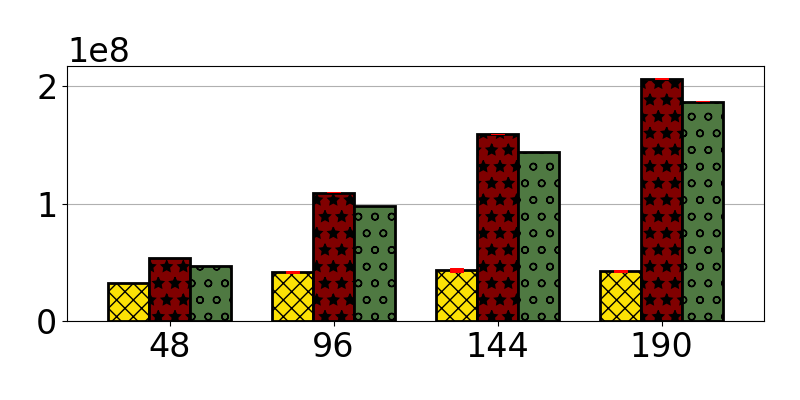}
\includegraphics[width=\plotwidth]{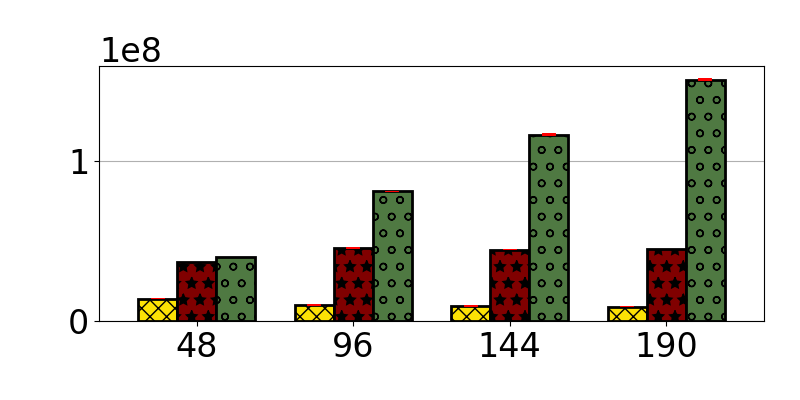}
\rotatebox{90}{\hspace{5mm}\textbf{1M keys}}
\includegraphics[width=\plotwidth]{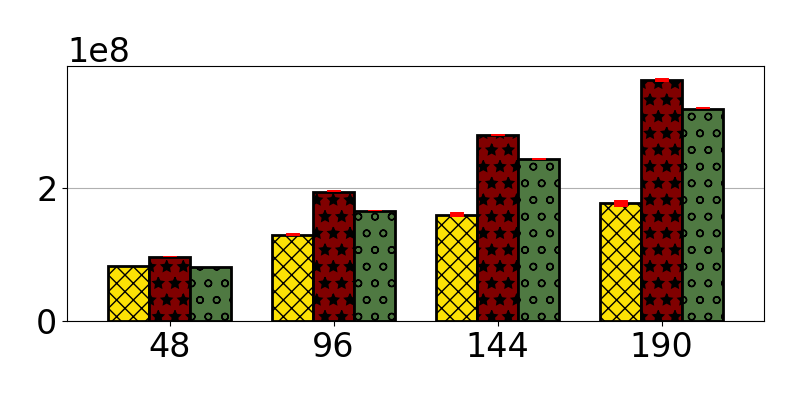}
\includegraphics[width=\plotwidth]{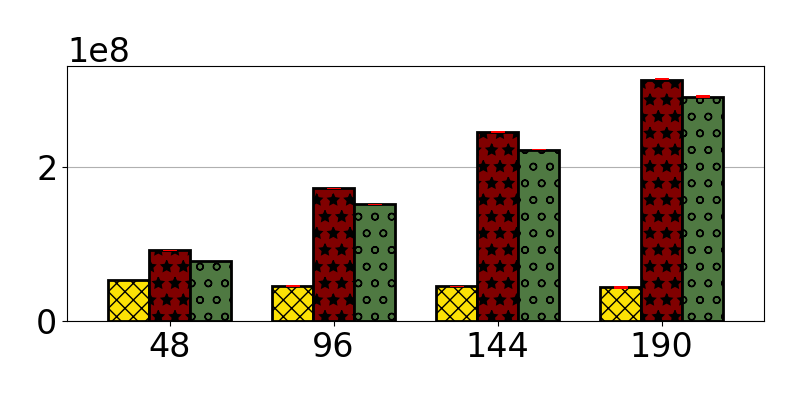}
\includegraphics[width=\plotwidth]{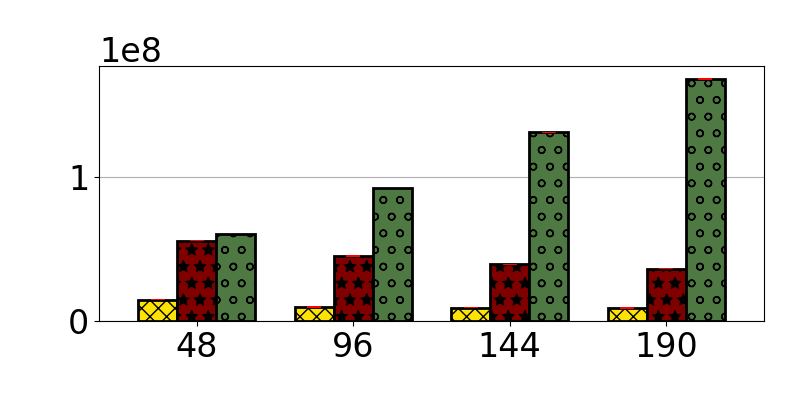}
\rotatebox{90}{\hspace{5mm}\textbf{100k keys}}
\includegraphics[width=\plotwidth]{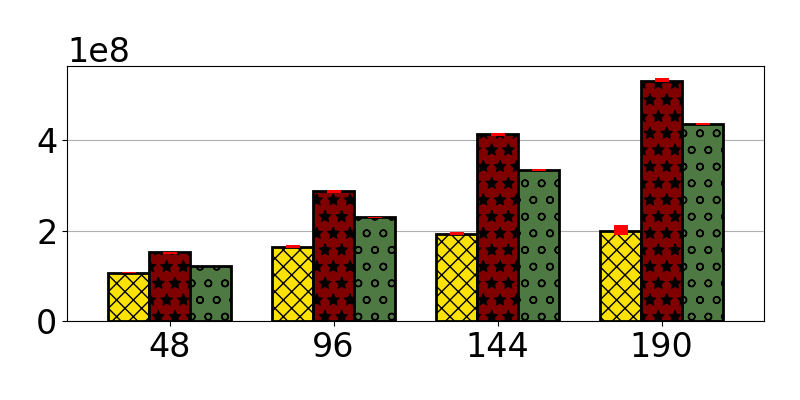}
\includegraphics[width=\plotwidth]{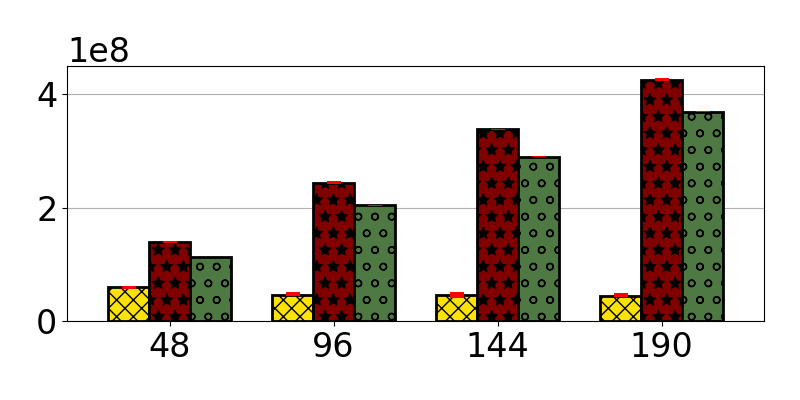}
\includegraphics[width=\plotwidth]{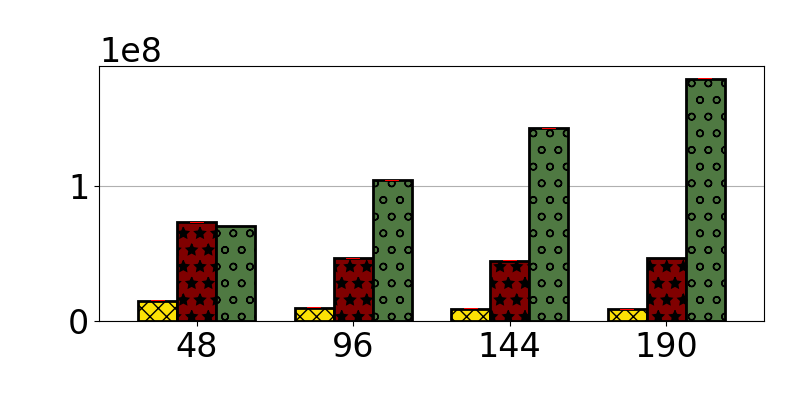}
\includegraphics[width=\legendwidth]{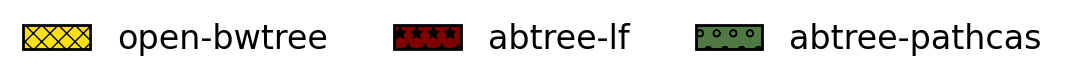}
\vspace{-4mm}
\caption{Comparing with \textbf{Handcrafted} \textbf{B-tree variants} on the \textbf{Intel} system. Operations per microsecond vs \# threads.}
\label{fig:jax-non-tm-abtree}
\end{figure*}

\begin{figure*}
\newcommand{\plotwidth}{0.32\linewidth}
\newcommand{\legendwidth}{0.8\linewidth}
\begin{tabular}{p{\plotwidth}p{\plotwidth}p{\plotwidth}}
\centering\noindent\textbf{1\% updates} & \centering\noindent\textbf{10\% updates} & \centering\noindent\textbf{100\% updates}
\end{tabular}
\rotatebox{90}{\hspace{5mm}\textbf{10M keys}}
\includegraphics[width=\plotwidth]{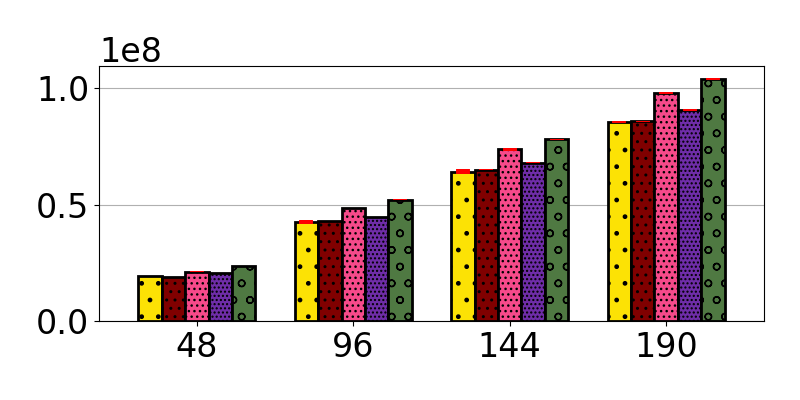}
\includegraphics[width=\plotwidth]{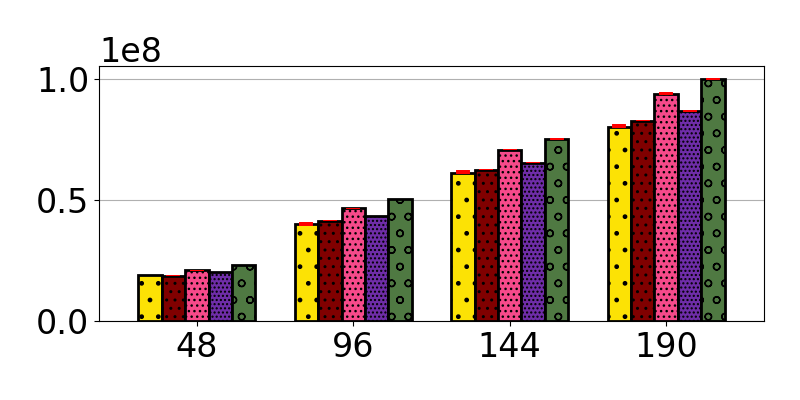}
\includegraphics[width=\plotwidth]{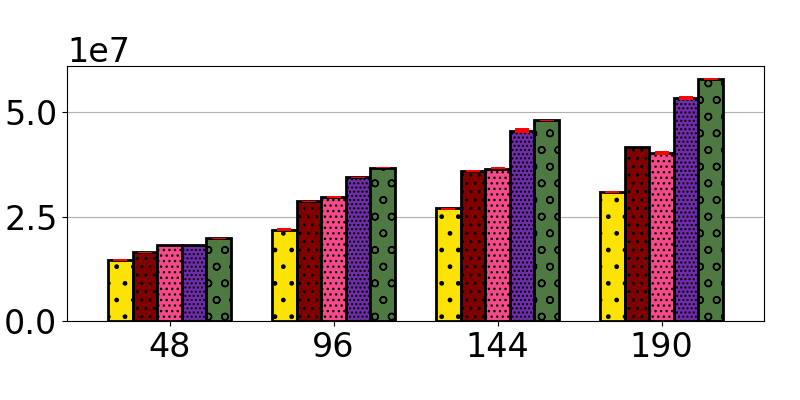}
\rotatebox{90}{\hspace{5mm}\textbf{1M keys}}
\includegraphics[width=\plotwidth]{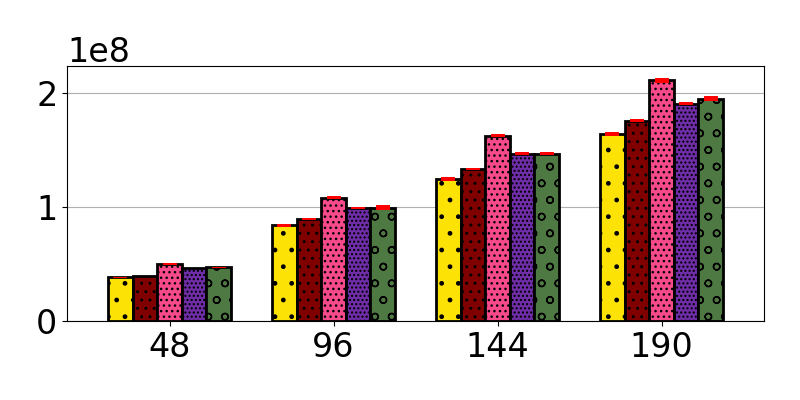}
\includegraphics[width=\plotwidth]{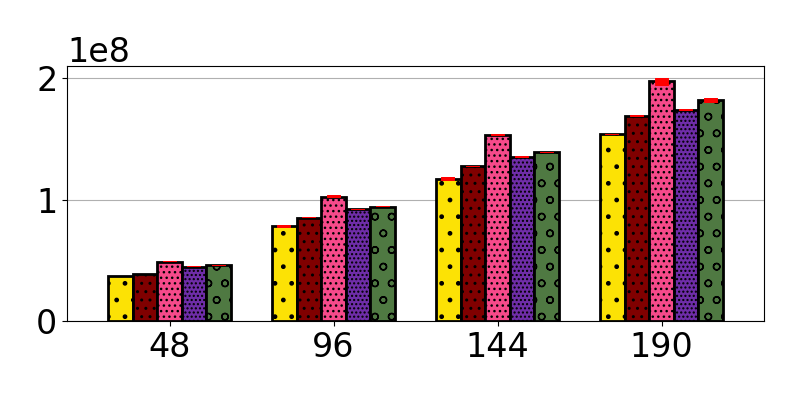}
\includegraphics[width=\plotwidth]{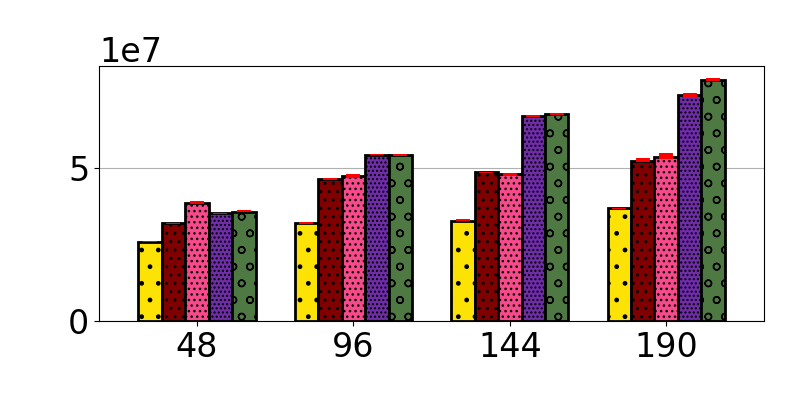}
\rotatebox{90}{\hspace{5mm}\textbf{100k keys}}
\includegraphics[width=\plotwidth]{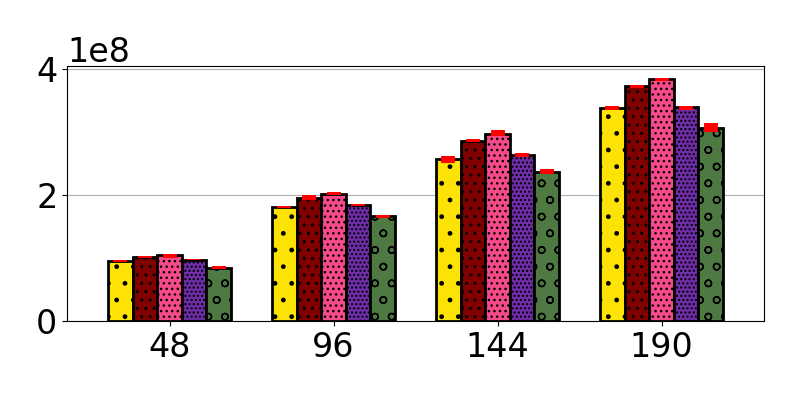}
\includegraphics[width=\plotwidth]{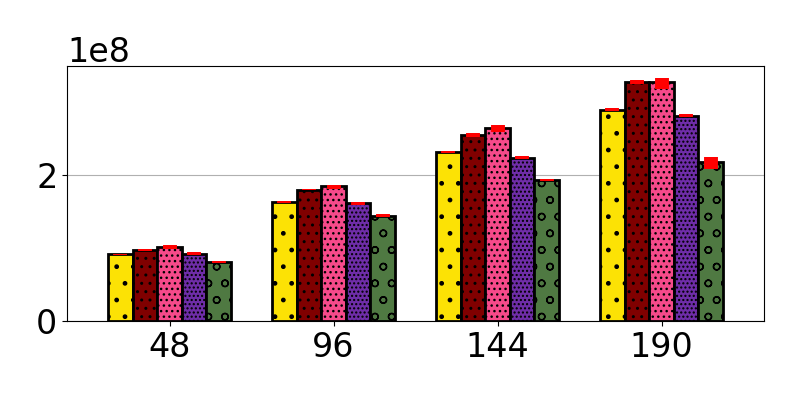}
\includegraphics[width=\plotwidth]{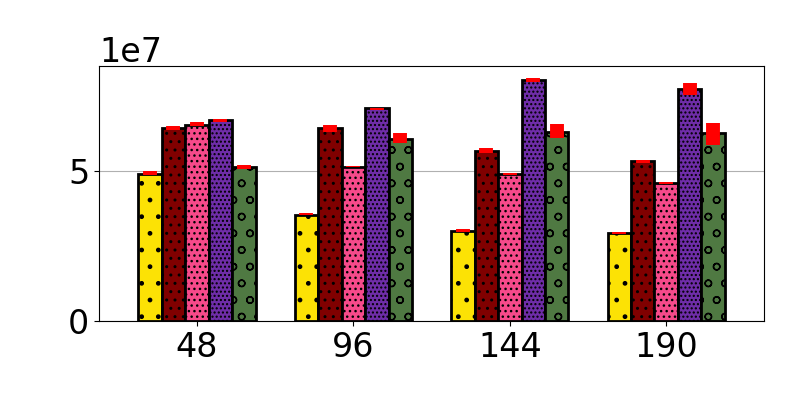}
\includegraphics[width=\legendwidth]{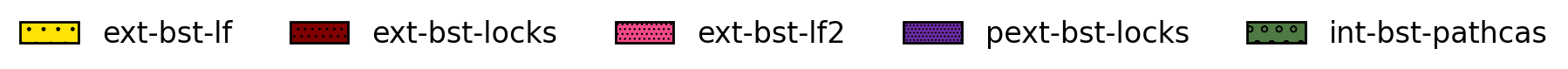}
\vspace{-4mm}
\caption{Comparing with \textbf{Handcrafted} \textbf{Unbalanced BSTs} on the \textbf{Intel} system. Operations per microsecond vs \# threads.}
\label{fig:jax-non-tm-bst}
\end{figure*}

\begin{figure*}
\newcommand{\plotwidth}{0.32\linewidth}
\newcommand{\legendwidth}{0.6\linewidth}
\begin{tabular}{p{\plotwidth}p{\plotwidth}p{\plotwidth}}
\centering\noindent\textbf{1\% updates} & \centering\noindent\textbf{10\% updates} & \centering\noindent\textbf{100\% updates}
\end{tabular}
\rotatebox{90}{\hspace{5mm}\textbf{10M keys}}
\includegraphics[width=\plotwidth]{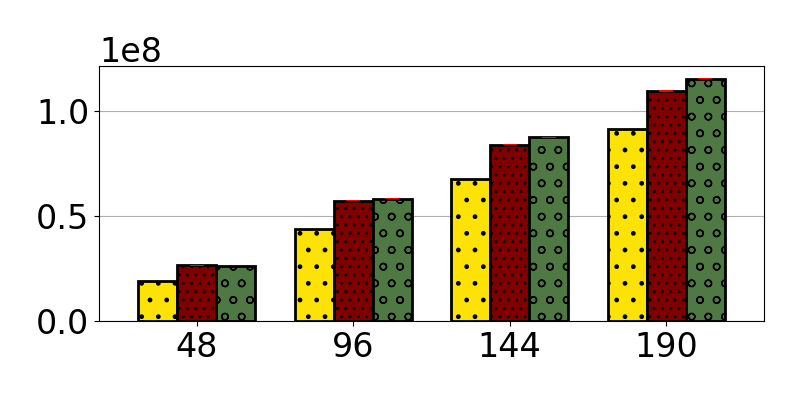}
\includegraphics[width=\plotwidth]{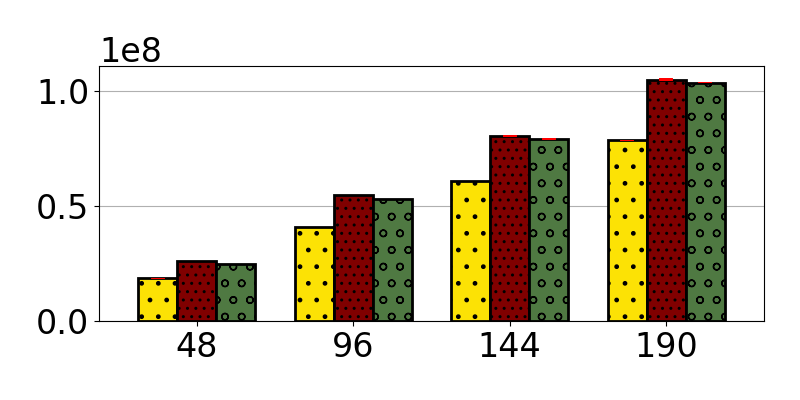}
\includegraphics[width=\plotwidth]{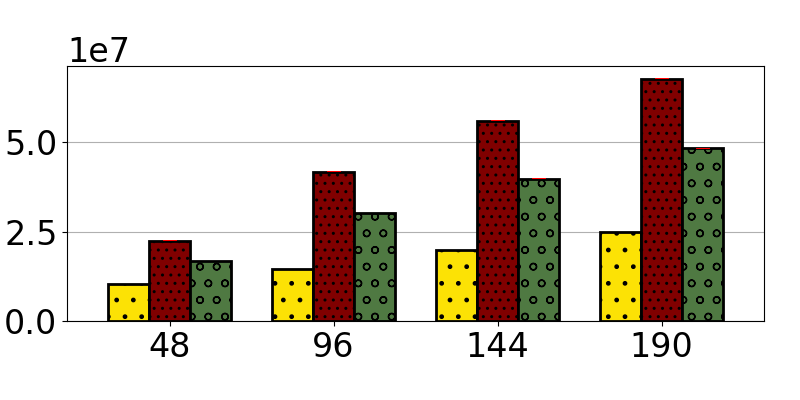}
\rotatebox{90}{\hspace{5mm}\textbf{1M keys}}
\includegraphics[width=\plotwidth]{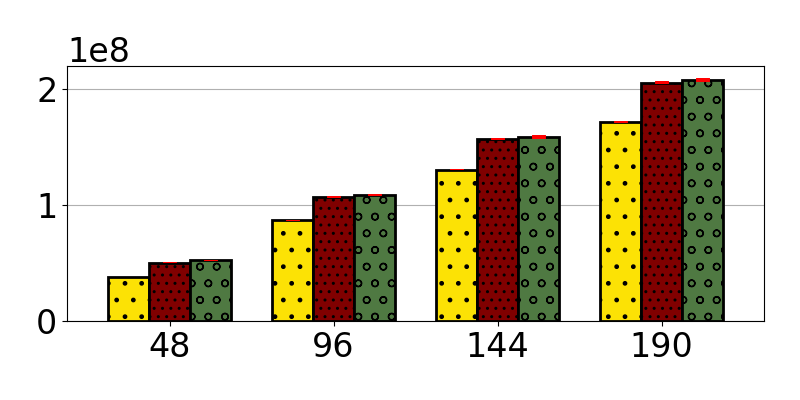}
\includegraphics[width=\plotwidth]{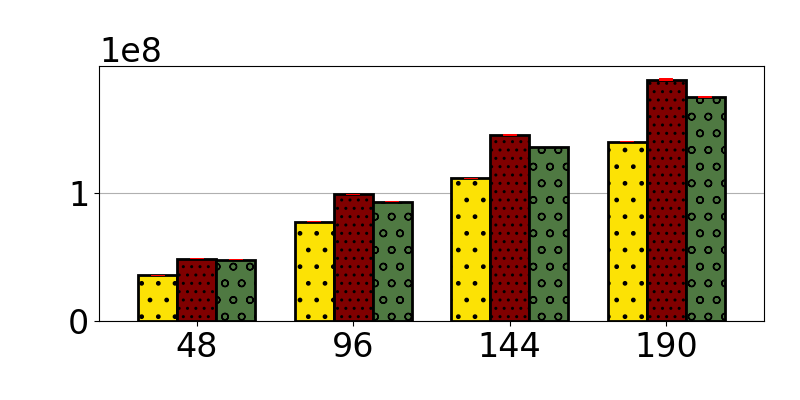}
\includegraphics[width=\plotwidth]{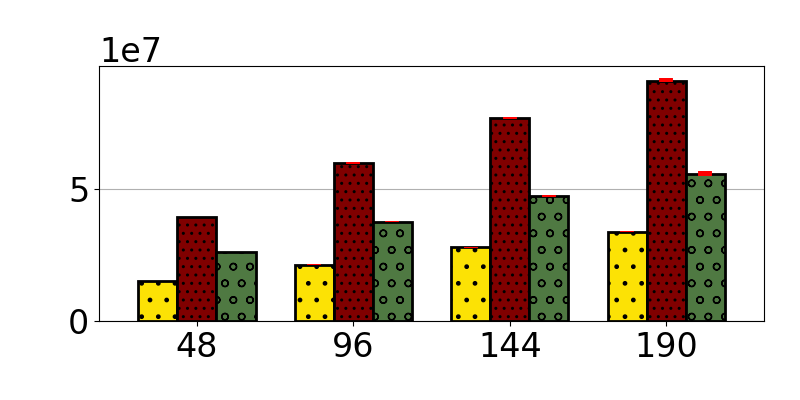}
\rotatebox{90}{\hspace{5mm}\textbf{100k keys}}
\includegraphics[width=\plotwidth]{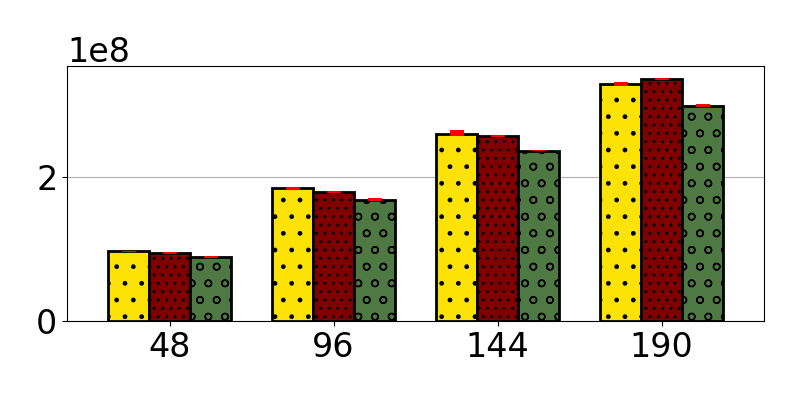}
\includegraphics[width=\plotwidth]{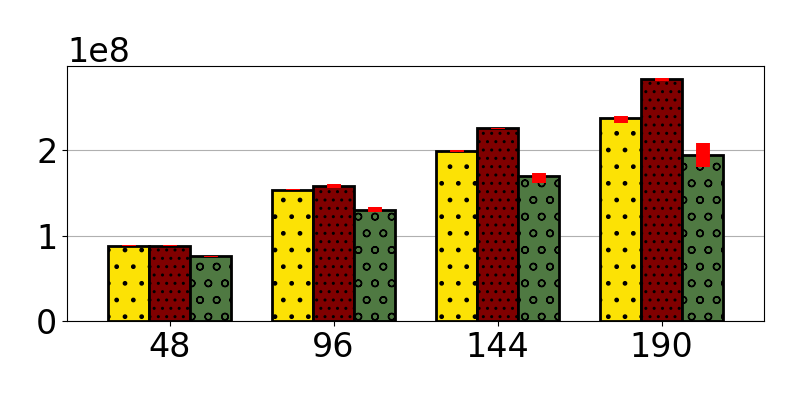}
\includegraphics[width=\plotwidth]{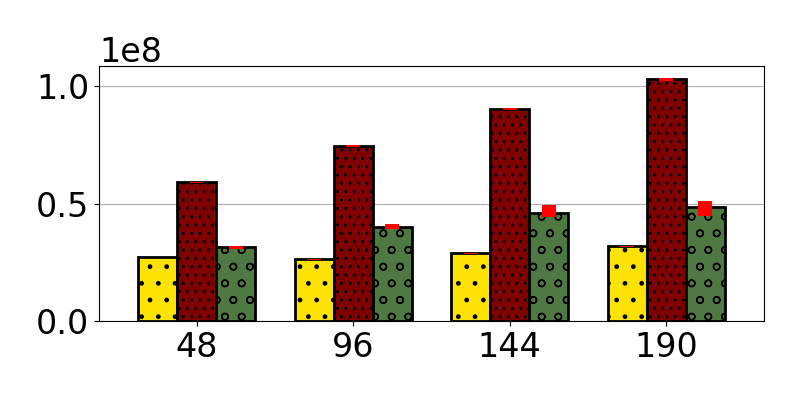}
\includegraphics[width=\legendwidth]{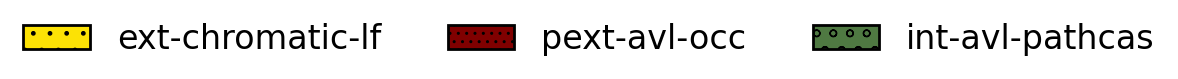}
\vspace{-4mm}
\caption{Comparing with \textbf{Handcrafted} \textbf{Balanced BSTs} on the \textbf{Intel} system. Operations per microsecond vs \# threads.}
\label{fig:jax-non-tm-avl}
\end{figure*}

\begin{figure*}
\newcommand{\plotwidth}{0.32\linewidth}
\newcommand{\legendwidth}{0.6\linewidth}
\begin{tabular}{p{\plotwidth}p{\plotwidth}p{\plotwidth}}
\centering\noindent\textbf{1\% updates} & \centering\noindent\textbf{10\% updates} & \centering\noindent\textbf{100\% updates}
\end{tabular}
\rotatebox{90}{\hspace{5mm}\textbf{10M keys}}
\includegraphics[width=\plotwidth]{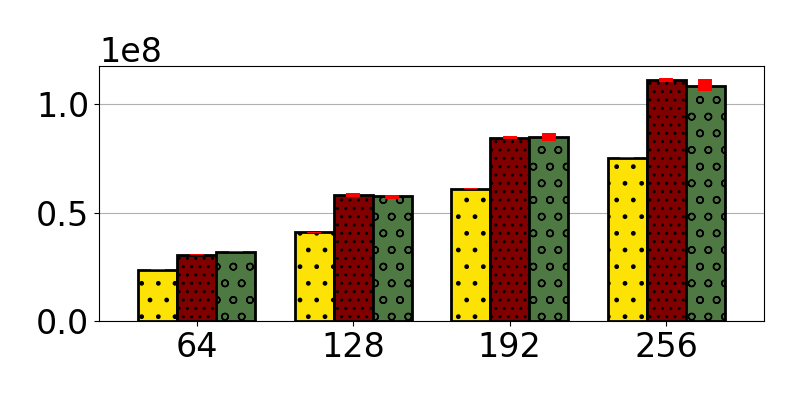}
\includegraphics[width=\plotwidth]{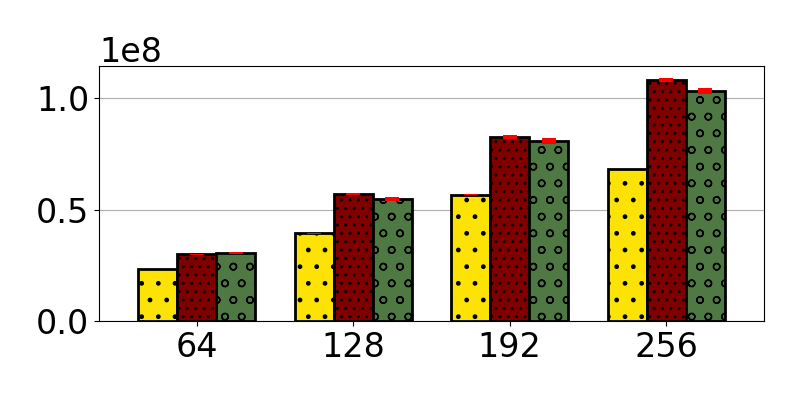}
\includegraphics[width=\plotwidth]{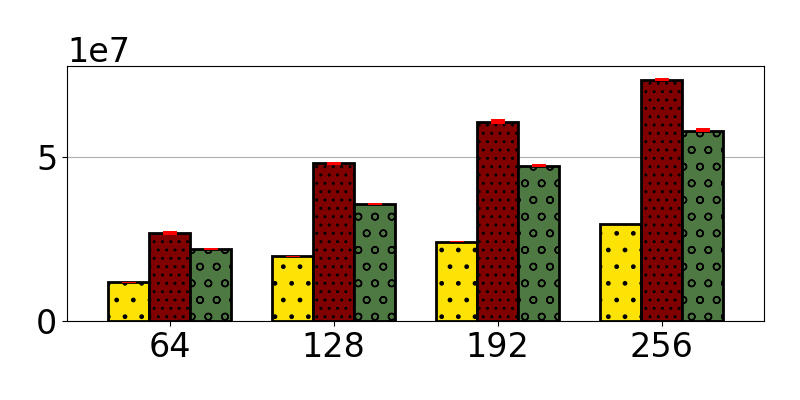}
\rotatebox{90}{\hspace{5mm}\textbf{1M keys}}
\includegraphics[width=\plotwidth]{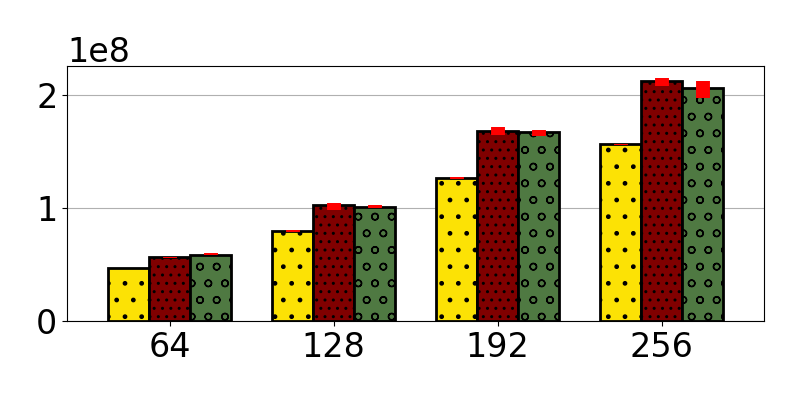}
\includegraphics[width=\plotwidth]{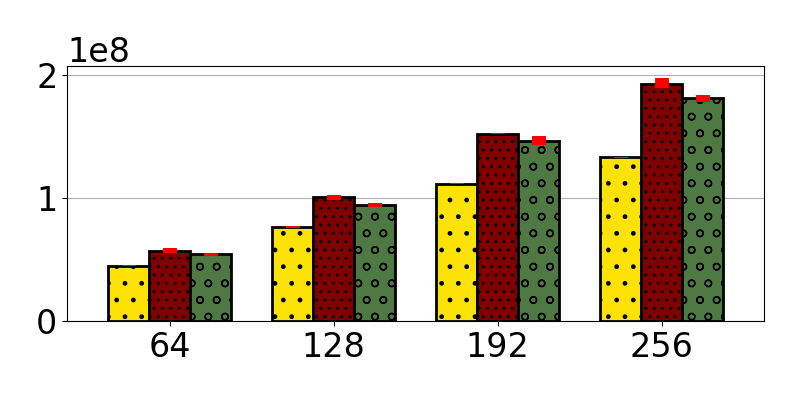}
\includegraphics[width=\plotwidth]{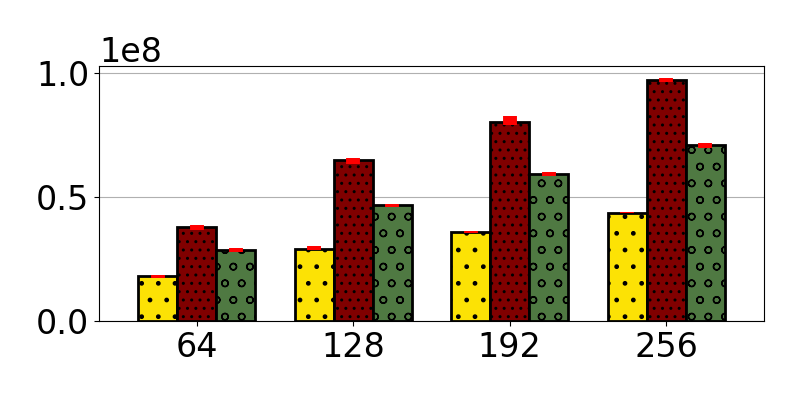}
\rotatebox{90}{\hspace{5mm}\textbf{100k keys}}
\includegraphics[width=\plotwidth]{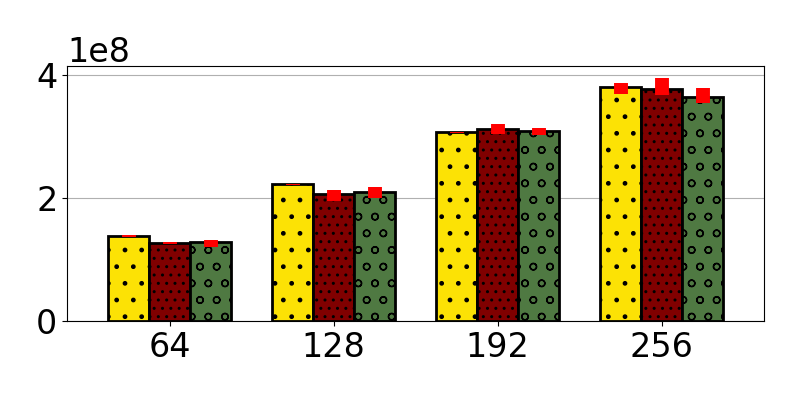}
\includegraphics[width=\plotwidth]{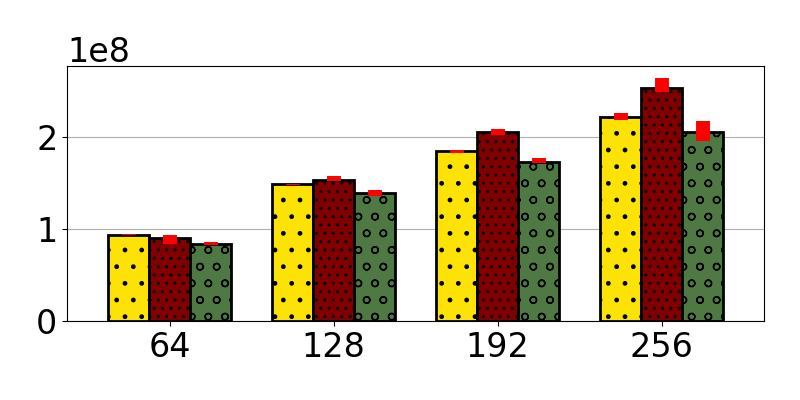}
\includegraphics[width=\plotwidth]{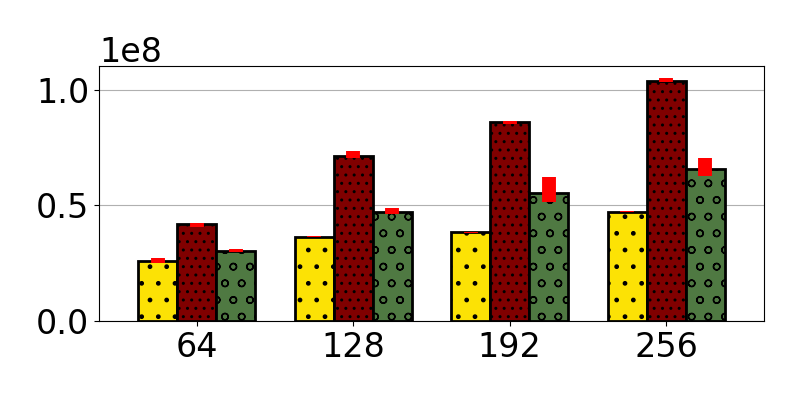}
\includegraphics[width=\legendwidth]{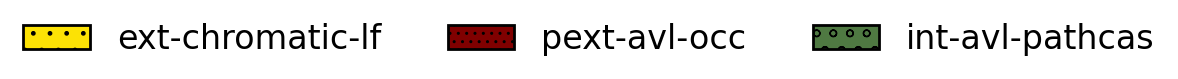}
\vspace{-4mm}
\caption{Comparing with \textbf{Handcrafted} \textbf{Balanced BSTs} on the \textbf{AMD} system. Operations per microsecond vs \# threads.}
\label{fig:nasus-non-tm-avl}
\end{figure*}

\begin{figure*}
\newcommand{\plotwidth}{0.32\linewidth}
\newcommand{\legendwidth}{\linewidth}
\begin{tabular}{p{\plotwidth}p{\plotwidth}p{\plotwidth}}
\centering\noindent\textbf{10M keys, 10\% updates} & \centering\noindent\textbf{1M keys, 10\% updates} & \centering\noindent\textbf{100k keys, 10\% updates}
\end{tabular}
\rotatebox{90}{\hspace{1mm}\textbf{ops/microsecond}}
\includegraphics[width=\plotwidth]{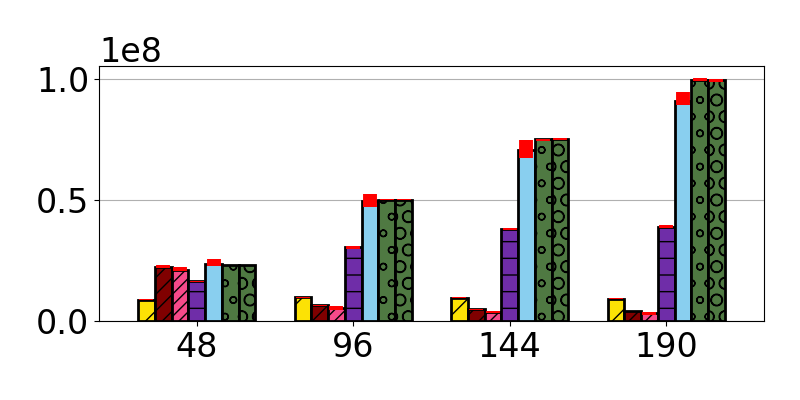}
\includegraphics[width=\plotwidth]{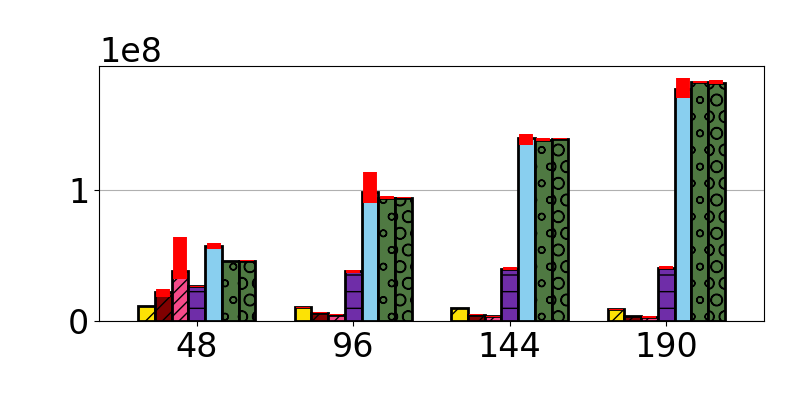}
\includegraphics[width=\plotwidth]{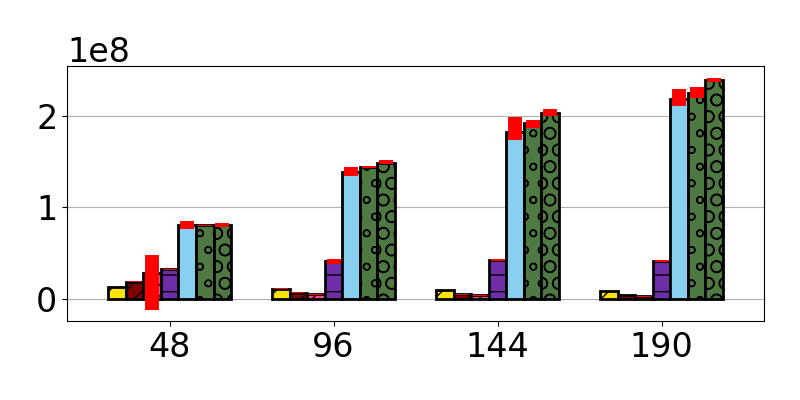}
\rotatebox{90}{\hspace{3mm}\textbf{abort rate (\%)}}
\includegraphics[width=\plotwidth]{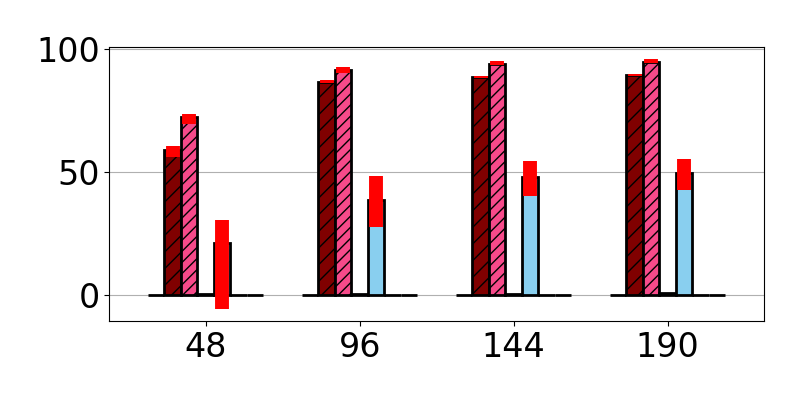}
\includegraphics[width=\plotwidth]{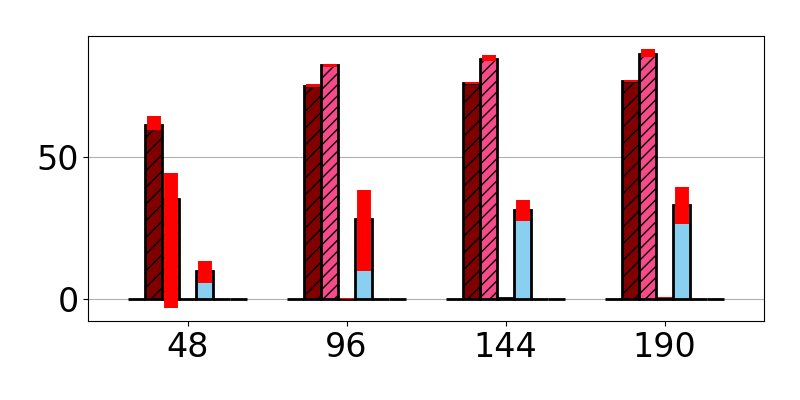}
\includegraphics[width=\plotwidth]{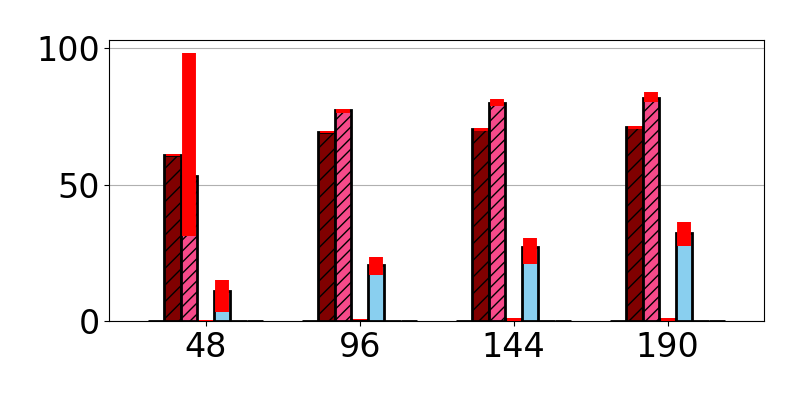}
\rotatebox{90}{\hspace{2mm}\textbf{time locked (s)}}
\includegraphics[width=\plotwidth]{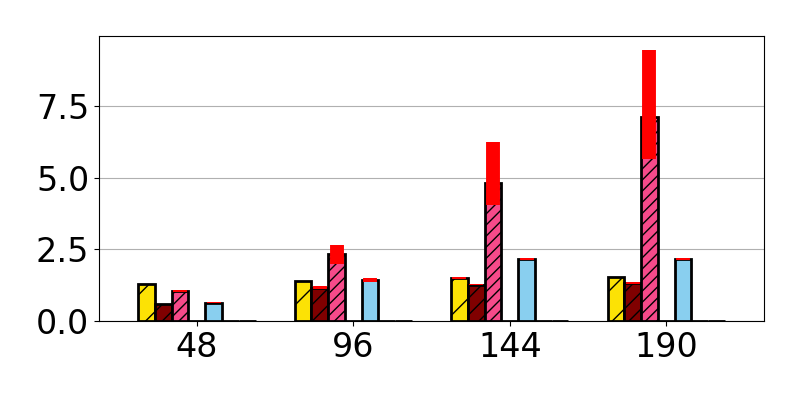}
\includegraphics[width=\plotwidth]{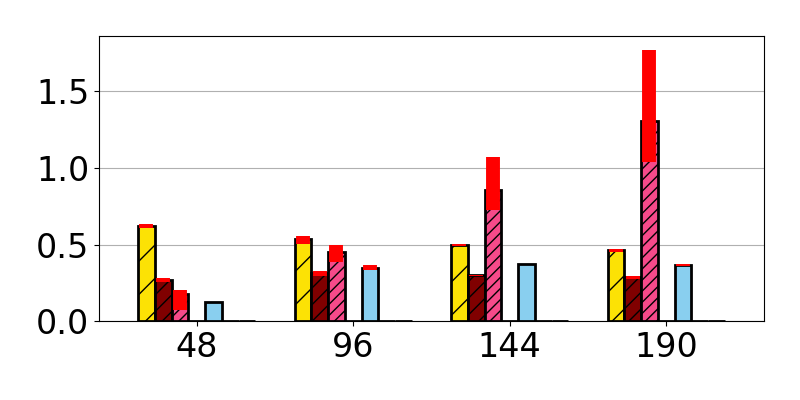}
\includegraphics[width=\plotwidth]{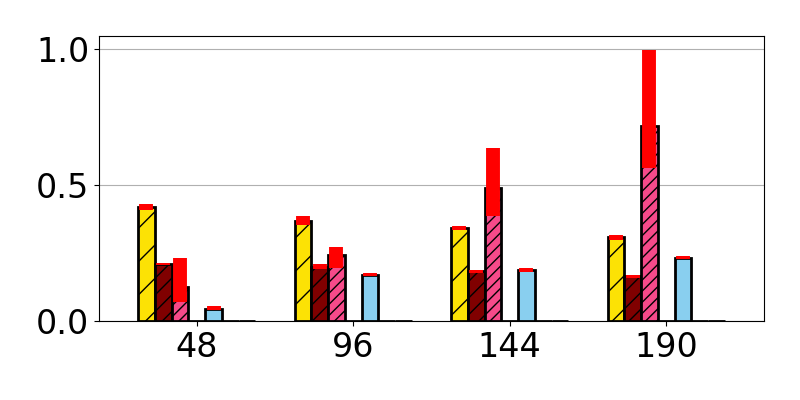}
\includegraphics[width=\legendwidth]{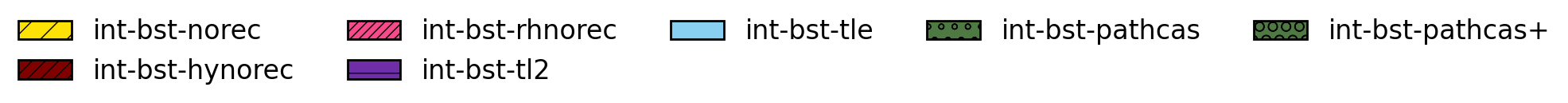}
\vspace{-6mm}
\caption{Comparing with \textbf{TM-based} \textbf{Unbalanced BSTs} on the \textbf{Intel} system. \textbf{Note varying y-axes} (x-axis = \# threads).}
\label{fig:jax-tm-bst}
\end{figure*}

\begin{figure*}
\newcommand{\plotwidth}{0.32\linewidth}
\newcommand{\legendwidth}{\linewidth}
\begin{tabular}{p{\plotwidth}p{\plotwidth}p{\plotwidth}}
\centering\noindent\textbf{10M keys, 10\% updates} & \centering\noindent\textbf{1M keys, 10\% updates} & \centering\noindent\textbf{100k keys, 10\% updates}
\end{tabular}
\rotatebox{90}{\hspace{1mm}\textbf{ops/microsecond}}
\includegraphics[width=\plotwidth]{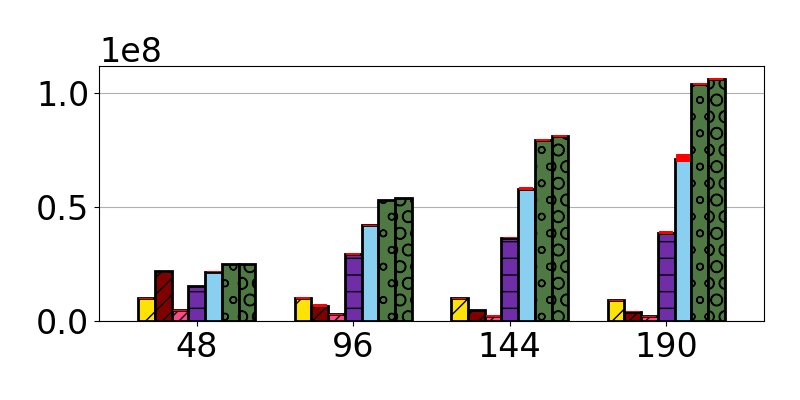}
\includegraphics[width=\plotwidth]{images/ppopp21/jax_tm/avl_tm_total_throughput-u5.0_5.0-k2000000.png}
\includegraphics[width=\plotwidth]{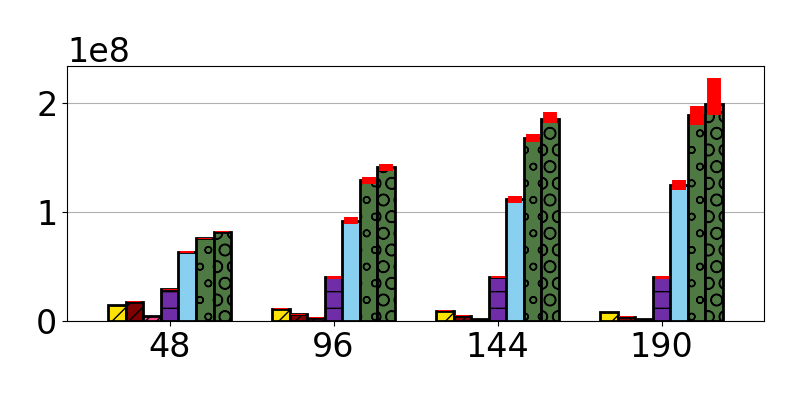}
\rotatebox{90}{\hspace{3mm}\textbf{abort rate (\%)}}
\includegraphics[width=\plotwidth]{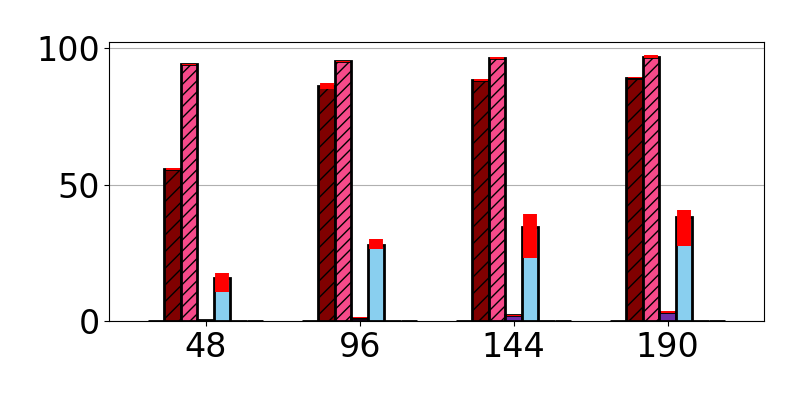}
\includegraphics[width=\plotwidth]{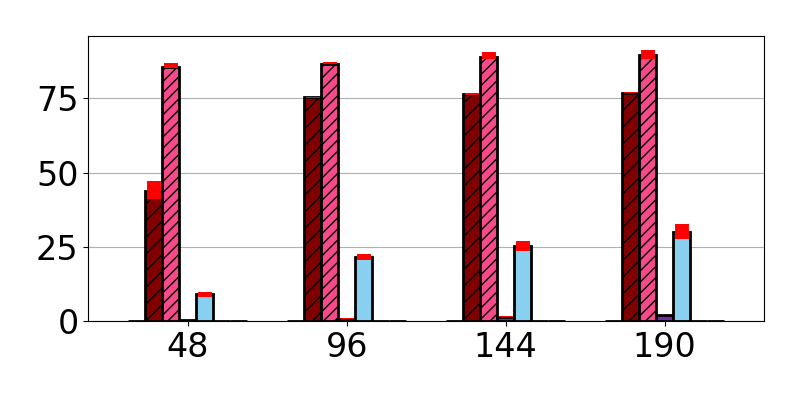}
\includegraphics[width=\plotwidth]{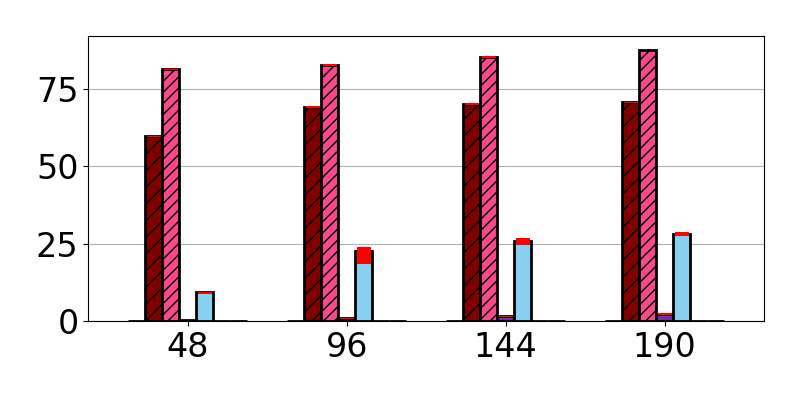}
\rotatebox{90}{\hspace{2mm}\textbf{time locked (s)}}
\includegraphics[width=\plotwidth]{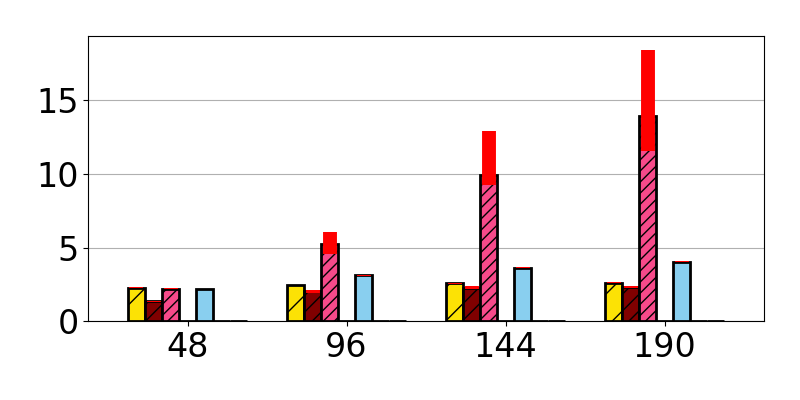}
\includegraphics[width=\plotwidth]{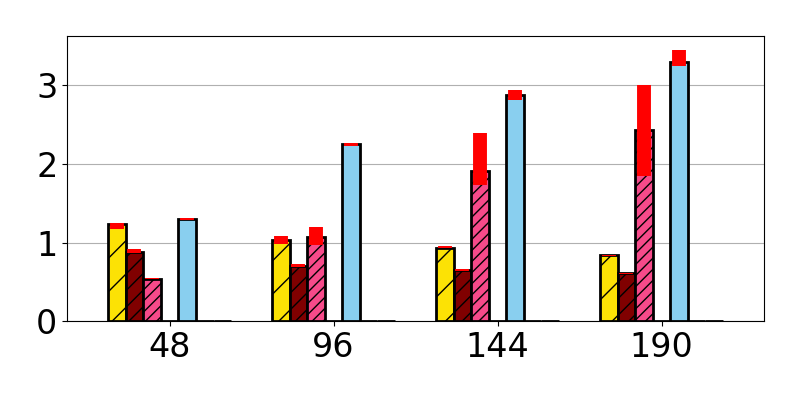}
\includegraphics[width=\plotwidth]{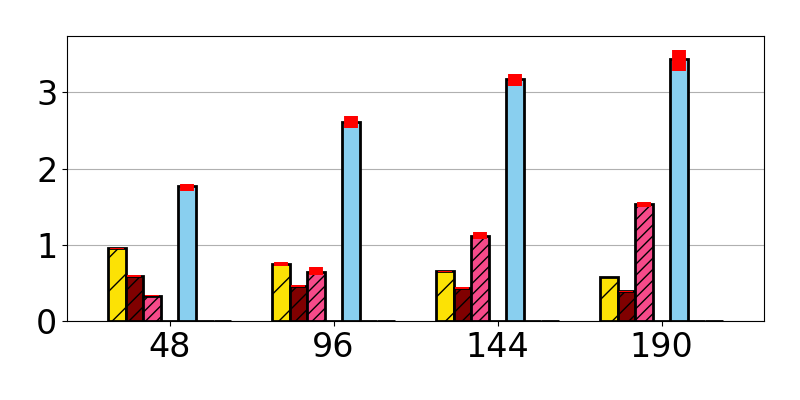}
\includegraphics[width=\legendwidth]{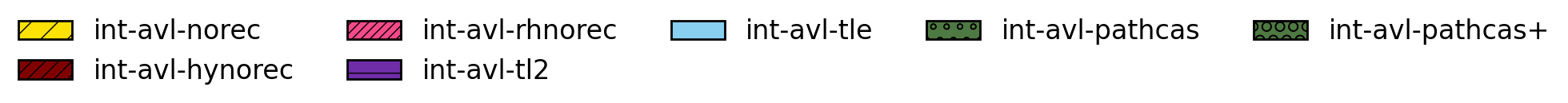}
\vspace{-6mm}
\caption{Comparing with \textbf{TM-based} \textbf{Balanced BSTs} on the \textbf{Intel} system. \textbf{Note varying y-axes} (x-axis = \# threads).}
\label{fig:jax-tm-avl}
\end{figure*}

\begin{figure*}
\newcommand{\plotwidth}{0.32\linewidth}
\newcommand{\legendwidth}{0.5\linewidth}
\begin{tabular}{p{\plotwidth}p{\plotwidth}p{\plotwidth}}
\centering\noindent\textbf{1\% updates, 1M keys} & \centering\noindent\textbf{10\% updates, 1M keys} & \centering\noindent\textbf{100\% updates, 1M keys}
\end{tabular}
\rotatebox{90}{\hspace{1mm}\textbf{ops/microsecond}}
\includegraphics[width=\plotwidth]{images/ppopp21/jax_non_tm/bst_hc_total_throughput-u0.5_0.5-k2000000.png}
\includegraphics[width=\plotwidth]{images/ppopp21/jax_non_tm/bst_hc_total_throughput-u5.0_5.0-k2000000.png}
\includegraphics[width=\plotwidth]{images/ppopp21/jax_non_tm/bst_hc_total_throughput-u50.0_50.0-k2000000.png}
\rotatebox{90}{\hspace{3mm}\textbf{L2 miss/op}}
\includegraphics[width=\plotwidth]{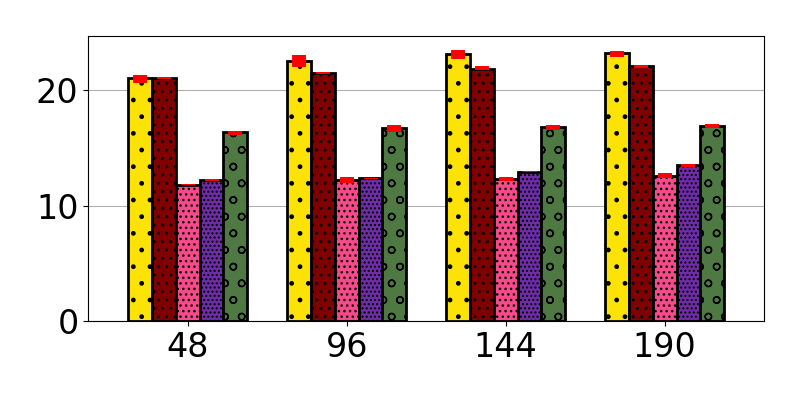}
\includegraphics[width=\plotwidth]{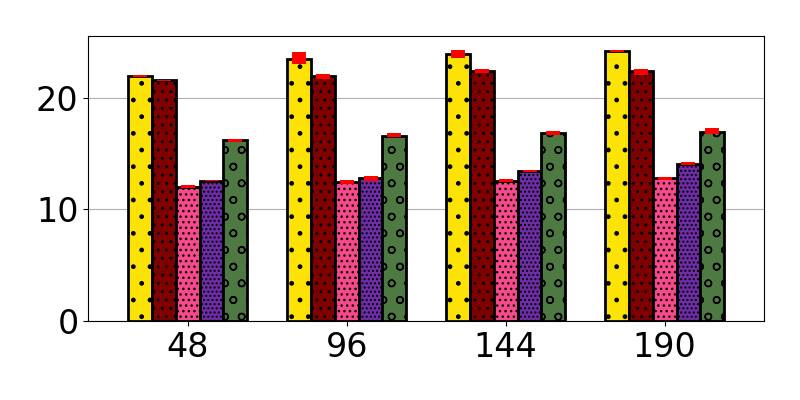}
\includegraphics[width=\plotwidth]{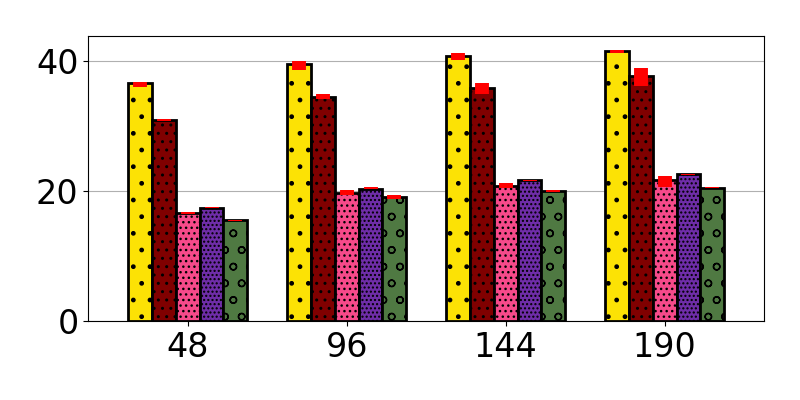}
\rotatebox{90}{\hspace{3mm}\textbf{L3 miss/op}}
\includegraphics[width=\plotwidth]{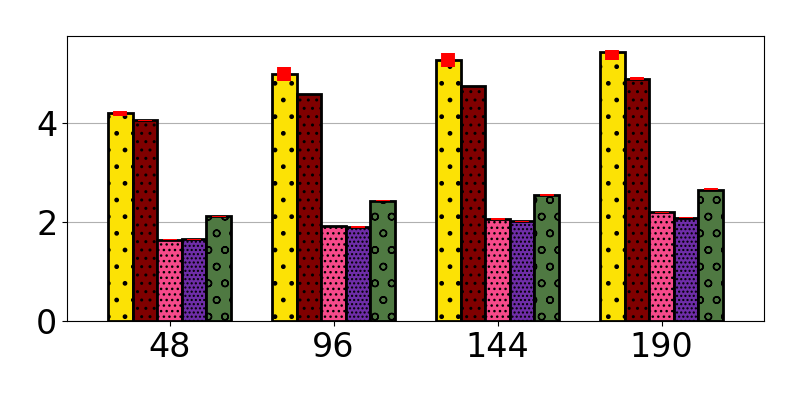}
\includegraphics[width=\plotwidth]{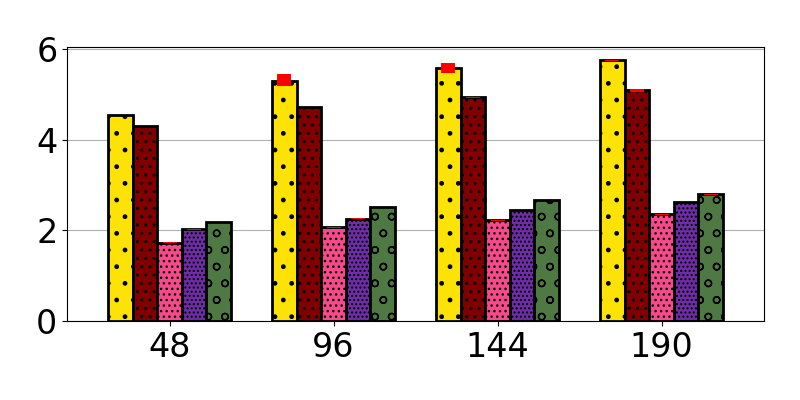}
\includegraphics[width=\plotwidth]{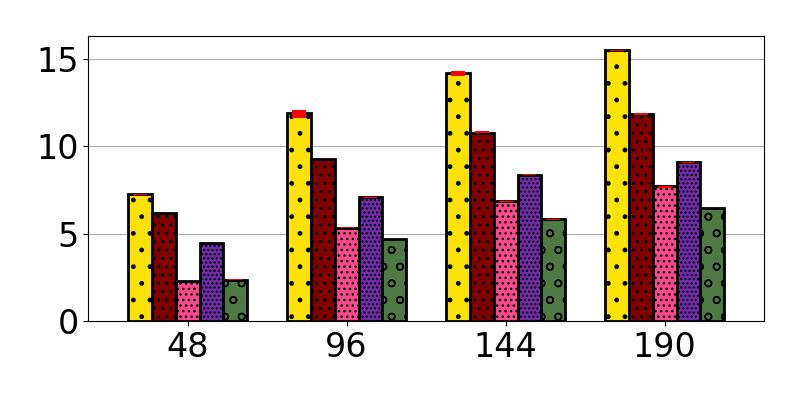}
\rotatebox{90}{\hspace{3mm}\textbf{cycles/op}}
\includegraphics[width=\plotwidth]{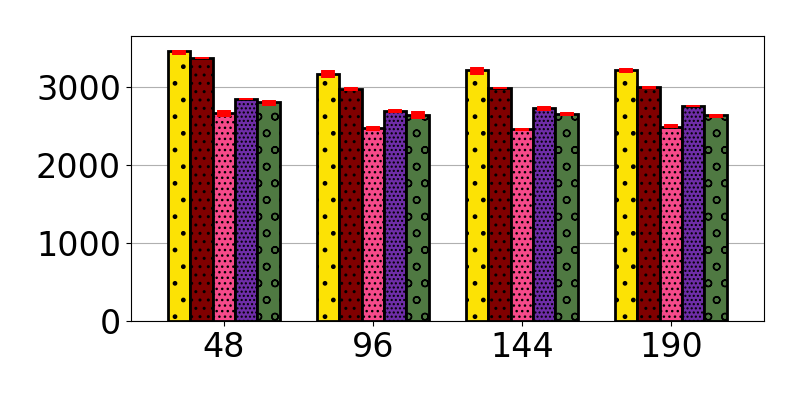}
\includegraphics[width=\plotwidth]{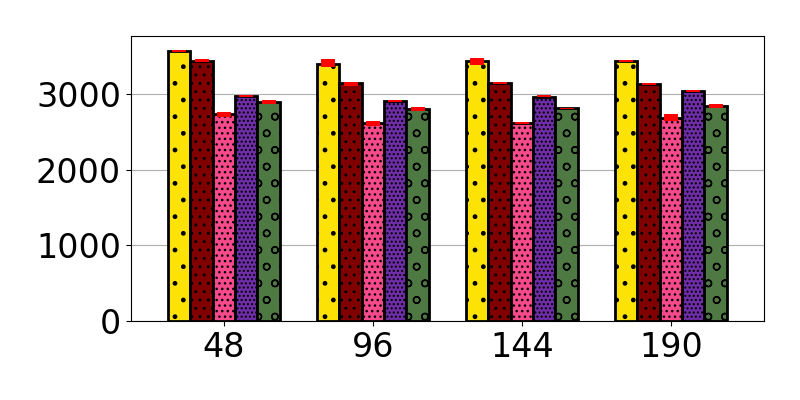}
\includegraphics[width=\plotwidth]{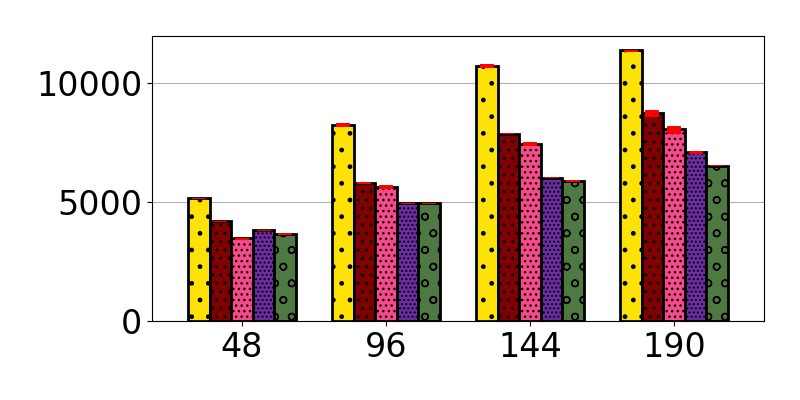}
\rotatebox{90}{\hspace{1mm}\textbf{instructions/op}}
\includegraphics[width=\plotwidth]{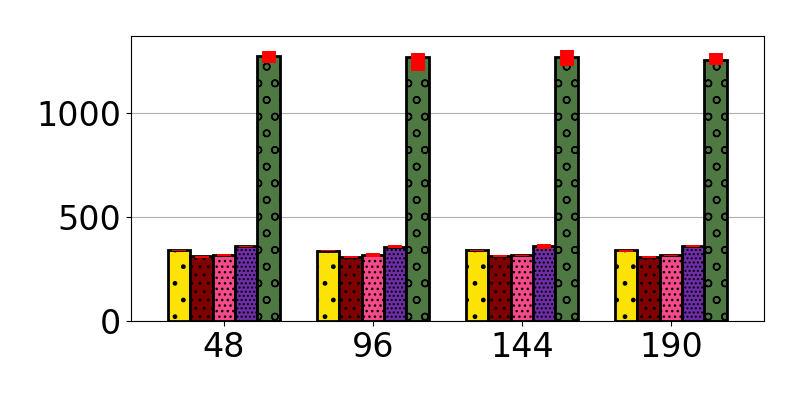}
\includegraphics[width=\plotwidth]{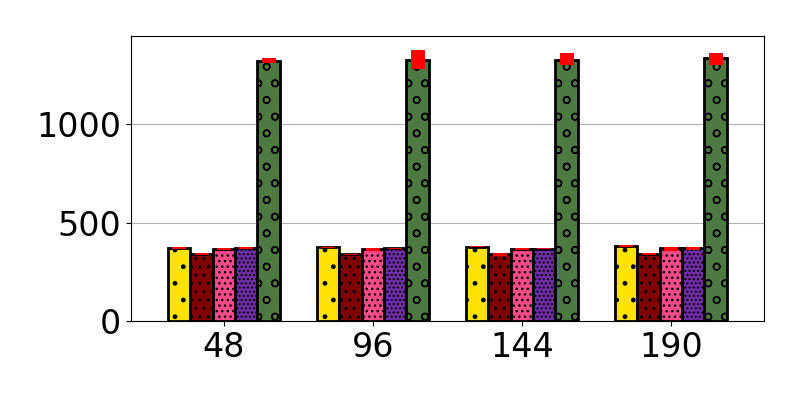}
\includegraphics[width=\plotwidth]{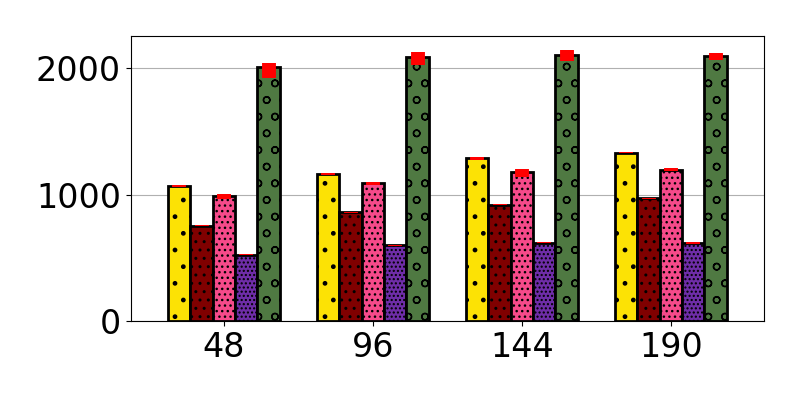}
\rotatebox{90}{\hspace{2mm}\textbf{page faults/op}}
\includegraphics[width=\plotwidth]{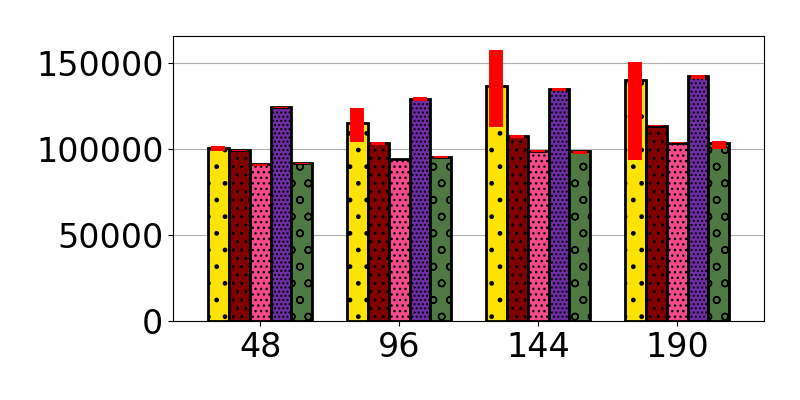}
\includegraphics[width=\plotwidth]{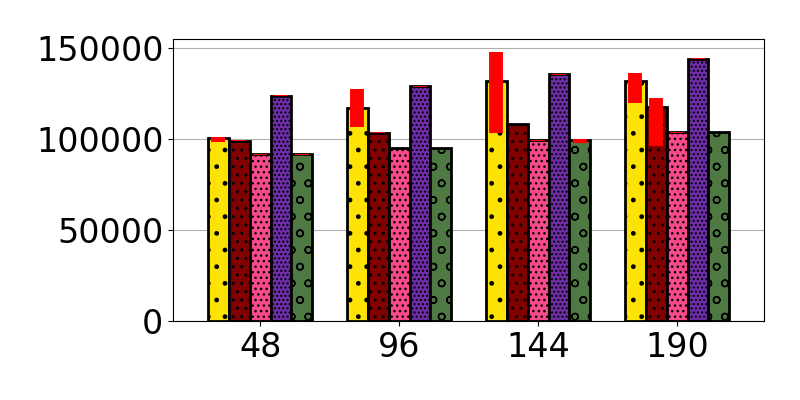}
\includegraphics[width=\plotwidth]{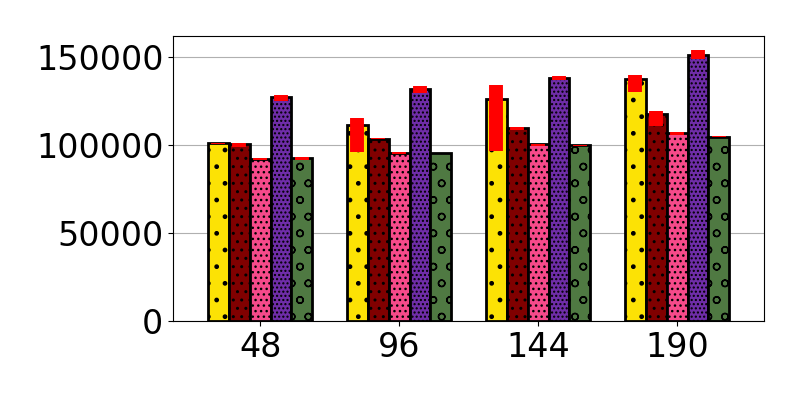}
\rotatebox{90}{\hspace{2mm}\textbf{avg key depth}}
\includegraphics[width=\plotwidth]{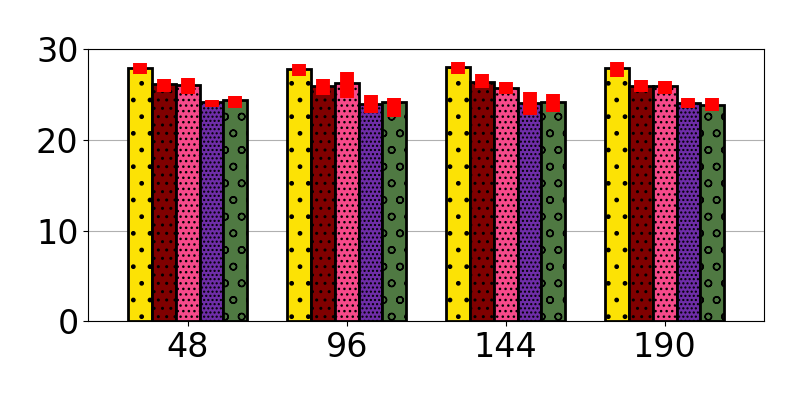}
\includegraphics[width=\plotwidth]{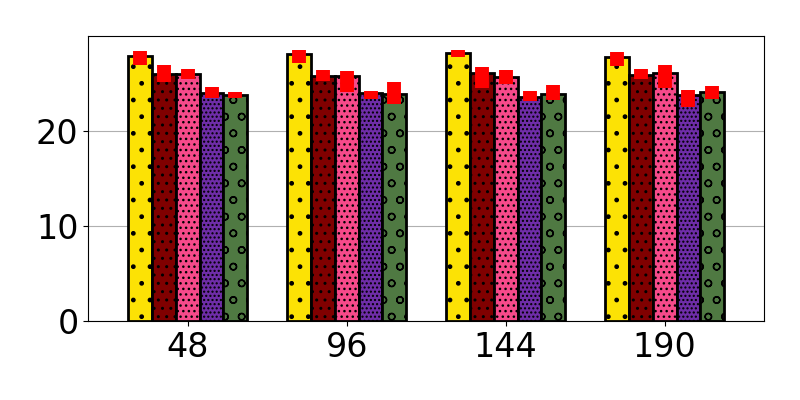}
\includegraphics[width=\plotwidth]{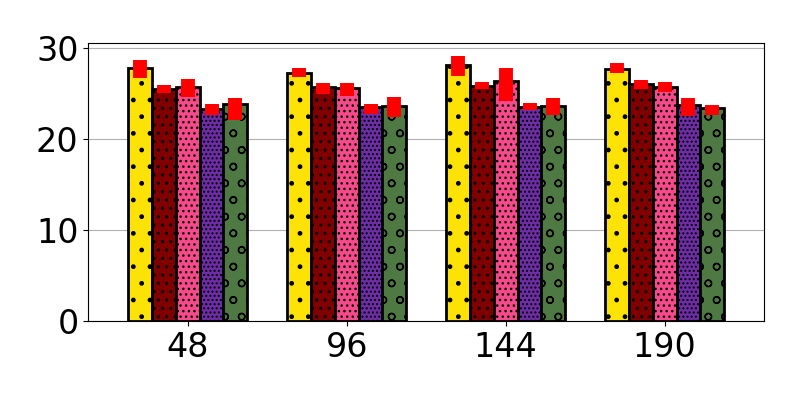}
\includegraphics[width=\legendwidth]{images/ppopp21/jax_non_tm/bst_hc-legend.png}
\vspace{-4mm}
\caption{\textbf{Factor analysis} for \textbf{Handcrafted} \textbf{Unbalanced BSTs} on the \textbf{Intel} system. \textbf{Note varying y-axes} (x-axis = \# threads).}
\label{fig:jax-non-tm-bst-factor-analysis}
\end{figure*}

\begin{figure*}
\newcommand{\plotwidth}{0.32\linewidth}
\newcommand{\legendwidth}{0.5\linewidth}
\begin{tabular}{p{\plotwidth}p{\plotwidth}p{\plotwidth}}
\centering\noindent\textbf{1\% updates, 1M keys} & \centering\noindent\textbf{10\% updates, 1M keys} & \centering\noindent\textbf{100\% updates, 1M keys}
\end{tabular}
\rotatebox{90}{\hspace{1mm}\textbf{ops/microsecond}}
\includegraphics[width=\plotwidth]{images/ppopp21/jax_non_tm/avl_hc_total_throughput-u0.5_0.5-k2000000.png}
\includegraphics[width=\plotwidth]{images/ppopp21/jax_non_tm/avl_hc_total_throughput-u5.0_5.0-k2000000.png}
\includegraphics[width=\plotwidth]{images/ppopp21/jax_non_tm/avl_hc_total_throughput-u50.0_50.0-k2000000.png}
\rotatebox{90}{\hspace{3mm}\textbf{L2 miss/op}}
\includegraphics[width=\plotwidth]{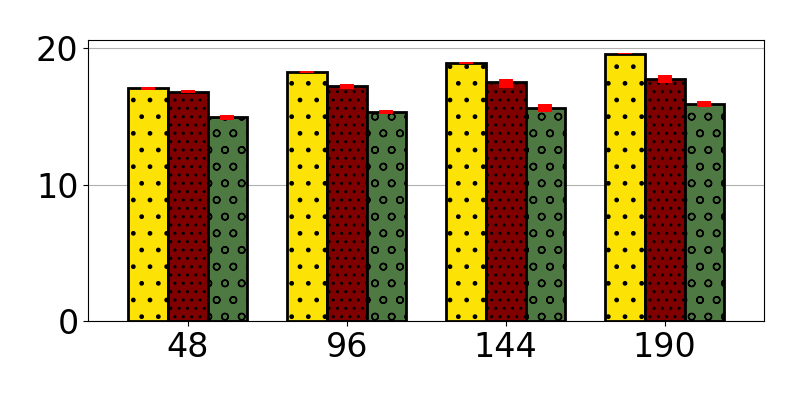}
\includegraphics[width=\plotwidth]{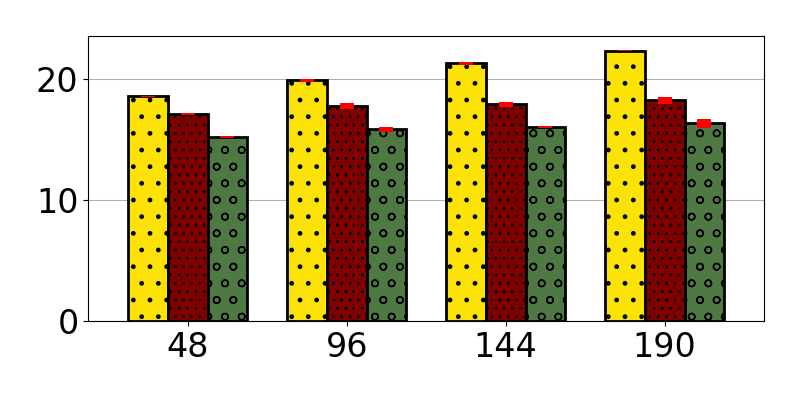}
\includegraphics[width=\plotwidth]{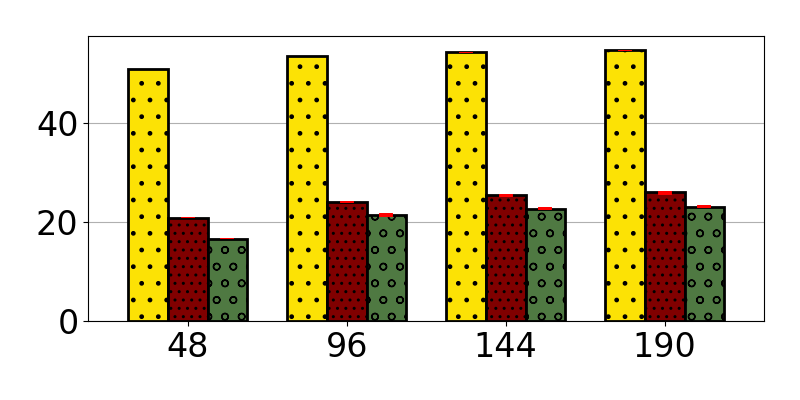}
\rotatebox{90}{\hspace{3mm}\textbf{L3 miss/op}}
\includegraphics[width=\plotwidth]{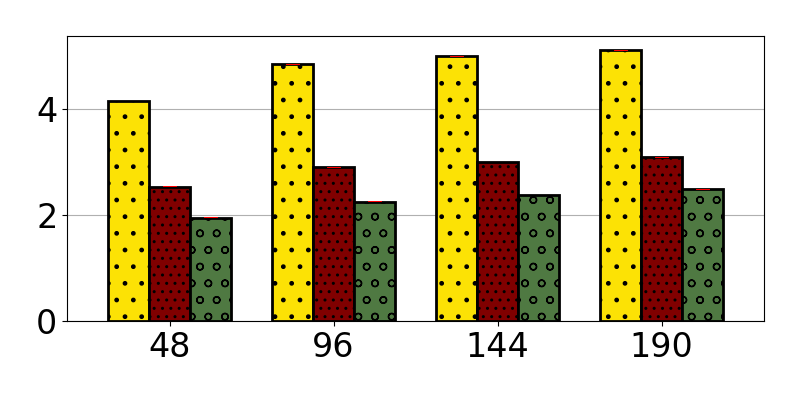}
\includegraphics[width=\plotwidth]{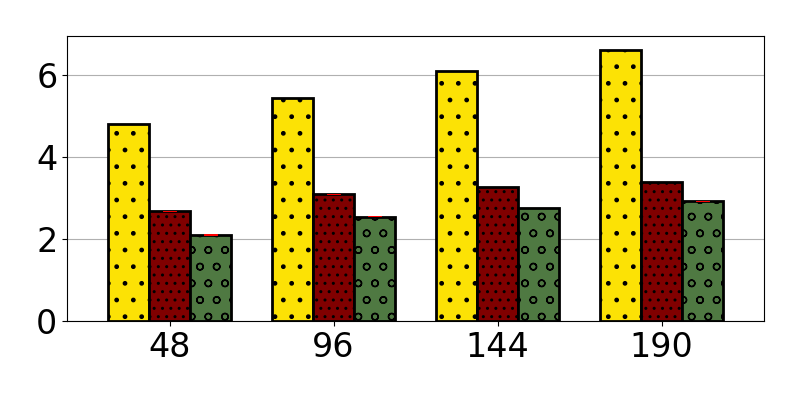}
\includegraphics[width=\plotwidth]{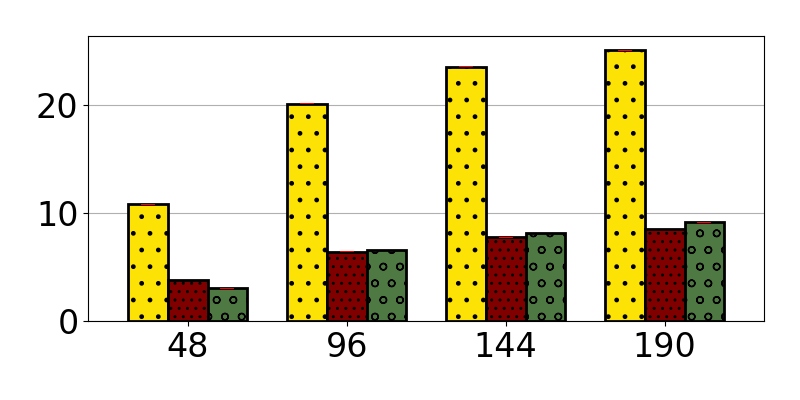}
\rotatebox{90}{\hspace{3mm}\textbf{cycles/op}}
\includegraphics[width=\plotwidth]{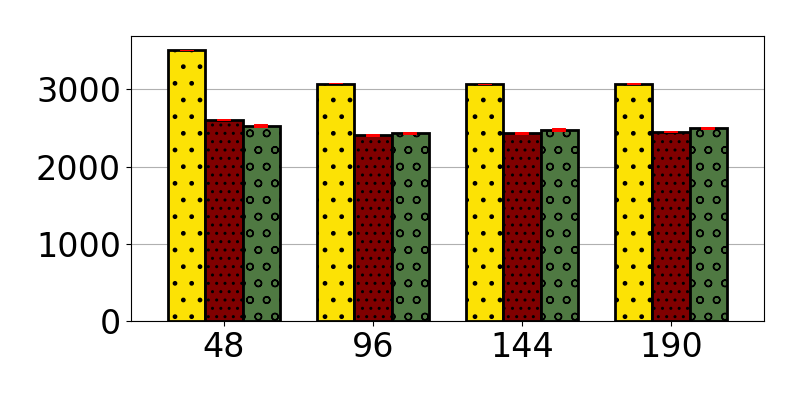}
\includegraphics[width=\plotwidth]{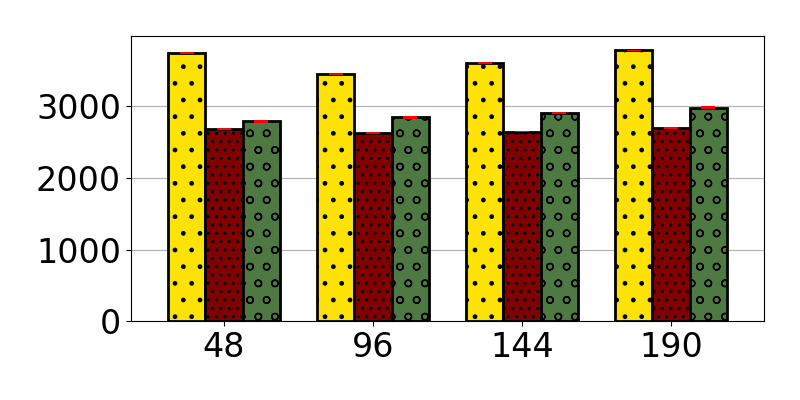}
\includegraphics[width=\plotwidth]{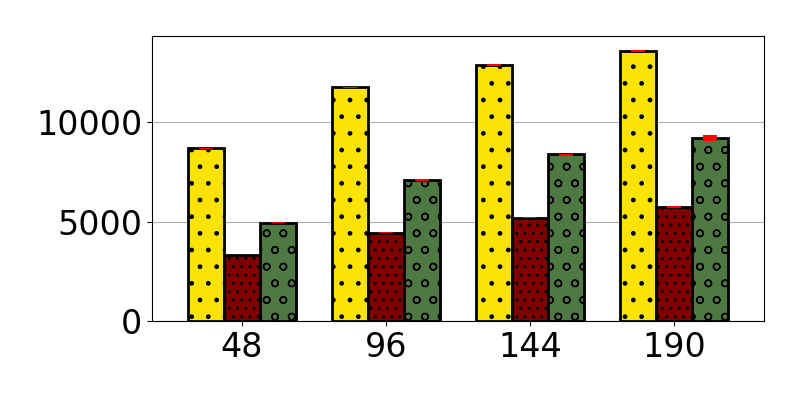}
\rotatebox{90}{\hspace{1mm}\textbf{instructions/op}}
\includegraphics[width=\plotwidth]{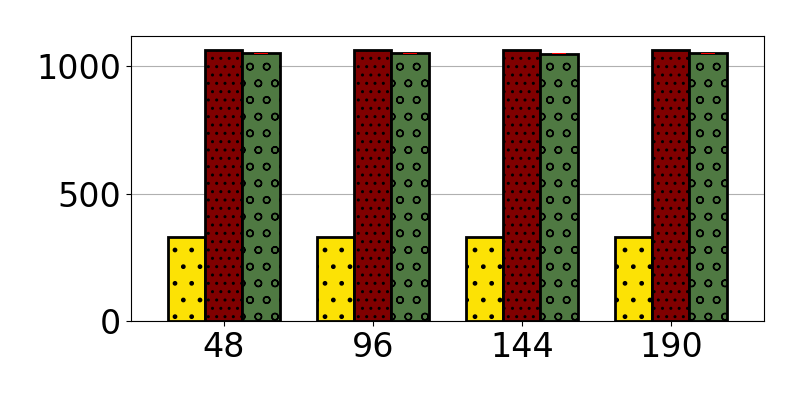}
\includegraphics[width=\plotwidth]{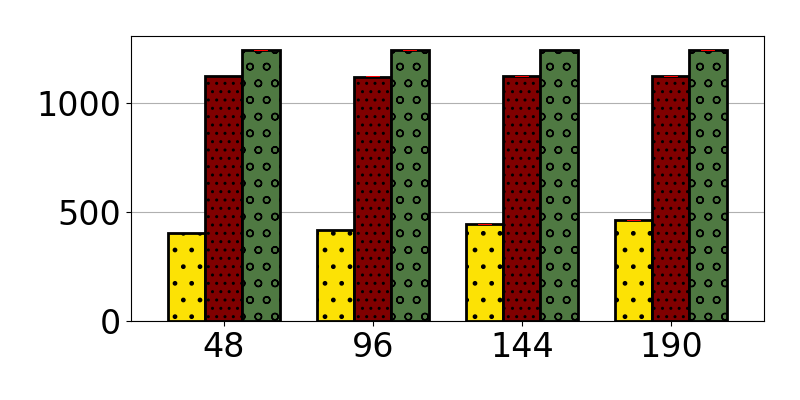}
\includegraphics[width=\plotwidth]{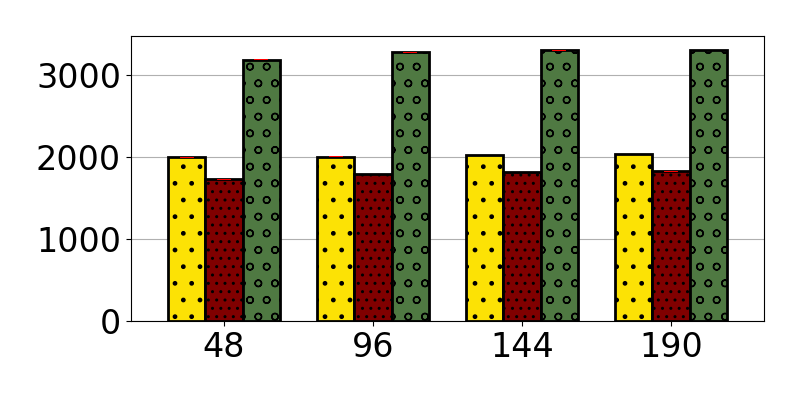}
\rotatebox{90}{\hspace{2mm}\textbf{page faults/op}}
\includegraphics[width=\plotwidth]{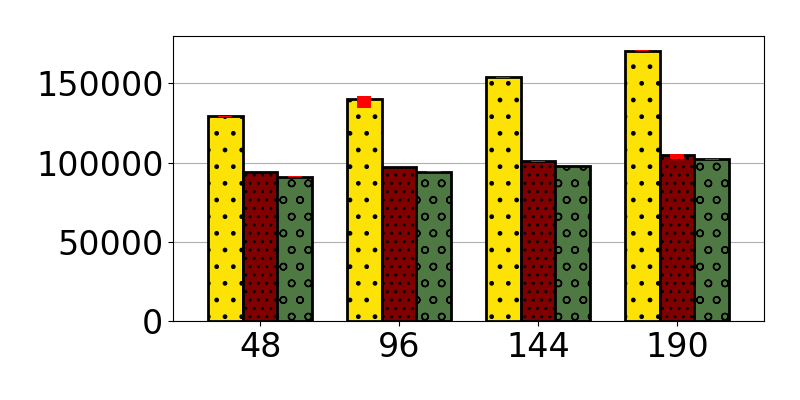}
\includegraphics[width=\plotwidth]{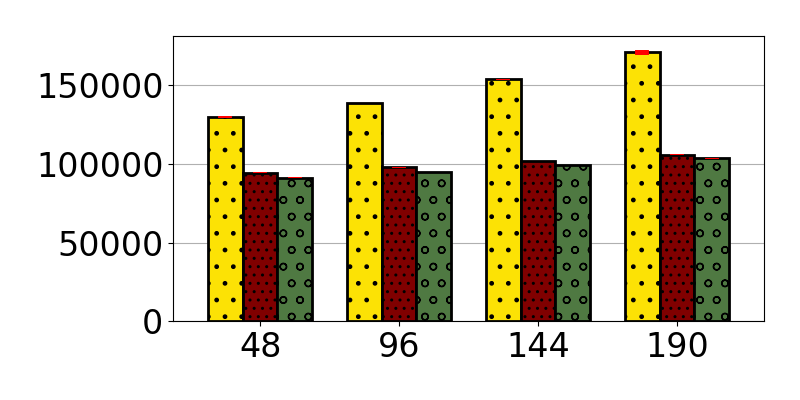}
\includegraphics[width=\plotwidth]{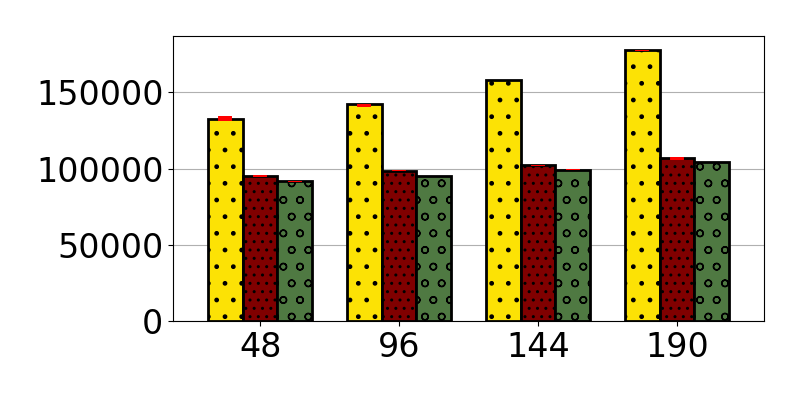}
\rotatebox{90}{\hspace{2mm}\textbf{avg key depth}}
\includegraphics[width=\plotwidth]{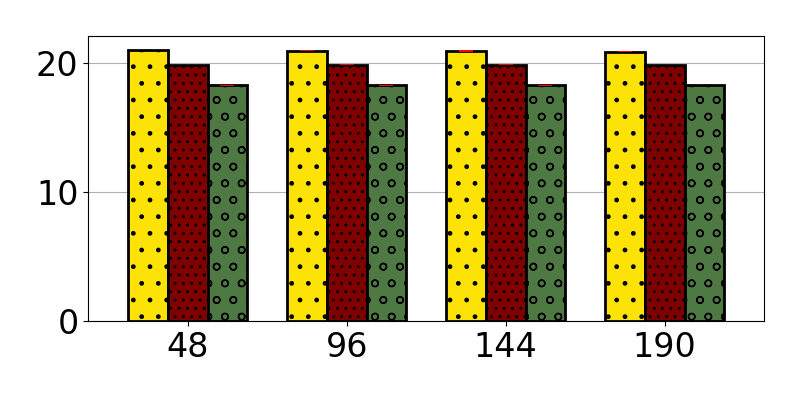}
\includegraphics[width=\plotwidth]{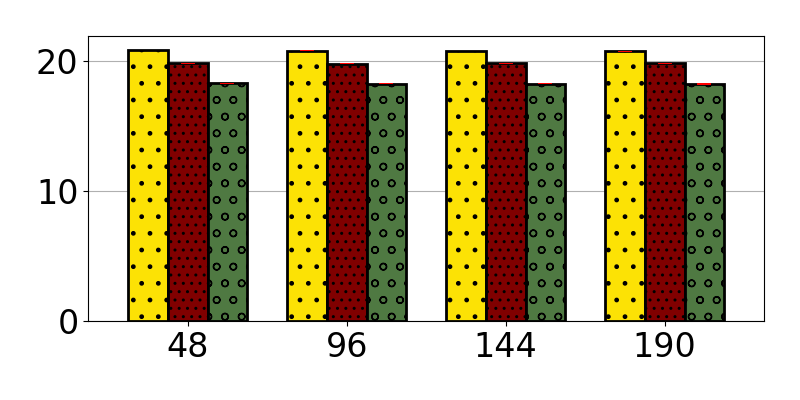}
\includegraphics[width=\plotwidth]{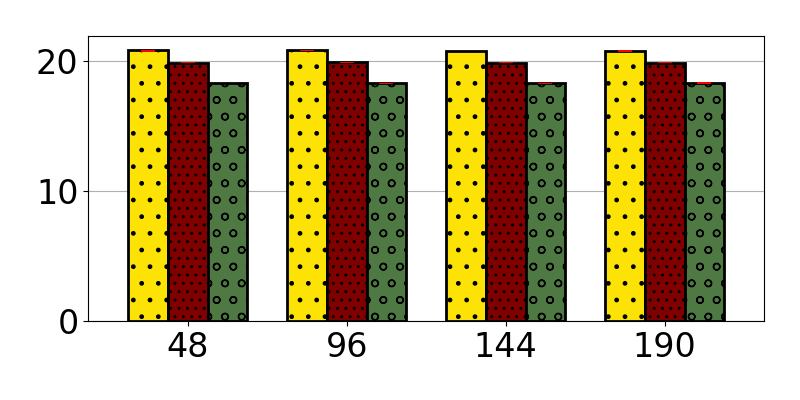}
\includegraphics[width=\legendwidth]{images/ppopp21/jax_non_tm/avl_hc-legend.png}
\vspace{-4mm}
\caption{\textbf{Factor analysis} for \textbf{Handcrafted} \textbf{Balanced BSTs} on the \textbf{Intel} system. \textbf{Note varying y-axes} (x-axis = \# threads).}
\label{fig:jax-non-tm-avl-factor-analysis}
\end{figure*}

 \section{Dynamic Connectivity Via PathCAS} \label{a:cut}
 As an example of how our approach can be applied to a different data structure, in this section we give a brief overview of the first lock-free concurrent solution to the dynamic connectivity problem on undirected acyclic graphs.

 \noindent\textbf{Overview.}
 The dynamic connectivity problem involves maintaining a graph containing a set of fixed vertices and a dynamic set of edges.
 Solving dynamic connectivity requires implementing three operations: 
 \textit{connected(v, w)}, \textit{link(v, w)} and \textit{cut(v, w)}.
 \textit{connected(v, w)} returns true if there exists a path from node \textit{v} to \textit{w}.
 Otherwise, it returns false. 
 If there is no path between $v$ and $w$, \textit{link(v, w)} creates an edge between them and returns true.
 Otherwise, it returns false. Note that you cannot link two nodes if there exists a path between them, as this would create a cycle in the graph. This is a common limitation of \emph{sequential} data structures for dynamic connectivity, which does not relate to our method.
 Finally, if there is an edge between $v$ and $w$, \textit{cut(v, w)} removes it and returns true.
 Otherwise, it returns false. 

 We follow the same general approach that we used to implement the AVL tree: all nodes have version numbers, and \textit{validatePath} is used to implement atomic searches.

 In the sequential setting, Euler Tours \cite{Euler-Tours} are typically used to implement dynamic connectivity.
 An Euler Tour starts at an arbitrary node and visits each edge exactly once (interpreting undirected edges as two directed edges for our purposes) recording each visit to a node as they are traversed.
  \begin{figure}[H]
      \includegraphics[width=8cm]{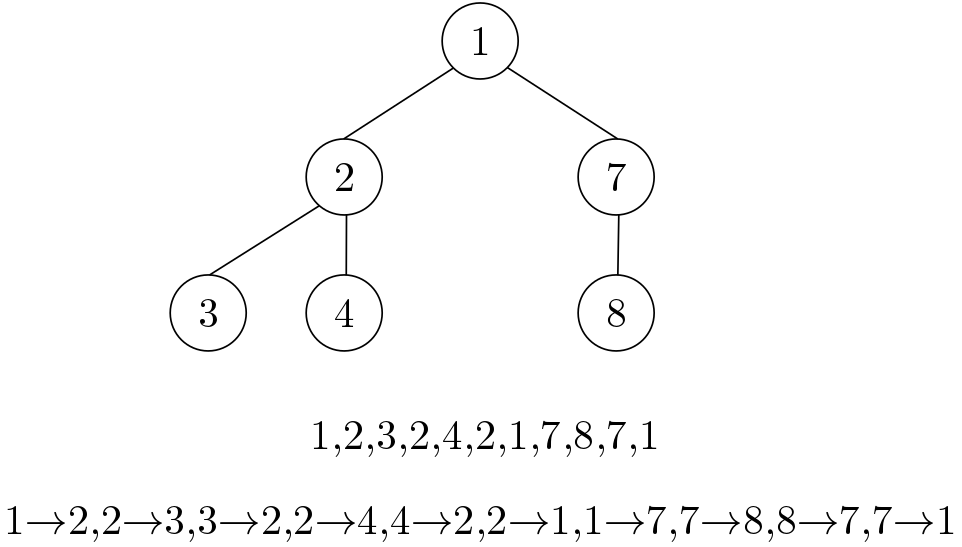}
  \centering

  \caption{Example of an Euler tour of a Tree.} 
  \label{fig:tour_ex}
  \end{figure}


 In the classical Euler tour data structure, Euler tours are stored in a BST.
 We chose to store the tour in a skip list rather than a BST, as in \cite{batch-skip}, which makes it easier to split and merge while maintaining (probabilistic) balance, and having each node of the skiplist represent 
 an \textit{edge} in an Euler tour rather than 
 a node \cite{edge-tours}.
 To clarify, there is a graph comprised of graph nodes, and a skiplist comprised of list nodes (which represent \textit{edges} in the graph), and each graph node has pointers to the list nodes for each of its incident edges.
 That paper also adds an additional self-edge for each node (which appears as a list node, pointed to by the appropriate graph node).
 It turns out that this self edge greatly simplifies the data structure operations in a concurrent setting.
 Figure \ref{fig:ll_link} shows an example of how such tours can be represented.
 We omit the upper levels of the skip-list to save space, and draw the graph representation above the list.
 To avoid special cases, towers of sentinel nodes are added to be beginning and the end of every tour list.

 Just as we used version numbers to impose a sequential ordering on modifications to a single BST node in our AVL tree, we use the version number of the leftmost sentinel node at the bottom level of the list (the \textbf{minimum sentinel}) to impose a sequential ordering on modifications to an individual \textit{Euler tour list} (i.e., a single version number protects the \textit{entire} tour list). 
 More precisely, all updates increment the version number of the minimum sentinel, allowing only a single update on any list at a time.
 This might seem like a concurrency bottleneck, but care must be taken to avoid the possibility of cycles being introduced by concurrent \textit{link}s.
 Additionally, every graph node is initially in \textit{its own} Euler tour tree, allowing plenty of concurrency.
 We now sketch how the operations are implemented.


 \noindent\textbf{Connection Queries.}
 The main purpose of this data structure is to answer connectivity queries: does a path between nodes \textit{v} and \textit{w} exist (or, equivalently, do \textit{v} and \textit{w} belong to the same connected component, or tour list)? Consider the self-edges from $v$ to $v$, and from $w$ to $w$.
 Let $L_v$ and $L_w$ be the list nodes that represent these self-edges.
 If a time can be established when $L_v$ and $L_w$ were in the \textit{same} tour list, then a path exists between those two graph nodes at that time.
 Conversely, if a time can be established where these list nodes were in \textbf{different} tour lists, then no path exists between the graph nodes at that time.


 This is simple to determine: starting from $L_v$ and $L_w$, the tour list(s) are traversed left towards the minimum sentinel.
 These traversals do the reverse of a traditional skiplist search, traversing up and left in the list until a sentinel node is reached. 
 Once a sentinel node is reached, the node is traversed down towards the bottom level of the list to locate the minimum sentinel.
 (We could simply have traversed left, without going upwards, but then our traversal time could be linear in the size of the tour list.)
 The paths taken by both traversals are then validated, and if either validation fails the entire operation is retried.
 If the \textit{same} minimum sentinel is found by these validated searches then there was a time when these two graph nodes existed in the same subgraph and therefore a path existed between them (so true is returned).
 If they are different, a time exists where the graph nodes were in different subgraphs and no path existed between them (so false is returned).


 \noindent\textbf{Link.}
 To simplify the presentation, we first assume each tour list is implemented as a doubly-linked list (rather than skiplist), then explain how this changes when skiplists are used.
 The goal of our implementation of \textit{link(v,w)} is to add a link from graph node \textit{v} to graph node \textit{w}, absorbing the subgraph (equivalently, tour list) of \textit{w} into the subgraph of \textit{v}.
 \textit{v} and \textit{w} can only be linked if they are not part of the same subgraph, so we start by performing the algorithm for \textit{connected(v,w)}.
 Let $M_v$ and $M_w$ be the minimum sentinels located while performing this algorithm. 
 If $M_v$ and $M_w$ are the same, then \textit{link(v,w)} can safely return false: a time was determined where they were already in the same subgraph.
 If $M_v$ and $M_w$ are different, then the operation can proceed. 
 We will include $M_v$ and $M_w$ in our (eventual) KCAS operation, using it to increment both of their version numbers. 

 We explain the next steps with an example.
 Consider the case \textit{link(3, 6)} presented in Figure \ref{fig:ll_link}. 
 The tour lists drawn at the top of that figure are logically split into two \textit{sublists} each: sublist \textit{L1} contains all nodes in the tour list to the left of and including 3's self-edge (excluding the minimum sentinel, which we call S1), and sublist \textit{L2} contains all nodes to the right of 3's self-edge (excluding the \textbf{maximum sentinel}, S2).
 In the tour list for node 6, S3, L3, L4, and S4 are similar to S1, L1, L2, and S2, respectively.
 We suffix the labels L1, L2, L3 and L4 with \textit{A} and \textit{B} to denote the beginning and end of sublists (i.e., L1A is the leftmost node of L1, and L1B is the rightmost node of L1).
 These nodes will require updates. 

 Since link adds a new edge, we should add that edge to the tour lists (twice, as it should be traversed in both directions).
 Two new list nodes are created (VW and WV), one for each direction.
 The resulting tour list can be constructed by arranging the sublists in the following order: [S1, L2, L1, VW, L4, L3, WV, S4], which requires the operation to change the \textit{left} or \textit{right} pointers of the nodes on the ends of the sublists, as well as those of sentinel nodes. Note that it is possible for L2 and L4 to be empty, in which case, the same sequence sublist order works if empty sublists are omitted.
 This arrangement effectively \textit{rotates} the individual tours containing $v$ and $w$ such that they are rooted at $v$ and $w$, respectively, and then links them together.
 Note that we also update the graph nodes for 3 and 6 to add their new neighbour (6 and 3, resp.) to their adjacency lists.

 All of these pointer changes are performed in a single KCAS.
 In other words, the KCAS needs to update the left and right fields of all the list nodes at the ends of the sublists, add neighbours to the graph nodes, increment the version numbers of all nodes involved (crucially, including the minimum sentinel), and \textit{mark} any nodes that are removed (S2 and S3, in this case).
 We use \textit{marking} to avoid erroneous modifications to deleted nodes. Before a KCAS is performed, we first verify that every node included in the KCAS is not marked.
 If a node is marked, we restart the entire operation.

 \textit{In a skiplist}, this list restructuring is simply repeated at \textit{every level}, in one large KCAS.
 The relevant sublists are determined at \textit{each level} by traversing starting from the bottom list, and are rearranged in the same order as the bottom list. Crucially, updates to a skiplist based tour list are still serialized on the same field: the version of the minimum sentinel.

 To determine a sublist at level \textit{i} + 1 from level \textit{i}, we traverse \textit{upwards} and \textit{inwards} from the ends of the sublist at level \textit{i}. 
 For example, consider the top left image in Figure \ref{fig:ll_link}, we will call the bottom level of the list represented here level 1. 
 If we wished to determine the sublist L1 at level 2, we would traverse from L1A right until a node is encountered that has a node above it at level 2.
 Similarly, to determine the other end of the sublist, we would traverse from L1B \textit{left} until a node is encountered that has a node above it at level 2. The two nodes found at level 2 are the ends of the sublist L1 at level 2. If these two traversals ever encounter the same node at some level \textit{i}, this indicates there is no such sublist at level \textit{i} + 1. By performing this traversal for every sublist, at every level, all the nodes that need to be modified can be found. This process is repeated until the maximum height of the skiplist is reached, or the sublist does not exist at some level. 

 The goal of our implementation of \textit{cut(v, w)} is to remove the edge connecting \textit{v} to \textit{w} if it exists, and split their tour list into two.
 Graph nodes contain adjacency lists, so determining if two nodes are directly connected by an edge is easy.
 If they are not neighbours, then the operation returns false.
 Otherwise, the version numbers of these graph nodes should be added to our (eventually) KCAS. 
 The minimum sentinel is located as in the previous operations (but we only need to traverse starting at one of $v$ or $w$). 

 From the graph nodes we can find the list nodes representing the edges \textit{VW} and \textit{WV}.
 These list nodes will be removed as part of the operation and the list nodes between them will form one of the new tour lists.
 The list is separated into \textit{three} sublists: L1 (which contains all nodes to the right of the minimum sentinel and to the left of the edge VW), L2 (which contains all nodes to the right of VW and to the left of WV), and L3 (which contains all nodes to the right of WV and to the left of the maximum sentinel (S2).
 Two new sentinel nodes are created for the new list, S3 and S4.
 This operation simply removes L2 from the center of the list, creating two lists as a result: [S1, L1, L3, S2] and [S3, L2, S4].\footnote{One of L1 or L3 (but not both) could be empty, but the order presented here remains correct if empty lists are omitted.} The sublist L2 represents the nodes no longer reachable from \textit{v} after the removal of \textit{w}, since the only way \textit{v} could reach these nodes was by first traversing \textit{w}. 

 \textit{To extend this to a skiplist} it is very similar to \textit{link}.
 These sublists are formed at each level and linked together in the same order as the bottom list. 

 \begin{figure}[t]
     \includegraphics[width=.45\textwidth]{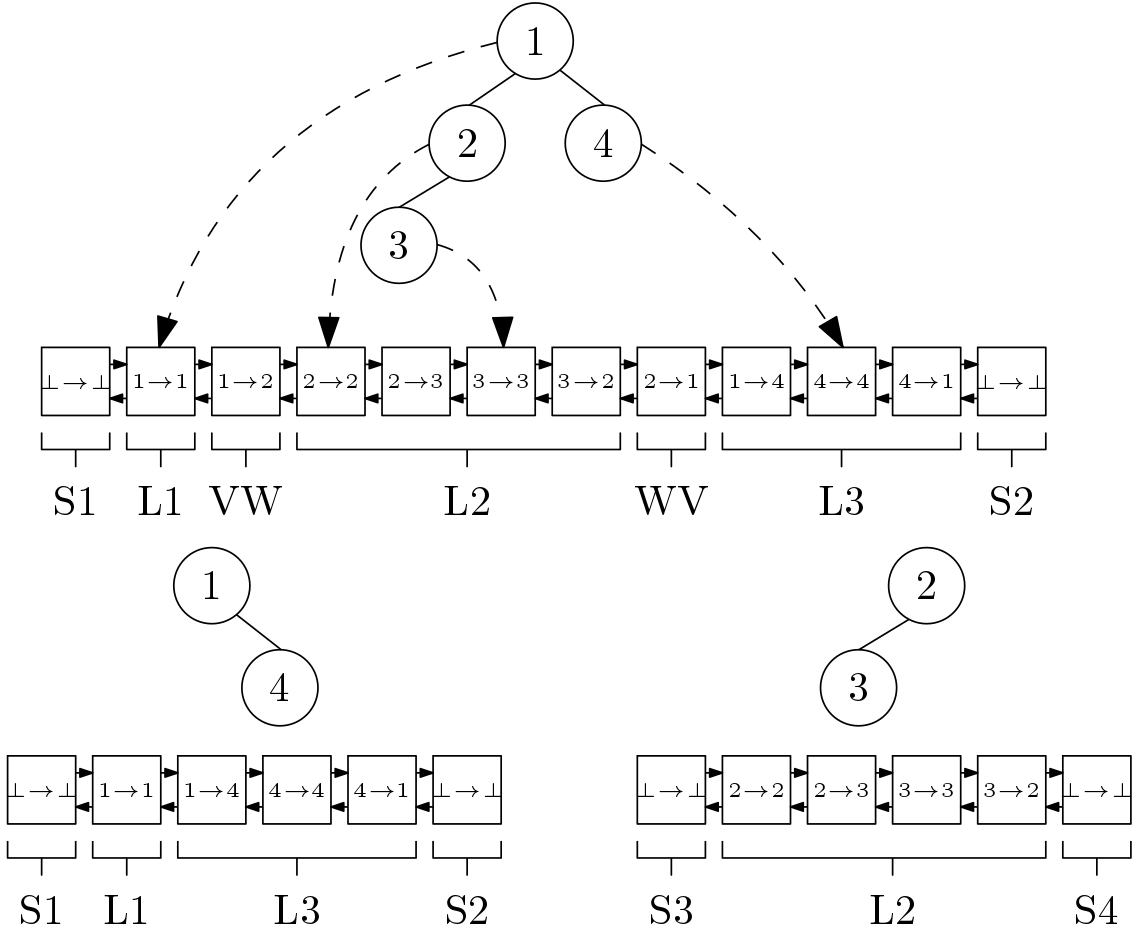}
 \centering

 \caption{Operation \textit{cut}$(3,6)$ on (simplified) Euler tour lists}
 \label{fig:ll_link}
 \end{figure}
 \section{Dynamic Connectivity Correctness Sketch} \label{a:dynamic-proof}

Recall that to avoid special cases, we take the original graph and add a self-edge for each graph node. Once this is done, there is still a well defined Euler tour that follows each edge once, but now these self-edges appear in the Euler tour. if we consider any Euler tour containing self-edges, and delete all self-edges from that tour, we obtain an Euler tour of the original graph.

 As in the AVL tree we mark list nodes at the time they are deleted, so we can verify (purely syntactically) before performing a KCAS that all nodes it will modify are unmarked, and in doing so guarantee that no deleted node is ever changed.
 \begin{defn} \label{def:fully-dynamic}
 Our \textit{fully-dynamic connectivity} data structure consists of a set of Euler tour skiplists and a set of graph nodes. each graph node $u$ participates in a single Euler tour skiplist (or tour list for short), and contains pointers to all of the (skip)list nodes that represent directed edges starting from $u$ (including the self edge $u\rightarrow u$). in each tour list, the bottom level list nodes represent the seq of edges visited in an \textbf{Euler tour} of the graph nodes that participate in the tour list.
 \end{defn}{}

 \begin{algorithm}[h]
 \caption{\textit{isConnected(v, w)}}\label{dconn:isconn}
 \begin{algorithmic}[1]
 \While{\textit{true}}
     \State \textit{$p_v$ = traversePathToMinSentinel(v)} \Comment A path $p_v$ to a minimum sentinel is followed from some the self edge of some graph node $v$, 
     \State \textit{$p_w$ = traversePathToMinSentinel(w)} \Comment A path $p_w$ to a minimum sentinel is followed from some the self edge of some \textbf{other} graph node $w$ 
     \IIf{\textbf{not} \textit{validatePath($p_v$)}} \textbf{continue} 
     \IIf{\textbf{not} \textit{validatePath($p_w$)}} \textbf{continue} 
     \IIf{$p_v$.minSent == $p_w$.minSent} \textbf{return} \textit{true} \Comment If both paths ended up at the same minimum sentinel, return \textit{true}
     \ElseI\ \textbf{return} \textit{false} \Comment If both paths ended up at different minimum sentinels, return \textit{false}
 \EndWhile
 \end{algorithmic}
 \end{algorithm}

 \begin{obs} \label{obs:min-sent}
 Any list node that has the value NIL in both its down and left field is the \textbf{minimum sentinel} of a tour list. 

 \end{obs}

 \begin{lemma}
 Our implementation 
 satisfies the following claims:
 \begin{enumerate}
 \item \textit{isConnected(v, w)} returns the same value it would if it were performed atomically at its linearization point (just before the first validation)
 \item \begin{enumerate}
     \item The data structure is a fully-dynamic connectivity structure (see definition \ref{def:fully-dynamic})
     \item Any \textit{link} or \textit{cut} operation that performs a successful KCAS returns the same value it would if it were performed atomically at its linearization point (the KCAS)
     \item Any \textit{link} or \textit{cut} operation that terminates without performing a successful KCAS returns the same value it would if it were performed atomically at its linearization point 
 \end{enumerate}
 \end{enumerate}
 \end{lemma}{}

 \begin{proof}
 Consider an arbitrary execution $E$.
 We prove these claims together by induction on the sequence of steps $s_1, s_2, ...$ (which can be shared memory reads, atomic KCASRead operations, or atomic KCAS operations) in $E$.

 Base case: There are a finite number \textit{i} graph nodes, and each is in its own tour list. These tour lists contain a single self-edge, and two sentinel towers on each side. 

 Inductive step: suppose the claims all hold before step $s$. We prove they hold after step $s$.

 \noindent \textsc{Claim 1}. The only operations that can impact this claim are KCAS operations from \textit{link} or \textit{cut}. Reads do not change the data structure, hence they will not change the paths followed by the traversals in \textit{isConnected}. 

 \textit{Subcase 1}: Consider an invocation of \textit{isConnected(v, w)} that returns true. In the final loop of this invocation, the following occurs: the traversal from \textit{v} to a minimum sentinel which follows path $p_v$ occurs, then the traversal from \textit{w} to the same minimum sentinel which follows path $p_w$ occurs, let the time this second traversal ends be $t_0$. (From Observation \ref{obs:min-sent} we can statically check that node reached by these traversals was, in fact, a minimum sentinel.) We then validate the path of the first traversal at $t_1$, and then validate the path of the second traversal at $t_2$ (therefore, $t_0<t_1<t_2$). This operation does not return unless \textit{validatePath} returns true for both paths. Since we know that \textit{validatePath($p_w$)} returns true, there were no modifications to any nodes in $p_w$ between $t_0$ and $t_2$. Hence, there were also no modifications to any node in $p_w$ between $t_0$ and $t_1$, and \textit{validatePath($p_w$)} would still have returned true if it were executed at $t_1$.  Consider a \textit{link} or \textit{cut} update that changes the configuration of the tour list during this operation. If this \textit{link} or \textit{cut} does not involve the current tour list, it does not change the minimum sentinel that would be reached by either traversal. If these updates were to occur on the current tour list, it must include the version number of the minimum sentinel in the KCAS. If this update occurs during one of the traversals, then the traversal will fail to validate, and this operation will be retried. Additionally, if the update occurs between the traversals and one of the validations, the validation will fail. Therefore, since both of these traversals end at the same minimum sentinel, they were in the same tour list just before $t_1$ (which is where we linearize this operation).

 \textit{Subcase 2}: Consider an invocation of \textit{isConnected(v, w)} that returns false. This argument is the same as Subcase 1, since \textit{isConnected(v, w)} still must validated both paths, however the minimum sentinels are different. 

 \noindent \textsc{Claim 2a}. The only operations that can impact this claim are the KCAS operations in \textit{link} an \textit{cut}, as reads to not change the data structure.

 \textit{Subcase 1}: Suppose \textit{s} is a successful KCAS of \textit{link(v, w)}. Before this KCAS, two traversals occurred and were validated that from the self-egdes of the two graph nodes \textit{v} and \textit{w} to two different minimum sentinels. As we proved in Claim 1, a time \textit{t} exists where \textit{v} and \textit{w} were in different tours (and hence there was no path between them). From the inductive hypothesis, these two tour lists were well-formed before this operation. Hence, it is correct to perform a \textit{link} operation on these two nodes at \textit{t}. Since the version number of both minimum sentinels are part of this KCAS, there are no changes to either tour list between \textit{t} and \textit{s}. This means that no update has modified any node in either tour list after the time they were validated. Therefore, the \textit{link} operation is still applicable at \textit{s}.

 \textit{Subcase 2}: Suppose \textit{s} is a successful KCAS of \textit{cut(v, w)}. This is simply an easier case than Subcase 1, as only a single list is tracked for this operation. 

 \noindent \textsc{Claim 2b}. This is proven in Claim 2A, as we proved that both \textit{link} and \textit{cut} are atomic at \textit{s}, which is when the KCAS is executed. 

 \noindent \textsc{Claim 2c}. This is proven in Claim 1, as both use the result of \textit{isConnected} to determine if a KCAS should be executed or not. If \textit{link} calls \textit{isConnected} and it returns true, \textit{link} will return false, as a time \textit{t} was established where a path already existed between the two nodes (right before the first validation of \textit{isConnected}), we linearize this operation at \textit{t}. The argument is identical for \textit{cut}, but reversed.

 \end{proof}{}

\section{Dynamic Connectivity Progress Sketch}

\begin{theorem}
 Our implementation of dynamic connectivity is lock-free.
 \end{theorem}{}
 \begin{proof}
 Consider some configuration \textit{C} where threads continue to take steps, but after some time \textit{t} no operations complete. All threads, therefore, must be stuck in retry loops, failing \textit{validatePath} or KCAS operations. Since \textit{validatePath} and KCAS operations only fail if a node has been modified since its version number was last read, the only way these infinitely many KCAS operations and \textit{validatePath} fail is if there are infinitely many modifications to the data structure. However, if operations stop completing after time \textit{t}, then eventually the data structure must stop changing, since each operation can perform at most one successful KCAS, and the data structure is changed exclusively by successful KCAS operations. This is a contradiction, the only way a thread fails an operation is if another thread has made progress. 

\end{proof}{}

\end{document}